\documentclass[12pt,a4paper]{article}

\usepackage{graphicx}
\usepackage{xcolor}
\usepackage{geometry}
\usepackage{amsmath}
\usepackage{amssymb}
\usepackage{booktabs}
\usepackage{longtable}
\usepackage{array}
\usepackage{enumitem}
\usepackage{setspace}
\usepackage{microtype}
\usepackage{fancyhdr}
\usepackage{titlesec}
\usepackage{multirow}
\usepackage{amsthm}

\usepackage{mathtools}
\usepackage{bm}
\usepackage{thmtools}

\usepackage{float}
\usepackage[format=hang, font={small, it}, labelfont=bf]{caption}
\usepackage{algorithm}
\usepackage{algorithmic}
\usepackage{makecell}
\usepackage{ragged2e}
\usepackage{array}
\usepackage{booktabs}

\usepackage{mathrsfs}
\usepackage{xfrac}
\usepackage[overload]{empheq}
\usepackage{authblk}
\usepackage{tabularx}

\usepackage[square,numbers, sort&compress]{natbib}
\newcommand{\dd}{\mathrm{d}}
\newcommand{\DD}{\mathcal{D}}
\newcommand{\EE}{\mathcal{E}}
\newcommand{\Mf}{\widetilde{M}}

\newcommand{\rank}{\operatorname{rank}}
\newcommand{\R}{\mathbb{R}}

\titleformat{\section}{\large\bfseries}{\thesection}{1em}{}
\titleformat{\subsection}{\normalsize\bfseries}{\thesubsection}{1em}{}
\titleformat{\subsubsection}{\normalsize\itshape}{\thesubsubsection}{1em}{}

\usepackage{hyperref}
\hypersetup{colorlinks=true, linkcolor=blue, citecolor=blue, urlcolor=blue,
 pdftitle={Action, Probability Measures, and Rate Function of Stochastic Systems},
 pdfsubject={Contact geometry of probability measures}}

\numberwithin{equation}{section}
\allowdisplaybreaks[4]

\declaretheoremstyle[
 headfont=\bfseries,
 bodyfont=\normalfont,
 postheadspace=1em,
 numbered=yes
]{thmstyle}
\declaretheorem[style=thmstyle, numberwithin=section]{theorem}
\declaretheorem[style=thmstyle, sibling=theorem]{lemma}

\declaretheorem[style=thmstyle, sibling=theorem]{definition}
\declaretheorem[style=thmstyle, sibling=theorem]{remark}
\declaretheorem[style=thmstyle, sibling=theorem]{corollary}

\begin{document}

\begin{titlepage}
\centering
\vspace*{2.5cm}

{\Large\bfseries{Action, Probability, and Rate Functions of Dissipative Stochastic Systems}
\\[1.4cm]}
\vspace{2cm}

{
\large
\textbf{D.Y. ZHONG}$^{*1,2}$
}

\vspace{.5cm}

{
\normalsize
$^{1}$State Key Laboratory of Hydroscience and Engineering, Tsinghua University, Beijing 100084, China\\
$^{2}$Department of Hydraulic Engineering, Tsinghua University, Beijing 100084, China
}

\vspace{0.5cm}

{
\normalsize
$^{*}:$ Corresponding author: zhongdy@tsinghua.edu.cn
}

\vfill

{
\large
October $27^{th}$ 2026
}

\end{titlepage}

\newpage
\section*{Abstract}

\addcontentsline{toc}{section}{Abstract}

\vspace{0.5cm}

A dissipative system distributes over its state space, and its rare events are measured on
an exponential scale by a rate function. This paper shows that one contact form carries
both. On the bundle of jets of probability densities over the state space, the constraint
function of the contact potential is a first integral of the characteristic field, so the
action is affine in the elapsed time along each characteristic and the contact probability
measure of a Lagrangian section has that action as its exponent at every finite time, with
no path integral, no variational problem over curves and no coercivity hypothesis. The rate
function is the deficit of the action below its supremum, and is therefore a transported
conserved quantity rather than the value of a limit; its amplitude is invariant and its
non-degeneracy is decided by the initial section alone. Large deviations form the
concentration regime of this finite-time structure and require one condition, the smooth
convergence of the endpoint map of the characteristic flow to a diffeomorphism. The drift
and the fluctuation tensor are the first two coefficients of the same potential, the second
confined to the iso-density surfaces by the degeneracy condition, so the fluctuation is
constrained by the invariant density rather than chosen beside the drift. The entropy
production is an observable of the same measure, and a fibre involution yields a
Gallavotti--Cohen-type duality between the forward and the reversed ensembles under a
separation condition on the tilted action.

\vspace{0.5cm}

\noindent\textbf{Keywords}: contact geometry, contact probability measure, large deviation theory, stochastic dynamics, Gallavotti--Cohen fluctuation theorem, Hamilton--Jacobi equation

\noindent

\newpage
\tableofcontents
\listoftheorems
\listoffigures

\newpage

\section{Introduction}\label{sec:intro}

A stochastic dissipative system is described by three objects: an action functional, a
probability measure on the state space, and a rate function. The action assigns a cost to
the evolution of the system, the probability measure describes how the state distributes
over the state space, and the rate function measures on an exponential scale the smallness
of the events that depart from that distribution \cite{dembo1998large}. For a conservative
system all three descend from one Hamiltonian: the action is its time integral, the measure
is the Gibbs weight of the energy, and the rate function of the equilibrium fluctuations is
a relative entropy with respect to that weight. For a dissipative system, in which
phase-space volume contracts and no conserved energy exists, the three objects are
constructed separately, and the question of this paper is whether one geometric object
carries all three.

The action functionals available for stochastic and dissipative systems are functionals on
path space. Herglotz's variational principle and its contact-geometric formulations extend
Hamilton's principle to dissipative equations of motion by letting the Lagrangian depend on
the action itself
\cite{herglotz1930beruehrungstransformationen,arnold1990contact,guenther1996herglotz,mrugala1990statistical,bravetti2017contact,deleon2019contact,bravetti2019contactGeometryThermo};
the resulting action evaluates a single curve in an extended space and generates no
probability. For a prescribed diffusion the Onsager--Machlup functional is the negative
logarithm of the probability that the path remains in a tube about a given curve
\cite{onsager1953fluctuations}, and in stochastic thermodynamics the action of a trajectory
is the logarithm of the ratio of the forward and reversed path probabilities of a prescribed
Langevin equation \cite{seifert2012stochastic,jarzynski1997nonequilibrium,crooks1999entropy}.
In the small-noise limit the Freidlin--Wentzell action is the rate function on path space,
and its minimum over the curves joining two states is the quasipotential
\cite{varadhan1966asymptotic,freidlin1998random}; after a hydrodynamic scaling the
macroscopic fluctuation theory supplies an analogous action for the joint fluctuations of
the density and the current \cite{bertini2015macroscopic}. In each of these constructions
the action is either the negative logarithm of the path probability of a process already
prescribed or the value of a variational problem in a limit; in none of them is it a
function on the state space.

The probability measure of the state is likewise downstream of a prescribed dynamics. From
a drift and a diffusion coefficient the law of the process is constructed as a path measure
or through the Fokker--Planck equation, and its invariant density is obtained as a
stationary solution or by an ergodic argument. In equilibrium the measure is postulated as
the Gibbs weight of a prescribed energy, which the Gibbs variational principle then
characterises \cite{ellis1985entropy}. Contact geometry supplies dissipative equations of
motion from a contact form and generalises Liouville's theorem to the resulting flow,
yielding invariant measures attached to the dynamics
\cite{bravetti2015liouville,bravetti2016thermostat}; the volume form of the ambient contact
manifold, however, is proportional to the constraint function and degenerates on its zero
set, and these invariant measures come with no rate function. A stochastic Lagrangian
representation writes the solution of a dissipative equation as an expectation over
stochastic characteristics driven by a prescribed Wiener process
\cite{constantin2008stochastic}. In every case the measure is derived from, or postulated
beside, a dynamics already chosen.

The rate function is asymptotic in all of these frameworks: a variational value over curves
in the small-noise limit \cite{varadhan1966asymptotic,freidlin1998random}, the value function
of a control problem \cite{dupuis1997weak}, the viscosity solution of a Hamilton--Jacobi
equation driven by the generator \cite{feng2006large}, a Donsker--Varadhan long-time limit
\cite{donsker1975asymptoticI,donsker1975asymptoticII,donsker1976asymptoticIII}, or a
macroscopic fluctuation theory object after hydrodynamic scaling \cite{bertini2015macroscopic}.
At a fixed finite time these constructions yield a path action and no rate function, and
converting one into the other requires a minimisation over curves with growth conditions,
superlinearity in the fibre variable, or hypoellipticity and controllability for a degenerate
diffusion \cite{hormander1967hypoelliptic}. The coefficients are inputs: nothing selects the
diffusion coefficient or ties it to the invariant density of the drift. The
fluctuation--dissipation theorem states such a tie only near equilibrium, as a relation
between response coefficients
\cite{callen1951irreversibility,kubo1966fluctuation}; it constrains neither the fluctuation of
a dissipative dynamics at a general time nor one far from equilibrium.

The three objects are therefore produced by three constructions running in one direction: a
dynamics is prescribed, a measure is derived from it, and a rate function is extracted from
the measure in a limit. In each framework at least one of the three remains an input,
namely the coefficients of the dynamics, the energy of the Gibbs weight, or the generator
of the Markov process. What is missing is a theorem in which one geometric object carries
the action, the probability measure of the state and the rate function at finite time, the
coefficients of the dynamics being read off that object rather than supplied beside it.

The object of this paper is a contact form on the bundle of jets of probability densities
over the state space, so that the fibre coordinate is assembled from the derivatives of a
density and the bundle is the geometry of distributions rather than of particles
(Section~\ref{sec:found}). The consequences run in one direction, and no step of the chain
is a limit. The contact form supplies a potential and its constraint function; the
constraint function is a first integral of the characteristic field; the action, its
integral along a characteristic, is therefore affine in the elapsed time; the contact
probability measure of a Lagrangian section $\phi=\nabla S$ has the exponential $e^{nS}$ of
that action as its density on the base; and the rate function is the deficit of the action
below its supremum. Here $y_{0}(t,\cdot)$ is the inverse of the endpoint map of the
characteristic flow. The probability measure thereby takes a Gibbs form with formal
inverse-temperature parameter $nt$, and the large deviation principle is its
concentration regime as $nt$ grows, in the sense in which the Gibbs variational principle
underlies large deviations
\cite{ellis1985entropy}.

The potential defining the form is itself determined rather than chosen. A master equation
fixes it order by order in the fibre variable (Section~\ref{sec:H}). Its zeroth coefficient
is the entropy potential $H^{(0)}=\gamma_{0}\ln(\rho/\rho_{0})$ built from the invariant
density $\rho$ of the dissipative dynamics, its first is the drift $v=H^{(1)}$, and its
second is the fluctuation tensor $A=H^{(2)}$, which the degeneracy condition
$A\nabla H^{(0)}=0$ confines to the tangent space of the iso-density surfaces and a
transport equation then evolves. The fluctuation is consequently not a second datum placed
beside the drift; it is constrained by an invariant density that the drift already
determines, at any time and arbitrarily far from equilibrium.

Two results enter as inputs and are not reproved. The least constraint theorem of
\cite{zhong2025least} supplies the contact dynamics on such a bundle together with the
order-by-order construction of the potential, and the invariant measure theorem of
\cite{zhong2026cgle} supplies the measure carried by the contact distribution along a
characteristic. Neither identifies the action as a function of the endpoint, neither
produces a rate function, and neither addresses the entropy production or the duality. The
contributions of this paper are the following.

\begin{enumerate}[label=(\arabic*),leftmargin=2.2em]
\item {The action is a function on the state space, affine in the elapsed time.}
Its slope is the constraint function transported from the initial section,
$S_{t}=S_{0}\circ y_{0}+t\,\varepsilon_{0}\circ y_{0}$, so the action is available at
every finite time without any path integral and without any variational value
(Section~\ref{sec:main}).
\item {The probability measure is explicit at every finite time.} On a Lagrangian
section the contact probability measure has the exponential of that action as its density
on the base, $\mathcal{P}_{t}=e^{nS_{t}}/\mathcal{Z}_{t}$, in Gibbs form with formal
inverse-temperature parameter $nt$ (Section~\ref{sec:main}). No path integral is formed, no variational
problem over curves is solved, and no asymptotic hypothesis is used.
\item {The rate function is a transported conserved quantity.} It is the deficit of
the action below its supremum, $I_{t}=n[\sup_{E}S_{t}-S_{t}]$, whose dominant part is the
excess of a first integral of the characteristic flow, in the zero-section case
$I_{t}=nt[\sup_{E}\varepsilon_{0}-\varepsilon_{0}\circ y_{0}]$; its amplitude is invariant
in time, $\operatorname{osc}(\varepsilon_{t})=\operatorname{osc}(\varepsilon_{0})$, and its
non-degeneracy is therefore decided by the initial section alone. The large deviation
principle at speed $t$ is the long-time regime of the same formula and requires one
condition, that the endpoint map converge in $C^{1}(E)$ to a diffeomorphism
(Section~\ref{sec:measure}).
\item {The entropy production is an observable of the same measure.} It is
$e=\tfrac12g^{T}Ag-\sigma$ with $g=\nabla\ln\rho$ and $\sigma=-\nabla\!\cdot v$; its large
deviation principle follows by contraction from the endpoint one, with rate function
$J(s)=\inf\{I(y):\ \bar e_{\infty}(y)=s\}$, and a fibre involution gives a
Gallavotti--Cohen-type relation between the forward and the reversed ensembles under a
separation condition on the tilted action
\cite{gallavotti1995dynamical,evans1993probability,lebowitzspohn1999gallavotti}; the
finite-time measure and the entropy production are in addition evaluated by a
stochastic-path simulation of the base equation driven by the coupling $\xi=A\phi$,
reweighted by the ratio of the contact density to the path-law density, on a model
with a nonzero section (Section~\ref{sec:applications}). The objects that the fluctuation theorems take as
given are here obtained from the geometry.
\end{enumerate}

Two analytic conditions delimit the construction and a third delimits the duality: the
finite-time statements require a smooth solution of Eq.~\eqref{eq:hj} whose characteristics issued
from the initial section do not meet; the large deviation principle requires in addition that
the endpoint map converge in $C^{1}(E)$ to a diffeomorphism; the duality requires a separation
property of the tilted action (Sections~\ref{sec:main}, \ref{sec:measure}
and~\ref{sec:applications}). The fibre coefficient is degenerate along $\nabla H^{(0)}$, which
is precisely where the superlinearity and hypoellipticity hypotheses of the classical theory
fail, so no coercivity or superlinearity hypothesis and no comparison principle for the
Hamilton--Jacobi equation is used.

The paper is organised as follows. Section~\ref{sec:found} introduces the contact form on
the stochastic vector bundle and the probability measure it carries.
Section~\ref{sec:main} states the probability measure theorem. Section~\ref{sec:H}
constructs the action functional through the contact potentials and the master equation.
Section~\ref{sec:measure} establishes the supporting lemmas for the action linearity, the
Gibbs form and the rate function, and derives the zero-section and large deviation
corollaries. Section~\ref{sec:applications} derives the entropy-production and duality
corollaries. Section~\ref{sec:numerical} verifies the finite-time measure and the entropy
production by stochastic-path sampling, and the long-time rate functions on a
damped variant of the same model.
Section~\ref{sec:conclusion} concludes.

\section{Preliminaries}\label{sec:found}

Contact geometry describes dissipative systems by replacing the conserved energy of Hamiltonian mechanics with a constraint function whose extremals define the evolution. The geometry is fixed by a contact 1-form $\Theta$ satisfying the non-degeneracy condition $\Theta\wedge(\dd\Theta)^{\wedge n}\neq0$. The least constraint theorem for stochastic vector bundles is stated in~\cite{zhong2025least} for real vector bundles $E\subset\mathbb{R}^n$; this section recalls its real form, on which the construction of Sections~\ref{sec:H} and~\ref{sec:measure} is built. Two definitions fix the bundle and its contact structure; one theorem states the contact structure, the contact dynamics, and the associated Hamilton--Jacobi equation, and one theorem constructs the contact probability measure. The proofs are given in~\cite{zhong2025least} and~\cite{zhong2026cgle}.

\begin{definition}[Stochastic vector bundle]\label{def:stochastic-vector-bundle}
Let $E$ be an $n$-dimensional real vector space and let $\mathbb{P}$ be the space of probability measures on $E$. The stochastic vector bundle over $E$, introduced in~\cite{zhong2025least}, is the infinite-order jet bundle
\begin{equation}
\pi^{\infty}_{E,0}:\;J^{\infty}(E,\mathbb{P})\;\longrightarrow\;E,
\label{eq:stochastic-jet-bundle}
\end{equation}
whose fibre over a point of $E$ carries the jets of a probability measure $P\in\mathbb{P}$. The bundle is infinite in its order, that is, in the number of jet coordinates retained, while its typical fibre remains finite-dimensional.
\end{definition}

\begin{definition}[Contact structure of the stochastic vector bundle]\label{def:contact-form}
On the stochastic vector bundle of Definition~\ref{def:stochastic-vector-bundle}, take
local coordinates $(P,y^{i},\phi_{i})$ on $J^{k}(E,\mathbb{P})$, where $P$ is the
probability measure, the $y^{i}$ are the base coordinates, and
\begin{equation}
\phi_{i}=\sum_{k=1}^{\infty}\ \sum_{\mu_1<\cdots<\mu_k}B^{\mu_1\cdots\mu_k}_{i}P_{\mu_1\cdots\mu_k},
\qquad
P_{\mu_1\cdots\mu_k}=\partial_{\mu_1\cdots\mu_k}P,
\qquad i=1,\dots,n,
\label{eq:local-connection}
\end{equation}
encodes the higher-order derivatives of the measure. The coefficients
$B^{\mu_1\cdots\mu_k}_{i}$ are required to satisfy
\begin{enumerate}[label=(\arabic*)]
\item injectivity, that is $B^{J}_{i}v^{i}=0$ for every multi-index $J$ implies $v^{i}=0$;
\item surjectivity, that is the linear map $C_{J}\mapsto\sum_{J}B^{J}_{i}C_{J}$ is
surjective from the sequence space onto $\R^{n}$.
\end{enumerate}
Under (1) and (2) the differentials $\dd P$, $\dd y^{i}$ and $\dd\phi_{i}$ are linearly
independent, and
\begin{equation}
\Theta_{\phi}=\dd P-\phi_{i}\,\dd y^{i}
\label{eq:contact-form-jet}
\end{equation}
is a contact $1$-form on $J^{k}(E,\mathbb{P})$, since
\begin{equation}
\Theta_{\phi}\wedge(\dd\Theta_{\phi})^{\wedge n}
=(-1)^{n}\,n!\,\dd P\wedge\bigwedge_{i=1}^{n}\bigl(\dd\phi_{i}\wedge\dd y^{i}\bigr)
\neq0 .
\label{eq:contact-volume-jet}
\end{equation}
The extended phase space is
\begin{equation}
\EE=T^{*}E\times\R,
\label{eq:extended-phase-space}
\end{equation}
with coordinates $(t,y^{i},\phi_{i})$, where $t$ is real. Along a path $y(t)$ the
total derivative of the measure is $\dd P=\dot P\,\dd t$ with
$\dot P=\dd P(\dot y)=\dot y^{i}\partial P/\partial y^{i}$, and replacing
$\dd P$ by $\dot P\,\dd t$ carries $\Theta_{\phi}$ to
\begin{equation}
\Theta=H\,\dd t-\phi_{i}\,\dd y^{i},
\qquad
H:=\dot P=\dd P(\dot y)=\dot y^{i}\frac{\partial P}{\partial y^{i}},
\label{eq:contact-form}
\end{equation}
with $H=H(t,y,\phi)$ a smooth function of $t$, $y^{i}$ and $\phi_{i}$, the contact
probability potential. The contact volume form is
\begin{equation}
\Omega_{C}=\Theta\wedge(\dd\Theta)^{\wedge n},
\qquad
n=\tfrac12\operatorname{rank}\bigl(\dd\Theta\big|_{\ker\Theta}\bigr)=\dim E .
\label{eq:contact-volume-bundle}
\end{equation}
\end{definition}

\begin{theorem}[Least constraint for stochastic vector bundles]\label{thm:least-constraint}
In the setting of Definitions~\ref{def:stochastic-vector-bundle}
and~\ref{def:contact-form}, and for a smooth real function $\varepsilon$ on $\EE$ with
$\varepsilon\neq0$, the following hold.
\begin{enumerate}[label=(\arabic*)]
\item \textbf{Contact structure.} The form $\Omega_{C}$ in
Eq.~\eqref{eq:contact-volume-bundle} is nowhere vanishing, so that $\Theta$ is
non-degenerate and characterises a non-trivial contact structure on
$\EE$.

\item \textbf{Contact dynamics.} $\Theta$ and $\varepsilon$ determine a unique vector
field
\begin{equation}
\mathfrak{X}(\EE)=\{X_{H}\mid X_{H}\in\ker(\dd\Theta)\subset T\EE\}
\label{eq:characteristic-set}
\end{equation}
of the form
\begin{equation}
X_{H}=\frac{\partial}{\partial t}+\dot y^{i}\frac{\partial}{\partial y^{i}}
+\dot\phi_{i}\frac{\partial}{\partial\phi_{i}}.
\label{eq:vector}
\end{equation}
Let $\gamma\colon I\to\EE$ be an integral curve of $X_{H}$, that is
$\dot\gamma(t)=X_{H}\big|_{\gamma(t)}$, with $I\subset\mathbb{R}$ an interval. If $\Theta$
and $\varepsilon=\varepsilon(\gamma(t))$ satisfy
\begin{equation}
\iota_{X_{H}}\Theta=-\varepsilon,
\qquad
\iota_{X_{H}}\dd\Theta=0,
\qquad
\mathcal{L}_{X_{H}}\Theta=0,
\label{eq:three-conditions}
\end{equation}
the contact dynamics is
\begin{equation}
\dot y^{i}=\{y^{i},H\}=\frac{\partial H}{\partial\phi_{i}},
\qquad
\dot\phi_{i}=\{\phi_{i},H\}=-\frac{\partial H}{\partial y^{i}},
\qquad
\frac{\partial\varepsilon}{\partial t}=-\{\varepsilon,H\},
\label{eq:contact-odes}
\end{equation}
with the bracket defined by
\begin{equation}
\{F,G\}=\frac{\partial F}{\partial y^{i}}\frac{\partial G}{\partial\phi_{i}}
-\frac{\partial F}{\partial\phi_{i}}\frac{\partial G}{\partial y^{i}},
\label{eq:poisson-bracket}
\end{equation}
for any two functions $F$ and $G$ of $(y,\phi)$. The time evolution of any function $F$ is
then
\begin{equation}
\frac{\dd F}{\dd t}=\{F,H\}+\frac{\partial F}{\partial t}.
\label{eq:dynamics-F}
\end{equation}

\item \textbf{Least constraint and the contact Hamilton--Jacobi equation.} The flow of
$\mathfrak{X}(\EE)$ corresponds to the extremal paths of the action functional of
$\varepsilon$, defined as
\begin{equation}
S=\int\varepsilon\,\dd t=-\int\Theta.
\label{eq:action}
\end{equation}
The least constraint is the stationarity of the action under variations that vanish at the
endpoints,
\begin{equation}
\delta S=0,
\label{eq:least-constraint}
\end{equation}
equivalently the contact Hamilton--Jacobi-type equation
\begin{equation}
\frac{\partial S}{\partial t}+H\!\left[y,\frac{\partial S}{\partial y},t\right]=0,
\label{eq:hj}
\end{equation}
with $\phi_{i}$ the derivative of $S$ with respect to $y^{i}$.
\end{enumerate}
\end{theorem}

Throughout, $\EE=T^{*}E\times\R$ carries the coordinates $(t,y^{i},\phi_{i})$,
the contact form $\Theta=H\,\dd t-\phi_{i}\,\dd y^{i}$ and the constraint function
$\varepsilon=\phi_{i}\dot y^{i}-H$, with characteristic field
$X_{H}\in\ker(\dd\Theta)$ and
\begin{equation}\label{eq:action-def}
S:=\int\varepsilon\,\dd s=-\int\Theta,
\qquad
\frac{\dd S}{\dd s}=\varepsilon,
\end{equation}
normalised by $S(\gamma(0))=0$ on an integral curve $\gamma$ of $X_{H}$. Write
$\xi:=\ker\Theta\subset T\EE$ for the contact distribution, a real vector bundle
of rank $2n$ over $\EE$, with
\begin{equation}\label{eq:n-def}
n:=\tfrac12\rank\bigl(\dd\Theta|_{\xi}\bigr)=\dim E .
\end{equation}

\begin{theorem}[Contact probability measure]\label{thm:contact-measure}
The contact structure $(\EE,\Theta)$ of Definition~\ref{def:contact-form} determines on $E$
the probability measure
\begin{equation}
\mathcal{P}[y]=\frac{1}{\mathcal{Z}}\,e^{\,nS[y]},
\qquad
\mathcal{Z}=\int_{E}e^{\,nS[y]}\,\dd^{n}y ,
\label{eq:prob-density}
\end{equation}
at each fixed time $t$, provided the action $S$ of Eq.~\eqref{eq:action} is real-valued and
$0<\mathcal{Z}<\infty$; here $\dd^{n}y$ is the coordinate volume element of $E$ and $n$ is
the exponent of Eq.~\eqref{eq:contact-volume-bundle}.
\end{theorem}

\section{Main Theorem}\label{sec:main}

\begin{theorem}[Action, probability measure and rate function of stochastic systems]\label{thm:main}
Let $E$ be compact and boundaryless, and let $(\EE,\Theta)$, $X_H$, $\varepsilon$ and $n$
be as in Section~\ref{sec:found}. Let $S\in C^\infty([0,T)\times E)$, $T\in(0,\infty]$, be
real-valued and solve the contact Hamilton--Jacobi equation
$\partial_tS+H(t,y,\nabla S)=0$; set $S_t:=S(t,\cdot)$ and $S_0:=S(0,\cdot)$, and write
$\Sigma_0:=\{(0,y,\nabla S_0(y)):y\in E\}$ for the initial section. The endpoint map of the
characteristic flow issued from $\Sigma_0$ is
\begin{equation}\label{eq:endpoint-map}
y(t;y_0):=\pi_E\!\left(
\Phi_t^{X_H}\bigl(y_0,\nabla S_0(y_0)\bigr)\right),
\end{equation}
and the constraint function on $\Sigma_0$ is
\begin{equation}\label{eq:eps-zero}
\varepsilon_0(y_0):=\varepsilon\bigl(0,y_0,\nabla S_0(y_0)\bigr).
\end{equation}
Assume that $y(t;\cdot)$ is a diffeomorphism of $E$ for every $t\in(0,T)$, with inverse
$y_0(t,\cdot)$. Then, for every $t\in(0,T)$, the following hold.

\begin{enumerate}[label=(\arabic*),leftmargin=2.6em]
 
\item \textbf{Action.} The constraint function is a first integral of $X_H$,
\begin{equation}\label{eq:main-euler}
\varepsilon=\DD[H]=\phi_i\frac{\partial H}{\partial\phi_i}-H ,
\qquad
\frac{\dd\varepsilon}{\dd s}=0 ,
\end{equation}
its order-two form is
\begin{equation}\label{eq:main-H}
H=H^{(0)}+v\cdot\phi+\tfrac12\,\phi^{T}\!A\,\phi ,
\qquad
\varepsilon=-H^{(0)}+\tfrac12\,\phi^{T}\!A\,\phi ,
\end{equation}
with $v:=H^{(1)}$ the drift and $A$ symmetric, its value on the section
$\phi=\nabla S_t$ is a function of the endpoint,
\begin{equation}\label{eq:main-eps-t}
\varepsilon_t(y):=\varepsilon\bigl(t,y,\nabla S_t(y)\bigr)
=\tfrac12\bigl(\nabla S_t(y)\bigr)^{T}\!A(t,y)\,\nabla S_t(y)-H^{(0)}(t,y)
=\varepsilon_0\bigl(y_0(t,y)\bigr) ,
\end{equation}
and the action with $\dd S/\dd s=\varepsilon$ is affine in the elapsed time,
\begin{equation}\label{eq:main-action}
S_t(y)=S_0\bigl(y_0(t,y)\bigr)+t\,\varepsilon_t(y) .
\end{equation}

\item \textbf{Probability measure.} With $\mathcal Z_t:=\int_Ee^{nS_t}\,\dd^ny$, the contact
probability measure carried by $\Theta$ on the section $\phi=\nabla S_t$ is
\begin{equation}\label{eq:main-density}
\mathcal P_t(\dd y)=\frac{1}{\mathcal Z_t}\,e^{nS_t(y)}\dd^ny,
\end{equation}
the Gibbs form with formal inverse-temperature parameter $nt$, explicit at every finite
time, with no path integral and no variational problem over curves.

\item \textbf{Rate function.} The deficit of the action below its supremum is the exponent
of the same measure,
\begin{equation}\label{eq:main-measure}
I_t(y):=n\Bigl[\sup_E S_t-S_t(y)\Bigr]\ \ge\ 0,
\qquad
\mathcal P_t(\dd y)=\frac{e^{-I_t(y)}}{\int_Ee^{-I_t}\,\dd^ny}\,\dd^ny ,
\end{equation}
and it is the excess of a transported first integral,
\begin{equation}\label{eq:main-rate}
I_t(y)=n\biggl[\sup_E\Bigl(S_0\circ y_0+t\,\varepsilon_t\Bigr)
-S_0\bigl(y_0(t,y)\bigr)-t\,\varepsilon_t(y)\biggr],
\end{equation}
so that $I_t$ is a conserved quantity carried by the characteristic flow rather than the
value of a limit, and $I_t\equiv0$ if and only if $S_0\circ y_0+t\,\varepsilon_t$ is
constant on $E$.
\end{enumerate}
\end{theorem}

\begin{proof}
\begin{itemize}[leftmargin=4.6em,labelsep=0.6em]
\item[(1)] The contact potential and the constraint function are fixed by the master equation and the order-by-order consistency of Section~\ref{sec:H}: Lemma~\ref{lem:transport-deg} determines $A$ through the projected Lie transport, and Lemma~\ref{lem:H0-construction} determines $H^{(0)}$ from the invariant density, with $v=H^{(1)}$ the drift. The endpoint map is a diffeomorphism, so each pair $(t,y)$ determines a
unique characteristic issued from the initial section. Along this characteristic,
$\dd S/\dd s=\varepsilon$, while $\dd\varepsilon/\dd s=0$ by
Theorem~\ref{thm:least-constraint}(2). Integration from $0$ to $t$ therefore gives the
affine-in-time formula for $S_t$. Evaluating the conserved quantity on the Legendre section
$\phi=\nabla S_t$ then gives the stated expression for $\varepsilon_t$. The characteristic
parametrisation is supplied by Lemma~\ref{lem:endpoint-map}, and the integration and
on-section evaluation are carried out in Lemma~\ref{lem:action-linear}.

\item[(2)] Part (1) expresses the exponent $nS_t$ as the transported initial action plus
a term linear in $t$. Substituting it into the density supplied by
Theorem~\ref{thm:contact-measure} gives Eq.~\eqref{eq:main-density}, and $E$ compact with
$S_t\in C^\infty(E)$ gives $0<\mathcal Z_t<\infty$. The normalisation is proved in
Lemma~\ref{lem:gibbs}.

\item[(3)] Since $S_t\le\sup_E S_t$, $I_t\ge0$, with equality exactly at the maximisers of
$S_t$; by part (1) $S_t$ is constant if and only if $S_0\circ y_0+t\,\varepsilon_t$ is
constant, which is the stated criterion for $I_t\equiv0$. Dividing
Eq.~\eqref{eq:main-density} by $\exp(n\sup_E S_t)$ gives Eq.~\eqref{eq:main-measure}, and
substituting part (1) into $I_t$, with the supremum taken over the initial section through
the endpoint bijection, gives Eq.~\eqref{eq:main-rate}. Both identities are proved in
Lemma~\ref{lem:gibbs}.
\end{itemize}
\end{proof}

\section{Action Functional}\label{sec:H}

\subsection{Master Equation and Homogeneous Expansion}

The master equation is the structural condition under which the constraint function becomes a first integral of the open stochastic system. Its homogeneous expansion yields a recursive system of order-based consistency conditions that determine the contact potential order by order.

\subsubsection{Master Equation of Constraint Function}

The fibre Euler operator is $\phi_i\,\partial/\partial\phi_i$, and $\DD$ is that operator minus the identity. The constraint function $\varepsilon$ can then be expressed as the action of the operator $\DD$ on the contact potential:
\begin{equation}
\varepsilon = \DD[H] = \phi_i \frac{\partial H}{\partial \phi_i} - H.
\label{eq:constraint_function}
\end{equation}
The two descriptions of $\varepsilon$ are reconciled by the Legendre map. Along a path, the base velocity is identified with the fibre gradient of $H$,
\begin{equation}
 \dot y^i=\frac{\partial H}{\partial\phi_i},
 \label{eq:legendre_map}
\end{equation}
and substituting $\varepsilon=\phi_i\dot y^i-H$ into Eq.~\eqref{eq:legendre_map} recovers Eq.~\eqref{eq:constraint_function}. Thus Eq.~\eqref{eq:legendre_map} is the Legendre transform carrying the cotangent description $T^{*}\mathcal{E}$ to the tangent description $T\mathcal{E}$, with $H$ as its generating function and $\varepsilon=\DD[H]$ as the Legendre transform of $H$; the stationarity condition of the fibre transform is $\partial_\phi[\phi\cdot\dot y-H]=0$, i.e.\ $\dot y^i=\partial H/\partial\phi_i$, which is Eq.~\eqref{eq:legendre_map}. Here $\varepsilon=\DD[H]$ is the phase-space representative. For symmetric positive semidefinite $A$, the velocity representative on $\dot y-v\in\operatorname{im}A$ is
\begin{equation}\label{eq:epsilon-velocity}
\varepsilon(\dot y)
=\phi\cdot\dot y-H
=\tfrac12(\dot y-v)^TA^{+}(\dot y-v)-H^{(0)},
\end{equation}
where $\phi$ is any solution of $v+A\phi=\dot y$; the value is independent of this choice. The two representatives agree on this domain. Equation~\eqref{eq:epsilon-velocity} also gives the inverse relation $\phi_i=\partial\varepsilon/\partial\dot y^i$, with the derivative restricted to $\operatorname{im}A$. When $A$ is degenerate the fibre map $\phi\mapsto\dot y=v+A\phi$ is invertible after restriction of both its domain and codomain to $\operatorname{im}A$, with inverse $\phi=A^+(\dot y-v)$.

The total time derivative of the constraint function along the contact flow satisfies
\begin{equation}
\frac{\dd\varepsilon}{\dd t} = \frac{\partial \varepsilon}{\partial t} + \{\varepsilon, H\},
\label{eq:epsilon_evolution}
\end{equation}
where the canonical Poisson bracket $\{\cdot,\cdot\}$ is as defined in Eq.~\eqref{eq:poisson-bracket}.
The conservation condition $\dd\varepsilon/\dd t=0$ is equivalent to the master equation that the contact potential must satisfy:
\begin{equation}
\DD\left[\frac{\partial H}{\partial t}\right] + \{\DD[H], H\} = 0.
\label{eq:master-equation}
\end{equation}
This equivalence follows by substituting $\varepsilon = \DD[H]$ and noting that
\begin{equation}\label{eq:deps-dt}
\DD[\partial_t H] = \phi_i \frac{\partial^2 H}{\partial\phi_i\partial t} - \frac{\partial H}{\partial t} = \frac{\partial\varepsilon}{\partial t}.
\end{equation}

\subsubsection{Homogeneous Expansion and Order-Based Consistency}\label{subsec:master}

\begin{definition}[$\phi$-homogeneous expansion]\label{def:H-expansion}
The contact potential admits a homogeneous Taylor expansion in $\phi$:
\begin{equation}\label{eq:H-expansion}
H(t,y,\phi) = \sum_{n=0}^{\infty} H^{(n)}(t,y,\phi), \qquad H^{(n)} = \frac{1}{n!}\,H^{(n)\,i_1 i_2 \cdots i_n}(t,y)\,\phi_{i_1}\phi_{i_2}\cdots\phi_{i_n},
\end{equation}
where $H^{(n)}$ is a homogeneous polynomial of degree $n$ in $\phi$, and $H^{(n)\,i_1\cdots i_n}$ is a fully symmetric $n$-th order tensor.
\end{definition}

By Euler's homogeneous function theorem, $\phi_i\,\partial H^{(n)}/\partial\phi_i = n\,H^{(n)}$, hence $\DD[H^{(n)}] = (n-1)\,H^{(n)}$. The constraint function therefore has the expansion
\begin{equation}\label{eq:epsilon-expansion}
\varepsilon = \DD[H] = \sum_{n=0}^{\infty}(n-1)\,H^{(n)} = -H^{(0)} + \sum_{n=2}^{\infty}(n-1)\,H^{(n)}.
\end{equation}
Substituting this into the master equation Eq.~\eqref{eq:master-equation} and collecting terms by the homogeneous order $p$ of $\phi$ (noting that $\{H^{(n)}, H^{(m)}\}$ is homogeneous of degree $n+m-1$) yields the order-by-order consistency equation for each $p=0,1,2,\dots$:
\begin{equation}\label{eq:order-consistency}
(p-1)\,\frac{\partial H^{(p)}}{\partial t} + \sum_{n+m-1=p}(n-1)\,\{H^{(n)}, H^{(m)}\} = 0.
\end{equation}
The explicit forms for the low orders, which form the basis of the construction, are:

\begin{enumerate}[label=(\arabic*)]
\item Order $p=0$ (zeroth-order transport):
\begin{equation}\label{eq:transport-zeroth}
\frac{\partial H^{(0)}}{\partial t} + H^{(1)i}(t,y)\,\frac{\partial H^{(0)}}{\partial y^i} = 0.
\end{equation}

\item Order $p=1$ (first-order degeneracy):
\begin{equation}\label{eq:degeneracy-first}
H^{(2)\,ij}(t,y)\,\frac{\partial H^{(0)}}{\partial y^j} = 0, \qquad \forall\, i=1,2,\dots,n.
\end{equation}

\item Order $p=2$: the consistency equation reads
\begin{equation}\label{eq:second-order-consistency}
\frac{\partial H^{(2)}}{\partial t} + \{H^{(2)}, H^{(1)}\} - 3\,\{H^{(0)}, H^{(3)}\} = 0.
\end{equation}

\item Orders $p\geq 3$: the general recurrence is
\begin{equation}\label{eq:general-recurrence}
\begin{split}
(p-1)\,\frac{\partial H^{(p)}}{\partial t} &+ (p-1)\,\{H^{(p)}, H^{(1)}\}\\
&+ \sum_{n=2}^{p-1}(n-1)\,\{H^{(n)}, H^{(p+1-n)}\} \\
&- (p+1)\,\{H^{(0)}, H^{(p+1)}\} = 0.
\end{split}
\end{equation}
\end{enumerate}

\begin{definition}[Material derivative and Jacobian]\label{def:material-deriv}
The material derivative and the Jacobian matrix are
\begin{equation}\label{eq:material-derivative}
\frac{D}{Dt} = \partial_t + H^{(1)}\cdot\nabla, \qquad M^{ij} = \frac{\partial H^{(1)i}}{\partial y^j}.
\end{equation}
\end{definition}

\subsection{Transport, Degeneracy, and Exact Closure}

The order-based consistency conditions decompose into transport equations that determine the time evolution of each order of the contact potential, degeneracy conditions that constrain the fluctuation tensor, and a structural closure condition that ensures exact satisfaction of the master equation at finite truncation order.

\subsubsection{Transport, Degeneracy, and Projected Lie Transport}\label{subsec:transport-deg}

\begin{lemma}[Transport and degeneracy conditions]\label{lem:transport-deg}
For the $N=2$ truncation $H = H^{(0)} + H^{(1)i}\phi_i + \tfrac12\,\phi_i A^{ij}\phi_j$ with $A^{ij}$ symmetric, the order-by-order consistency yields:
\begin{align}
\text{Zeroth-order transport:}\quad & \partial_t H^{(0)} + H^{(1)}\cdot\nabla H^{(0)} = 0, \label{eq:transport-zeroth-restated}\\
\text{First-order degeneracy:}\quad & A\,\nabla H^{(0)} = 0, \label{eq:degeneracy-first-matrix}\\
\text{Second-order transport:}\quad & \frac{D A}{Dt} = \Mf\,A + A\,\Mf^{T} + R_A, \label{eq:transport-second}
\end{align}
where $D/Dt=\partial_t+H^{(1)}\cdot\nabla$ (Definition~\ref{def:material-deriv}), $\Mf=PMP$ is the reduced Jacobian, $P$ is the orthogonal projector onto the level sets of $H^{(0)}$, and $R_A$ is the rotational compensator (both defined in the proof below).
\end{lemma}

\begin{proof}
We derive Eq.~\eqref{eq:transport-second} in five steps from the $N=2$ consistency condition.

\textit{Step 1: Standard Lie transport from $N=2$ consistency}. For $N=2$ truncation ($H^{(n)}=0$ for $n>2$), the order-$p=2$ consistency equation Eq.~\eqref{eq:second-order-consistency} reduces to
\begin{equation}\label{eq:standard-lie-derivation}
\frac{\partial H^{(2)}}{\partial t} + \{H^{(2)}, H^{(1)}\} = 0.
\end{equation}
We expand the Poisson bracket explicitly. Substituting $H^{(1)}=H^{(1)k}\phi_k$ into the definition~\eqref{eq:poisson-bracket} gives
\begin{equation}\label{eq:poisson-expand}
\{H^{(2)}, H^{(1)}\} = H^{(1)k}\frac{\partial H^{(2)}}{\partial y^k} - H^{(2)kl}\phi_l\, M^{mk}\,\phi_m,
\end{equation}
where $M^{mk} = \partial H^{(1)m}/\partial y^{k}$ is the Jacobian matrix (Definition~\ref{def:material-deriv}). Translating the second term into matrix notation, $H^{(2)kl}\phi_l\, M^{mk}\,\phi_m$ represents the quadratic form $\phi^{T} H^{(2)} M^{T}\phi$. Symmetrising this quadratic form yields
\begin{equation}\label{eq:symmetrize}
\phi^{T} H^{(2)} M^{T}\phi = \tfrac{1}{2}\,\phi^{T}\!\left(M H^{(2)} + H^{(2)} M^{T}\right)\phi,
\end{equation}
since for any matrix $B$, $\phi^{T}B\phi = \phi^{T}\tfrac{1}{2}(B+B^{T})\phi$. Thus the Poisson bracket becomes
\begin{equation}\label{eq:poisson-simplified}
\{H^{(2)}, H^{(1)}\} = H^{(1)k}\frac{\partial H^{(2)}}{\partial y^k} - \tfrac{1}{2}\,\phi^{T}\!\left(M H^{(2)} + H^{(2)} M^{T}\right)\phi.
\end{equation}
Matching the symmetric matrix coefficients on both sides of
Eq.~\eqref{eq:standard-lie-derivation}, with the material derivative of
Definition~\ref{def:material-deriv} and $H^{(2)}=\tfrac12\phi^{T}A\phi$, yields the standard Lie transport equation
\begin{equation}\label{eq:standard-lie}
\frac{D A}{Dt} = M A + A M^{T}.
\end{equation}
This equation is valid regardless of whether $H^{(0)}$ is constant or non-constant. It preserves the degeneracy constraint $A\,\nabla H^{(0)}=0$: since the gradient is transported by
\begin{equation}\label{eq:gradient-transport-standard}
\frac{D\nabla H^{(0)}}{Dt}=-M^{T}\nabla H^{(0)},
\end{equation}
and $A\,\nabla H^{(0)}=0$, one has
\begin{equation}\label{eq:constraint-preservation-standard}
\frac{D}{Dt}\bigl(A\,\nabla H^{(0)}\bigr)=M\,A\,\nabla H^{(0)}=0.
\end{equation}

\textit{Step 2: Jacobian decomposition into tangential and normal blocks}. To separate the tangential and normal blocks of the transport relative to the level sets of $H^{(0)}$, we decompose the right-hand side using the orthogonal decomposition $I = P + (I - P)$ with
\begin{equation}\label{eq:projection}
P^{ij} = \delta^{ij} - \frac{1}{|\nabla H^{(0)}|^2}\,
\frac{\partial H^{(0)}}{\partial y^i}\,\frac{\partial H^{(0)}}{\partial y^j}
\qquad (\nabla H^{(0)}\neq 0)
\end{equation}
(with $P=I$ when $\nabla H^{(0)}=0$). This yields the four-block decomposition of the Jacobian:
\begin{equation}\label{eq:jacobian-decomp}
\begin{split}
M &= \bigl(P + (I - P)\bigr)\, M\, \bigl(P + (I - P)\bigr) \\
&= \Mf + P\,M\,(I - P) + (I - P)\,M\,P + (I - P)\,M\,(I - P),
\end{split}
\end{equation}
where $\Mf = PMP$ is the reduced Jacobian. Substituting this decomposition into the right-hand side of Eq.~\eqref{eq:standard-lie} and expanding:
\begin{equation}\label{eq:lie-expanded}
M A + A M^{T} = \Mf\,A + A\,\Mf^{T} + R_A^{\mathrm{full}},
\end{equation}
where $R_A^{\mathrm{full}}$ collects all terms involving at least one factor of $(I - P)$:
\begin{equation}\label{eq:R-full}
\begin{split}
R_A^{\mathrm{full}} ={}& \bigl[P\,M\,(I - P) + (I - P)\,M\,P + (I - P)\,M\,(I - P)\bigr]\,A \\
&+ A\,\bigl[(I - P)\,M^{T}\,P + P\,M^{T}\,(I - P) + (I - P)\,M^{T}\,(I - P)\bigr].
\end{split}
\end{equation}

\textit{Step 3: Degeneracy eliminates normal terms}. The degeneracy condition $A\,\nabla H^{(0)} = 0$ Eq.~\eqref{eq:degeneracy-first-matrix} implies that $A$ maps the normal direction $\nabla H^{(0)}$ to zero. Since $(I - P)$ is the orthogonal projector onto the span of $\nabla H^{(0)}$, we have
\begin{equation}\label{eq:A-projector-kill}
A\,(I - P) = 0, \qquad (I - P)\,A = 0,
\end{equation}
where the second identity follows from the symmetry of $A$. As immediate consequences, the $P$-adjacent terms yield
\begin{equation}\label{eq:A-projector-consequence}
P\,A = A, \qquad A\,P = A.
\end{equation}
Evaluating the six terms in $R_A^{\mathrm{full}}$ Eq.~\eqref{eq:R-full} using Eq.~\eqref{eq:A-projector-kill}--\eqref{eq:A-projector-consequence}, all terms vanish except two, leaving
\begin{equation}\label{eq:R-reduced}
R_A^{\mathrm{full}} = (I - P)\,M\,A + A\,M^{T}\,(I - P) = R_A,
\end{equation}
which is the rotational compensator:
\begin{equation}\label{eq:rotation-term}
R_A^{ij} = \frac{1}{|\nabla H^{(0)}|^2}\biggl[
\frac{\partial H^{(0)}}{\partial y^j}\,\bigl(A\,M^T\,\nabla H^{(0)}\bigr)^i
+ \frac{\partial H^{(0)}}{\partial y^i}\,\bigl(A\,M^T\,\nabla H^{(0)}\bigr)^j\biggr],
\end{equation}
a symmetric rank-at-most-$2$ correction tracking the rotation of the constraint direction. When $H^{(0)}=c$ is constant, $P=I$ and $R_A$ vanishes identically. When $\nabla H^{(0)} \neq 0$, writing $I - P = \mathbf{n}\otimes\mathbf{n}$ with $\mathbf{n} = \nabla H^{(0)}/|\nabla H^{(0)}|$, the compensator becomes
\begin{equation}\label{eq:RA-dyadic}
R_A = A\,M^{T}\,\mathbf{n}\otimes\mathbf{n} + \mathbf{n}\otimes\mathbf{n}\,M\,A.
\end{equation}

\textit{Step 4: Projected equation}. Substituting Eq.~\eqref{eq:R-reduced} back into Eq.~\eqref{eq:lie-expanded} yields the projected Lie transport equation Eq.~\eqref{eq:transport-second}.

\textit{Step 5: Constraint preservation}. To verify that Eq.~\eqref{eq:transport-second} strictly preserves the degeneracy constraint $A\,\nabla H^{(0)}=0$ for all $t>0$, we act on $\nabla H^{(0)}$ from the right. Since $H^{(0)}$ satisfies the zeroth-order transport $DH^{(0)}/Dt = 0$ Eq.~\eqref{eq:transport-zeroth-restated}, its gradient evolves by Eq.~\eqref{eq:gradient-transport-standard}. This identity involves the
full Jacobian, whereas Eq.~\eqref{eq:transport-second} involves only the projected $\Mf$: it follows
from differentiating $\partial_tH^{(0)} = -v\cdot\nabla H^{(0)}$ and is independent of the
projection, and
\begin{equation}\label{eq:Mf-kills-gradient}
\Mf^{T}\nabla H^{(0)} = P\,M^{T}P\,\nabla H^{(0)} = 0 .
\end{equation}
The two Jacobians are
reconciled by the identity Eq.~\eqref{eq:RA-reduces}, established below, which
cancels the normal component exactly. Using $P\,\nabla H^{(0)} = 0$ and $(I - P)\,\nabla H^{(0)} = \nabla H^{(0)}$, direct substitution yields
\begin{equation}\label{eq:constraint-preserved}
\frac{D}{Dt}\bigl(A\,\nabla H^{(0)}\bigr) = \bigl[\Mf\,A + A\,\Mf^{T} + R_A\bigr]\nabla H^{(0)} + A\bigl(-M^{T}\nabla H^{(0)}\bigr).
\end{equation}
The first two terms of Eq.~\eqref{eq:constraint-preserved} vanish,
\begin{align}
&\Mf\,\nabla H^{(0)} = PMP\,\nabla H^{(0)} = 0
\qquad (P\nabla H^{(0)}=0), \nonumber\\
&A\,\Mf^{T}\nabla H^{(0)} = A\,P\,M^{T}\,P\,\nabla H^{(0)} = 0 .
\label{eq:constraint-terms-vanish}
\end{align}
while the rotational compensator reduces to
\begin{equation}\label{eq:RA-reduces}
R_A\,\nabla H^{(0)} = A\,M^{T}(I-P)\nabla H^{(0)} = A\,M^{T}\nabla H^{(0)}
\qquad ((I-P)\nabla H^{(0)}=\nabla H^{(0)}),
\end{equation}
together with $(I-P)M\,A\,\nabla H^{(0)} = 0$ (since $A\,\nabla H^{(0)}=0$). Hence
\begin{equation}\label{eq:constraint-preserved-final}
\frac{D}{Dt}\bigl(A\,\nabla H^{(0)}\bigr) = A\,M^{T}\nabla H^{(0)} - A\,M^{T}\nabla H^{(0)} = 0.
\end{equation}
The degeneracy constraint is therefore strictly preserved for all $t>0$, provided the initial data satisfy it.
\end{proof}

\subsubsection{The Exact Closure of Contact Potential}\label{subsec:truncation}

The master equation~\eqref{eq:master-equation} governs $H(t,y,\phi)$; a truncated potential
$H = \sum_{n=0}^N H^{(n)}$ cannot satisfy it exactly in general, because the discarded terms
$H^{(N+1)},\dots$ couple back into the lower orders through the Poisson brackets. The
question is which conditions on the retained coefficients close the truncation exactly.

We expand the master equation order by order in the fibre coordinate $\phi$. Collecting terms of homogeneous degree $p$, the master equation is equivalent to the infinite sequence of consistency equations $\widetilde{C}_p = 0$, where
\begin{equation}\label{eq:Cp_0}
\widetilde{C}_p := (p-1) \frac{\partial H^{(p)}}{\partial t} + \sum_{n=0}^{p+1} (n-1) \{H^{(n)}, H^{(p+1-n)}\} = 0.
\end{equation}
We denote the bracket-only part by $C_p := \sum_{n=0}^{p+1} (n-1) \{H^{(n)}, H^{(p+1-n)}\}$, so that $\widetilde{C}_p = (p-1)\,\partial_t H^{(p)} + C_p$. This distinction matters because for orders $p=0$ and $p \geq 2$ the time-derivative term is generally nonzero (yielding evolution equations), whereas for $p=1$ (where the prefactor $p-1=0$) and for $p \geq N+1$ (where $H^{(p)}=0$ by truncation) the time derivative vanishes identically and $\widetilde{C}_p = C_p$ (yielding constraint conditions).
With the truncation convention $H^{(N+1)} \equiv 0$ and $H^{(n)} \equiv 0$ for all $n>N$, the time-derivative term vanishes identically for $p \geq N+1$, and the problem reduces to identifying the conditions under which every $\widetilde{C}_p$ vanishes.

For orders $p=0$ and $p=1$,
\begin{equation}\label{eq:C0-poisson}
\widetilde{C}_0=-\partial_tH^{(0)}-\{H^{(0)},H^{(1)}\},
\qquad
\{H^{(0)},H^{(1)}\}=H^{(1)i}\frac{\partial H^{(0)}}{\partial y^i},
\end{equation}
the second identity using $H^{(1)}=H^{(1)i}\phi_i$. Thus $\widetilde{C}_0=0$ is exactly the
zeroth-order transport condition Eq.~\eqref{eq:transport-zeroth}, i.e.\ $H^{(0)}$ is
Lie-transported along $H^{(1)}$.

For $p=1$, the prefactor $p-1=0$ kills the time-derivative term, so 
\begin{equation}\label{eq:C_1}
\widetilde{C}_1 = C_1 = -2\{H^{(0)}, H^{(2)}\} = -2 H^{(2)ij} \phi_j \frac{\partial H^{(0)}}{\partial y^i}.
\end{equation}
For this to vanish identically for all $\phi$ the coefficient of each $\phi_j$ must vanish,
which is the first-order degeneracy condition Eq.~\eqref{eq:degeneracy-first}: $H^{(2)}$ is
confined to the tangent space of the level sets of $H^{(0)}$.

For orders $2 \leq p \leq N$, expanding the summation in $\widetilde{C}_p$ for $n=0$ and $n=p+1$, and using the antisymmetry of the Poisson bracket $\{H^{(p+1)}, H^{(0)}\} = -\{H^{(0)}, H^{(p+1)}\}$, yields the boundary term $-(p+1)\{H^{(0)}, H^{(p+1)}\}$. The $n=1$ term vanishes because the prefactor is zero. The $n=p$ term gives $(p-1)\{H^{(p)}, H^{(1)}\}$. Thus, $\widetilde{C}_p$ simplifies to
\begin{equation}
\begin{split}
\widetilde{C}_p 
&= (p-1) \frac{\partial H^{(p)}}{\partial t} + (p-1) \{H^{(p)}, H^{(1)}\} \\
&+ \sum_{n=2}^{p-1} (n-1) \{H^{(n)}, H^{(p+1-n)}\} - (p+1) \{H^{(0)}, H^{(p+1)}\}.
\end{split}
\end{equation}
For this to vanish for every $p=2,3,\dots,N$, the coefficients must satisfy the high-order
recurrence condition Eq.~\eqref{eq:general-recurrence}, in which each $H^{(p)}$ is determined by
the lower-order coefficients $H^{(0)}$, $\dots$, $H^{(p-1)}$ together with the next-order
coefficient $H^{(p+1)}$, the latter acting as a source term through the bracket
$\{H^{(0)}, H^{(p+1)}\}$.

For orders $p \geq N+1$, the time derivative vanishes because $H^{(p)} = 0$, so $\widetilde{C}_p = C_p$. By the truncation convention $H^{(k)} = 0$ for $k > N$, the summation in $C_p$ reduces to indices where both $n \leq N$ and $p+1-n \leq N$. This requires $n \geq p+1-N$. Thus,
\begin{equation}
\label{eq:C_p}
C_p = \sum_{n=\max(0, p+1-N)}^{N} (n-1) \{H^{(n)}, H^{(p+1-n)}\}.
\end{equation}
These are the residual Poisson brackets that survive the truncation: they are built from the retained coefficients only (pairs with $n\le N$ and $p+1-n\le N$); for general $N$ they do not vanish identically, which is the structural closure condition. For $C_p$ to vanish for all $p \geq N+1$, these residual brackets must vanish identically, yielding the structural closure condition
\begin{equation}\label{eq:structural_closure_condition}
\sum_{n=\max(0, p+1-N)}^{N} (n-1) \{H^{(n)}, H^{(p+1-n)}\} = 0.
\end{equation}

Collecting the four requirements that emerge from the order-by-order analysis, the truncated contact potential $H = \sum_{n=0}^N H^{(n)}$ exactly satisfies the master equation Eq.~\eqref{eq:master-equation} if and only if $H^{(0)}, H^{(2)}, \dots, H^{(N)}$ satisfy:
\begin{enumerate}[label=(\arabic*)]
\item the zero-order transport condition ($\widetilde{C}_0 = 0$, an evolution equation determining $H^{(0)}$),
\item the first-order degeneracy condition ($C_1 = 0$, a constraint on $H^{(2)}$ relative to $H^{(0)}$),
\item the high-order recurrence condition for all $p=2,3,\dots,N$ ($\widetilde{C}_p = 0$, an evolution equation determining $H^{(p)}$),
\item the structural closure condition for all $p \geq N+1$ ($C_p = 0$, a constraint ensuring no information leaks through the truncation boundary), with the truncation convention $H^{(N+1)} \equiv 0$ and $H^{(n)} \equiv 0$ for all $n>N$.
\end{enumerate}

When these conditions hold, the constraint function $\varepsilon$ is strictly conserved along the contact flow.

The four conditions are a mixed system: (1) and (3) are first-order evolution equations in
$t$, requiring initial data ($H^{(0)}(0,y)$ from the stationary potential below, $H^{(p)}(0,y)$
for $p\ge2$ from the physical problem, e.g.\ $A(0,y)$ for $H^{(2)}$), while (2) and (4) are
algebraic constraints preserved by the evolution (Section~\ref{subsec:transport-deg}). Since
$E$ is compact and boundaryless, no boundary condition enters and $\nabla\!\cdot(\rho v)=0$
is imposed on all of $E$. The degeneracy condition Eq.~\eqref{eq:degeneracy-first-matrix} makes
$A$ singular, so the contact Hamilton--Jacobi equation Eq.~\eqref{eq:hj} is degenerate along
$\nabla H^{(0)}$ and admits neither a coercive nor a comparison-principle treatment; hence
Section~\ref{sec:main} takes a smooth solution $S$ and a diffeomorphic endpoint map
as hypothesis.

For the $N=2$ truncation, which is the focus of this paper, the structural closure condition
(4) holds automatically: for $p\ge3$ the sum in Eq.~\eqref{eq:structural_closure_condition} has
lower bound $p-1\ge2$ and upper bound $2$, so the only surviving term is
$\{H^{(2)},H^{(2)}\}$ at $p=3$, which vanishes by antisymmetry, and the interval is empty for
$p>3$. Conditions (1)--(3) therefore suffice for exact satisfaction, with no closure
assumption. When $H^{(0)}=c$ is constant, (2) is automatic and, for $N=2$, the recurrence
reduces to the Lie transport equation Eq.~\eqref{eq:second-order-consistency} with $H^{(3)}=0$.

\subsubsection{Geometric Interpretation of the First-Order Degeneracy}\label{subsec:degeneracy-geom}

The degeneracy condition $A\,\nabla H^{(0)} = 0$ carries geometric content~\cite{zhong2025least}.
Linear algebraically $\nabla H^{(0)}\in\ker(A)$: for non-constant $H^{(0)}$ the stiffness $A$ is
singular with rank at most $n-1$, its normal eigenvalue vanishes, the fluctuations are confined
to the iso-$H^{(0)}$ surfaces, and the normal direction is a lost controllable degree of freedom;
for constant $H^{(0)}$ the condition is trivial, $A$ is generically full rank, and all
directions are controllable.

\subsection{Construction of the Contact Potential and Action}

With the consistency framework established, the action functional is constructed through the contact potential, order by order for the $N=2$ truncation: the zeroth-order potential from the invariant measure, the first-order as the drift, and the second-order as the fluctuation stiffness tensor.

\subsubsection{Zero-Order Contact Potential \texorpdfstring{$H^{(0)}$}{H0}}\label{subsec:H0}

The zeroth-order contact potential is determined by solving the zeroth-order transport equation Eq.~\eqref{eq:transport-zeroth}, a first-order linear PDE. We derive its general solution by a separated ansatz that reduces the transport equation to a defining relation between a geometric potential $\chi(y)$ and a source $\alpha(y)$. The pair $(\chi,\alpha)$ is then fixed by the invariant-measure construction. The invariant density $\rho$ is taken smooth and strictly positive throughout, so that $\ln\rho$ and $g=\nabla\ln\rho$ are well defined.

\begin{lemma}[Contact-potential construction]\label{lem:H0-construction}
Let $v$ be smooth with backward-complete flow and let $\rho$ be the unique ergodic invariant density with $\nabla\!\cdot(\rho v)=0$. Define
\begin{equation}\label{eq:sigma-def}
\sigma(t,y) := -\nabla\!\cdot v(t,y),
\end{equation}
and
\begin{equation}\label{eq:gamma0-def}
\gamma_0 > 0,\qquad [\gamma_0]=T^{-1}=[H^{(0)}],\qquad \text{(free scale parameter, not }\langle\sigma\rangle_\rho\text{)}.
\end{equation}
The zeroth-order contact potential is
\begin{equation}\label{eq:H0-construction}
H^{(0)}(t,y) = \gamma_0\ln\!\frac{\rho\bigl(\Phi_{-t}(y)\bigr)}{\rho_0} = \gamma_0\!\left[\ln\!\frac{\rho(y)}{\rho_0} - \int_0^t \sigma(t-\tau, \Phi_{-\tau}(y))\,\dd\tau\right],
\end{equation}
where $\Phi_{-\tau}$ is the flow of $v$ and $\rho_0$ is a spatially uniform reference probability
density. The value of $\rho_0$ is a free reference scale: replacing it by another uniform constant
shifts $H^{(0)}$ by a constant, and therefore leaves the gradient $\nabla H^{(0)}$ and the degeneracy
condition Eq.~\eqref{eq:degeneracy} unchanged.
\end{lemma}

\begin{proof}
Let $\mathcal{L} := v^i\partial/\partial y^i$ denote the Lie derivative along the vector field $v(t,y)$, and let $\Phi_t$ denote the flow generated by $v$. The flow $\{\Phi_t\}_{t \in \mathbb{R}}$ constitutes a one-parameter group of diffeomorphisms acting on the base $E$. The zeroth-order transport equation Eq.~\eqref{eq:transport-zeroth} demands that $H^{(0)}$ is a first integral of the extended flow $\partial_t + \mathcal{L}$:
\begin{equation}\label{eq:H0-transport-recast}
\frac{\partial H^{(0)}}{\partial t} + \mathcal{L}\, H^{(0)} = 0.
\end{equation}

\textit{Step 1: Separated ansatz}. We introduce a smooth geometric potential $\chi(y)$ and define its associated source $\alpha(y) := \mathcal{L}\, \chi(y)$. The defining relation
\begin{equation}\label{eq:source-relation}
\mathcal{L}\, \chi(y) = \alpha(y)
\end{equation}
reduces the problem of determining $H^{(0)}$ to the identification of $\chi$ and $\alpha$. We seek a separated solution of the transport equation Eq.~\eqref{eq:H0-transport-recast} of the form $H^{(0)}(t, y) = \chi(y) + F(t, y)$. Substituting this ansatz and using Eq.~\eqref{eq:source-relation} yields
\begin{equation}\label{eq:F-equation}
\frac{\partial F}{\partial t} + \mathcal{L}\, F = -\alpha(y).
\end{equation}

\textit{Step 2: Construction of $F$ by backward characteristic}. To satisfy Eq.~\eqref{eq:F-equation}, we construct $F$ by accumulating the source $\alpha$ along the backward characteristic $\Phi_{-\tau}(y)$. Setting the initial condition $F(0, y) = 0$, the explicit time-dependence is
\begin{equation}\label{eq:F-integral}
F(t, y) = -\int_0^t \alpha\bigl(\Phi_{-\tau}(y)\bigr)\,\dd\tau.
\end{equation}
To verify this construction, we compute the material derivative of $F$. Using the chain rule and
\begin{equation}\label{eq:chain-rule-alpha}
\frac{\dd}{\dd\tau}\alpha\bigl(\Phi_{-\tau}(y)\bigr)=-(\mathcal{L}\alpha)\bigl(\Phi_{-\tau}(y)\bigr),
\end{equation}
we find
\begin{equation}\label{eq:LF-verify}
\mathcal{L}\, F = \alpha\bigl(\Phi_{-t}(y)\bigr) - \alpha(y).
\end{equation}
The partial time derivative gives $\partial_t F = -\alpha(\Phi_{-t}(y))$. Summing these yields $\partial_t F + \mathcal{L}\, F = -\alpha(y)$, which exactly cancels the $\alpha(y)$ from $\mathcal{L}\, \chi = \alpha$. Thus,
\begin{equation}\label{eq:master-formula}
H^{(0)}(t, y) = \chi(y) - \int_0^t \alpha\bigl(\Phi_{-\tau}(y)\bigr)\,\dd\tau
\end{equation}
strictly satisfies the zeroth-order transport equation. The construction holds for an arbitrary smooth source $\alpha(y)$.

\textit{Step 3: Invariant-measure construction of $\chi$ and $\alpha$}. The master formula Eq.~\eqref{eq:master-formula} expresses $H^{(0)}$ through a geometric potential $\chi$ and a source $\alpha$ linked by the defining relation $\mathcal{L}\, \chi = \alpha$. The physical meaning of $H^{(0)}$ motivates the construction: it represents the background potential of the system on the zero-fibre (the deterministic skeleton). In dissipative systems, deterministic trajectories collapse onto attractors; the natural mathematical tool to describe the long-term statistical behaviour is the ergodic invariant measure $\mu$.

Let $\Omega$ denote the canonical volume form on the base manifold, and let $\mu$ denote the ergodic invariant measure of the flow $\Phi_t$ generated by $v$, supported on the $\omega$-limit set of a typical trajectory. When $\mu$ is absolutely continuous with respect to $\Omega$, denote its Radon--Nikodym derivative by
\begin{equation}\label{eq:rho-def}
\rho(y) := \frac{\dd\mu}{\dd\Omega}(y),
\end{equation}
which carries the dimension $[\rho]=L^{-n}$ of a probability density. Since the logarithm requires dimensionless arguments, we introduce a reference probability density $\rho_0$ (with $[\rho_0]=[\rho]$, spatially uniform) and the dimensional constant $\gamma_0$ with $[\gamma_0]=T^{-1}=[H^{(0)}]$, so that $\gamma_0\ln(\rho/\rho_0)$ carries the same dimension as $H^{(0)}$.

\textit{Step 4: Logarithmic continuity equation}. To determine $\chi$ and $\alpha$ physically, we utilize the probability conservation of the invariant measure. The invariance $\Phi_t^*\mu = \mu$ reads, in density form, $\mathcal{L}(\rho\,\Omega)=0$. Expanding by the Leibniz rule and dividing by $\rho\,\Omega$, we obtain the logarithmic continuity equation:
\begin{equation}\label{eq:continuity-log}
\mathcal{L}(\rho\,\Omega) = (\mathcal{L}\rho)\,\Omega + \rho\,(\mathcal{L}\Omega) = 0 \;\Longrightarrow\; \frac{\mathcal{L}\rho}{\rho} = -\,\frac{\mathcal{L}\Omega}{\Omega}.
\end{equation}
Since the local phase-space volume contraction rate is
\begin{equation}\label{eq:volume-contraction}
\mathcal{L}\Omega/\Omega = \nabla\!\cdot\, v = -\sigma(t,y),
\end{equation}
the continuity equation becomes
\begin{equation}\label{eq:liouville-steady}
v\cdot\nabla\ln\rho = \sigma(t,y) \quad\Longleftrightarrow\quad \nabla\!\cdot(\rho v) = 0,
\end{equation}
which states that the growth of $\rho$ along the flow exactly balances the volume contraction.

\textit{Step 5: Identification of $\chi$ and $\alpha$}. Comparing the logarithmic continuity equation Eq.~\eqref{eq:liouville-steady} with the defining relation $\mathcal{L}\, \chi = \alpha$, we identify the geometric potential as the dimensionally-correct logarithmic density and the source as the contraction rate:
\begin{equation}\label{eq:invariant-measure-rule}
\chi(y) := \gamma_0\,\ln\frac{\rho(y)}{\rho_0}, \qquad \alpha(y) := \gamma_0\,\sigma(y) = -\gamma_0\,\nabla\!\cdot\, v(y),
\end{equation}
With these definitions, the defining relation $\mathcal{L}\, \chi = \alpha$ is strictly satisfied as the continuity equation of the invariant measure:
\begin{equation}\label{eq:source-from-invariance}
\mathcal{L}\, \chi = \gamma_0\, v\cdot\nabla\ln\frac{\rho}{\rho_0} = \gamma_0\, v\cdot\nabla\ln\rho = \gamma_0\,\sigma = \alpha,
\end{equation}
where the second equality uses the spatial uniformity of $\rho_0$. Substituting Eq.~\eqref{eq:invariant-measure-rule} into the master formula Eq.~\eqref{eq:master-formula} gives Eq.~\eqref{eq:H0-construction}.
\end{proof}

The construction above is written for an autonomous drift $v=v(y)$, for which the time
argument $t-\tau$ in Eq.~\eqref{eq:H0-construction} is redundant; for a time-dependent drift
$v(t,y)$ the same argument applies with the two-parameter flow $\Phi_{s,t}$,
$H^{(0)}(t,y)=\gamma_0\ln(\rho(\Phi_{0,t}(y))/\rho_0)$ satisfies
$\partial_tH^{(0)}+v(t,y)\cdot\nabla H^{(0)}=0$, and
\begin{equation}\label{eq:two-parameter-log}
\ln\frac{\rho(\Phi_{0,t}(y))}{\rho(y)}=-\int_0^t\sigma\bigl(s,\Phi_{s,t}(y)\bigr)\,\dd s,
\end{equation}
which is exactly the integrand written in Eq.~\eqref{eq:H0-construction} with $\Phi_{-\tau}$ denoting $\Phi_{t-\tau,t}$.

Write $\langle f\rangle_{\rho}:=\int_{E}f\,\rho\,\dd^{n}y$ for the average with respect to
the invariant density. For every $f\in C^{1}(E)$ the Leibniz rule and
$\nabla\!\cdot(\rho v)=0$ give $\nabla\!\cdot(\rho v f)=\rho\,v\cdot\nabla f$, and $E$ is
compact without boundary (Section~\ref{sec:main}), so the divergence theorem
gives $\langle v\cdot\nabla f\rangle_{\rho}=0$. With $f=\ln\rho$ and
Eq.~\eqref{eq:liouville-steady},
\begin{equation}\label{eq:mean-sigma-zero}
\langle\sigma\rangle_\rho=\bigl\langle v\cdot\nabla\ln\rho\bigr\rangle_{\rho}=0.
\end{equation}
Thus the mean contraction rate vanishes identically, which is why $\gamma_0$ must be a free scale parameter rather than $\langle\sigma\rangle_\rho$.

Along the deterministic flow,
\begin{equation}\label{eq:lnrho-derivative}
\frac{\dd}{\dd\tau}\ln\rho\bigl(\Phi_{-\tau}(y)\bigr)=-\sigma\bigl(\Phi_{-\tau}(y)\bigr),
\end{equation}
and integration gives the coboundary identity
\begin{equation}\label{eq:lnrho-coboundary-bound}
\ln\frac{\rho(\Phi_{-t}(y))}{\rho(y)}=-\int_0^t\sigma(\Phi_{-\tau}(y))\,\dd\tau=O(1),
\end{equation}
which is bounded uniformly in $t$ and $y$ whenever $\ln(\rho/\rho_0)$ is bounded, as it is on a compact $E$ with $\rho$ smooth and strictly positive. Here $\sigma$ is a general smooth function, a coboundary $\sigma=v\cdot\nabla\ln(\rho/\rho_0)$ by Eq.~\eqref{eq:liouville-steady}, so the time integral in Eq.~\eqref{eq:H0-construction} is the bounded quantity of Eq.~\eqref{eq:lnrho-coboundary-bound}, and $H^{(0)}(t,y)=\gamma_0\ln(\rho(\Phi_{-t}(y))/\rho_0)$ is bounded for all $t$, with no $\gamma_0^2t$ secular term, compatible with $\langle\sigma\rangle_\rho=0$. In the divergence-free case $\sigma\equiv0$ (with $\rho\equiv\rho_0$) one has $H^{(0)}\equiv0$, $g=0$ and $e=0$.

\subsubsection{First-Order Contact Potential \texorpdfstring{$H^{(1)}=v$}{H1=v}}

By the contact dynamical equation $\dot{y}^i = \partial H/\partial\phi_i$ and the Taylor expansion Eq.~\eqref{eq:H-expansion}, the zeroth-order term $H^{(0)}$ is independent of $\phi$, and thus its derivative $\partial H^{(0)}/\partial\phi_i$ vanishes everywhere. Furthermore, each higher-order term $H^{(n)}$ ($n\geq 2$) is a homogeneous polynomial of degree $n\geq 2$ in $\phi$; its partial derivative $\partial H^{(n)}/\partial\phi_i$ is homogeneous of degree $n-1\geq 1$, and therefore strictly vanishes at $\phi=0$. Consequently, along the zero-fibre section the sole surviving contribution is the first-order term:
\begin{equation}\label{eq:drift-zero-fibre}
\left.\frac{\partial H}{\partial \phi_i}\right|_{\phi=0} = H^{(1)i}(t,y),
\end{equation}
which identifies the first-order coefficient as the drift field~\cite{zhong2025least}:
\begin{equation}\label{eq:H1}
H^{(1)}(t,y) = v(t,y),
\end{equation}
with $v$ the base drift of the dissipative system. In the deterministic limit $\phi \to 0$, the Legendre map Eq.~\eqref{eq:legendre_map} collapses to $\dot{y}^i|_{\phi=0} = H^{(1)i}(t,y) = v^i(t,y)$, recovering the deterministic drift equation. Its Jacobian $M^{ij}=\partial v^i/\partial y^j$ (Definition~\ref{def:material-deriv}) enters the second-order transport equation; its divergence $\nabla\!\cdot v = -\sigma(t,y)$ (Eq.~\eqref{eq:sigma-def}) enters $H^{(0)}$.

\subsubsection{Second-Order Contact Potential \texorpdfstring{$H^{(2)}$}{H2}}\label{subsec:H2}

For $N=2$ truncation, the quadratic component is $H^{(2)}=\tfrac12\,\phi_i A^{ij}(t,y)\phi_j$ with $A^{ij}$ a symmetric tensor, the second-order contact stiffness of fluctuations, whose eigenvalues correspond to the contact stiffness in each direction of $\phi$. The full contact potential is
\begin{equation}\label{eq:H-full}
H(t,y,\phi) = H^{(0)}(t,y) + v^i(t,y)\,\phi_i + \tfrac12\,\phi_i A^{ij}(t,y)\phi_j.
\end{equation}

The tensor $A$ is determined by the projected Lie transport equation (Lemma~\ref{lem:transport-deg}, Eq.~\eqref{eq:transport-second}),
\begin{equation}\label{eq:A-transport}
\frac{D A}{Dt} = \Mf\,A + A\,\Mf^{T} + R_A,
\end{equation}
where $D/Dt=\partial_t+v\cdot\nabla$, $\Mf=PMP$ is the reduced Jacobian, and $R_A$ is the rotational compensator (Eq.~\eqref{eq:rotation-term}). In matrix notation:
\begin{equation}\label{eq:RA-matrix}
(R_A)^{ij} = \frac{1}{|\nabla H^{(0)}|^2}\biggl[
\frac{\partial H^{(0)}}{\partial y^j}\,\bigl[A\,M^T\,\nabla H^{(0)}\bigr]^i
+ \frac{\partial H^{(0)}}{\partial y^i}\,\bigl[A\,M^T\,\nabla H^{(0)}\bigr]^j\biggr].
\end{equation}
It is simultaneously constrained by the first-order degeneracy condition (Eq.~\eqref{eq:degeneracy-first}),
\begin{equation}\label{eq:degeneracy}
A\,\nabla H^{(0)} = 0,
\end{equation}
i.e., $A$ has a zero eigenvalue in the $\nabla H^{(0)}$ direction. Along the characteristic line $y(t)=Y(t;0,y_0)$, the variation-of-constants form of the solution is~\cite{zhong2025least}
\begin{equation}\label{eq:H2-solution}
A(t, Y(t;0,y_0)) = \widetilde{\Phi}(t,0;y_0)\,A_0(y_0)\,\widetilde{\Phi}^T(t,0;y_0) + \int_0^t \widetilde{\Phi}(t,\tau;y_0)\,R_A(\tau)\,\widetilde{\Phi}^T(t,\tau;y_0)\,\dd\tau,
\end{equation}
where $R_A(\tau):=R_A(\tau,Y(\tau;0,y_0))$ and the left-hand side is evaluated along the characteristic $y(t)=Y(t;0,y_0)$, $A_0(y_0)=A(0,y_0)$ is the initial tensor satisfying $A_0\,\nabla H^{(0)}(0,y_0)=0$, and $\widetilde{\Phi}(t,s;y_0)$ is the fundamental solution matrix of the projected variational equation $\dot{\widetilde{\Phi}}=\Mf\,\widetilde{\Phi}$ with $\widetilde{\Phi}(s,s;y_0)=I$.

The operator $A\mapsto\Mf A+A\Mf^{T}+R_A$ is linear in $A$, since $R_A$ of Eq.~\eqref{eq:rotation-term} is linear in $A$ while $\Mf$, $M$ and $\nabla H^{(0)}$ are determined by $v$ and $\rho$ alone; Eq.~\eqref{eq:H2-solution} is therefore a linear integral equation for $A$, equivalently a linear ODE, hence determined uniquely by $A_0$. It is verified by differentiation along the characteristic line: applying the Leibniz integral rule to the integral term, substituting the variational equation $\dot{\widetilde{\Phi}} = \Mf\,\widetilde{\Phi}$, and factoring yields $DA/Dt = \Mf\,A + A\,\Mf^{T} + R_A$, recovering Eq.~\eqref{eq:A-transport}. The initial condition is satisfied since $\widetilde{\Phi}(0,0;y_0)=I$ and the integral vanishes at $t=0$. When $H^{(0)}=c$ (constant), $P=I$, $\Mf=M$, $R_A=0$, and the solution reduces to the homogeneous Lyapunov form $A(t)=\widetilde{\Phi}(t,0;y_0)\,A_0\,\widetilde{\Phi}^T(t,0;y_0)$.

The second-order term makes the coupling between probability and determinism
visible~\cite{zhong2025least}: the base evolution becomes $\dot{y}^i = v^i + A^{ij}\phi_j$, so
changes in $\phi$ back-react on the base, and for $N=2$ the constraint function is
$\varepsilon = -H^{(0)} + \tfrac12\phi^T A\phi$. In the ideal limit $A$ is transported to zero
while $\phi$ sharpens in compensation; physically the conserved combination
$\tfrac12\phi^T A\phi = \varepsilon + H^{(0)}$ binds $A$ and $\phi$, and since $\phi$ cannot
diverge, $A$ stays bounded away from zero and the coupling $\xi = A\phi$ remains a finite
random variable.

At $t=0$, $\nabla H^{(0)}(0,y) = \gamma_0\,\nabla\ln\rho(y)$, so the degeneracy condition Eq.~\eqref{eq:degeneracy} reduces to the stationary degeneracy $A\,\nabla\ln\rho = 0$: $A$ has a zero eigenvalue in the $\nabla\ln\rho$ direction, confining the fluctuation stiffness to the tangent space of the iso-density surfaces. In the divergence-free case $\sigma\equiv0$ (hence $\rho\equiv\rho_0$, $H^{(0)}\equiv0$),
$\nabla H^{(0)}\equiv0$ and the degeneracy condition $A\nabla H^{(0)}=0$ is trivially preserved;
in the general case $A\,\nabla\ln\rho=0$ is the condition that the initial tensor $A_0$ of Eq.~\eqref{eq:H2-solution} must satisfy.

\subsubsection{Complete Contact Dynamics and Stochastic Coupling}\label{subsec:complete-dynamics}

Substituting the complete second-order potential Eq.~\eqref{eq:H-full} into the contact dynamical equations Eq.~\eqref{eq:contact-odes} yields the explicit dynamics generated by the contact potential. The base equation is
\begin{equation}\label{eq:complete-base}
\dot{y}^i \;=\; \frac{\partial H}{\partial \phi_i} \;=\; v^i(t,y) + A^{ij}(t,y)\,\phi_j ,
\end{equation}
and the fibre equation is
\begin{equation}\label{eq:complete-fibre}
\dot{\phi}_i \;=\; -\frac{\partial H}{\partial y^i}
\;=\; -\frac{\partial H^{(0)}}{\partial y^i}
     - \frac{\partial v^j}{\partial y^i}\,\phi_j
     - \frac{1}{2}\,\frac{\partial A^{jk}}{\partial y^i}\,\phi_j\,\phi_k .
\end{equation}

The base equation decomposes into the drift field and an effective stochastic coupling,
\begin{equation}\label{eq:stochastic-coupling}
\dot{y}^i \;=\; v^i(t,y) + \xi^i , \qquad \xi^i := A^{ij}(t,y)\,\phi_j .
\end{equation}
Here $\xi^i$ quantifies the back-reaction of the probability field on the macroscopic motion;
its stochastic character follows from the fibre coordinate being the weighted sum of spatial
probability gradients (Section~\ref{sec:found}),
\begin{equation}\label{eq:phi-weighted-sum}
\phi_i=\sum_{1\le k\le\infty}\ \sum_{\mu_1<\cdots<\mu_k}B_i^{\mu_1\cdots\mu_k}\,P_{\mu_1\cdots\mu_k},
\end{equation}
so that $\phi=0$ in exactly two situations: the equiprobable distribution, where every gradient
vanishes, and a non-uniform steady state, where the gradients cancel in the weighted sum because
the probability flux is in geometric balance. In both cases $\xi^i=0$ and the base equation
reduces to $\dot{y}^i=v^i$; conversely $\phi\neq0$ signals a net directional probability
gradient generated by an external perturbation. Because $\xi^i$ is fixed by the instantaneous
spatial structure of the probability field rather than by $y$ alone, it is a random variable,
zero in equilibrium and nonzero precisely when probability fluctuations are present, encoded in
the fibre coordinate rather than injected as an independent additive noise.

The randomness of $\xi^i$ is made precise by contrasting the ideal limit of the Lie transport with the physical constraint.

\textit{Ideal limit.} Dropping the rotational compensator $R_A$ from
Eq.~\eqref{eq:A-transport}, which vanishes identically when $H^{(0)}$ is
constant, and retaining only the linear fibre transport leaves the reduced
system $\dot A=\Mf A+A\Mf^{T}$, $\dot\phi=-\Mf^{T}\phi$. Let
$\widetilde\Phi(t,s;y_0)$ be the fundamental matrix of
Eq.~\eqref{eq:H2-solution}, write $\widetilde\Phi(t)=\widetilde\Phi(t,0;y_0)$,
and set $\phi_0=\phi(0)$ and $\xi_0=A_0\phi_0$. The exact factorisations along
the drift characteristic are
\begin{equation}\label{eq:ideal-laws}
A(t)=\widetilde\Phi(t)\,A_0\,\widetilde\Phi(t)^{T},\qquad
\phi(t)=\widetilde\Phi(t)^{-T}\phi_0,\qquad
\xi(t)=\widetilde\Phi(t)\,\xi_0 ,
\end{equation}
the first being the homogeneous case of Eq.~\eqref{eq:H2-solution} and the
third following from $\xi=A\phi$. They give
$\phi^{T}A\phi=\phi_0^{T}A_0\phi_0$ directly, so the quadratic form is
conserved along the flow. If in addition
$\|\widetilde\Phi(t)\|\le C\,e^{-t/\tau_\eta}$ for some $C>0$ and
$\tau_\eta>0$, then
\begin{equation}\label{eq:ideal-decay}
\|A(t)\|\le C^{2}\|A_0\|\,e^{-2t/\tau_\eta},\qquad
\|\xi(t)\|\le C\,\|\xi_0\|\,e^{-t/\tau_\eta},
\end{equation}
with $\|A\|$ the induced matrix 2-norm and $\|\xi\|$ the Euclidean norm. Under
this hypothesis the base dynamics decouples from the probability field in the
limit $t\to\infty$.

\textit{Physical constraint.} Throughout this paragraph, $\|\phi\|$ and $\|\xi\|$ denote the Euclidean vector norm and $\|A\|$ the induced matrix 2-norm. Physically the decoupling limit is never attained. Because $\varepsilon$ is a conserved first integral, the quadratic form
\begin{equation}\label{eq:constraint-quadratic}
\tfrac12\,\phi^T A\,\phi \;=\; \varepsilon + H^{(0)} \;\equiv\; C_\varepsilon
\end{equation}
is fixed along the flow, which binds $A$ and $\phi$ to each other: one can tend to zero only if the other diverges. But $\phi$, the weighted sum of the probability-field gradients, cannot diverge, because $\|\phi\|\to\infty$ would describe a probability field collapsing to a $\delta$-distribution, which no physical system sustains under persistent microscopic fluctuations; hence $\|\phi\|\le\phi_{\max}$. When $C_\varepsilon>0$, the identity $2C_\varepsilon=\phi^T A\phi$ with $\|\phi\|\le\phi_{\max}$ yields the lower bounds
\begin{equation}\label{eq:A-xi-lower-bound}
\|A\|\;\ge\;\frac{2C_\varepsilon}{\phi_{\max}^{\,2}}\;>\;0, \qquad \|\xi\| = \|A\phi\|\;\ge\;\frac{2C_\varepsilon}{\phi_{\max}}\;>\;0 .
\end{equation}
Whenever $C_\varepsilon>0$ the stochastic coupling therefore never vanishes: the decay to zero is a mathematical limit, whereas $\xi$ is in physical reality a finite random variable. Deterministic behaviour (the $A\to0$ limit) and residual stochasticity (finite $\xi$) coexist.

The bound Eq.~\eqref{eq:A-xi-lower-bound} is stated in terms of $\phi_{\max}$, whose finiteness is
argued physically (a diverging $\|\phi\|$ would describe a probability field collapsing to a
$\delta$-distribution). So that the statement is self-contained, the bound $\|\phi\|\le\phi_{\max}$
is imposed here, and is used only in Eq.~\eqref{eq:A-xi-lower-bound}. The conditional estimates of
Eq.~\eqref{eq:ideal-decay} hold under their fundamental-matrix hypothesis,
while the lower bounds of Eq.~\eqref{eq:A-xi-lower-bound} hold for every time
at which $C_\varepsilon>0$ and $\|\phi\|\le\phi_{\max}$ are satisfied; the two
describe distinct aspects of the same conserved quadratic form.

The three terms of the fibre equation Eq.~\eqref{eq:complete-fibre} have a direct physical meaning. The term $-\partial H^{(0)}/\partial y^i$ is the background force generated by the invariant measure; the term $-M^{ji}\phi_j$ is the linearised transport of the fibre coordinate generated by the drift Jacobian $M^{ij}=\partial v^i/\partial y^j$ (Definition~\ref{def:material-deriv}); and the quadratic term $-\tfrac12(\partial A^{jk}/\partial y^i)\,\phi_j\,\phi_k$ accounts for the spatial variation of the contact stiffness. Together with Eq.~\eqref{eq:complete-base}, these form the coupled base--fibre system; its Legendrian projection $\phi=\nabla_y S$ produces the Hamilton--Jacobi system studied next.

The finiteness of $\xi$ is one face of the single conserved quantity $\varepsilon$. By the first-order degeneracy $A\nabla H^{(0)}=0$ (Eq.~\eqref{eq:degeneracy}), the coupling is tangent to the iso-$H^{(0)}$ surfaces,
\begin{equation}\label{eq:xi-tangent}
\xi\cdot\nabla H^{(0)}=\phi\cdot(A\nabla H^{(0)})=0;
\end{equation}
together with the zeroth-order transport $\partial_t H^{(0)}+v\cdot\nabla H^{(0)}=0$ this makes $H^{(0)}$ conserved along the physical path, so the fluctuation quadratic form $\tfrac12\phi^T A\phi=\varepsilon+H^{(0)}\equiv C_\varepsilon$ is a conserved first integral. The same $\varepsilon$ is the coefficient of the contact volume form, $\Omega_C\propto\varepsilon$ (Eq.~\eqref{eq:contact-volume-bundle}), so the contact structure is non-degenerate precisely where $\varepsilon\neq0$. Degeneracy of the constraint function and degeneracy of the contact volume form are therefore different conditions.

\subsubsection{Complete Contact Potential and Coupled HJ System}\label{subsec:coupled}

Substituting $\boldsymbol{\phi}=\nabla_y S$ into Eq.~\eqref{eq:H-full} and applying Theorem~\ref{thm:least-constraint} yields the complete HJ equation
\begin{equation}\label{eq:hj-full}
\frac{\partial S}{\partial t} + H^{(0)}(t,y) + v^i\,\frac{\partial S}{\partial y^i} + \tfrac12\,(\nabla_y S)^T A\,(\nabla_y S) = 0,
\end{equation}
a fully nonlinear first-order PDE for $S(t,y)$, quadratically nonlinear in $\nabla_y S$, with the coefficient $A$ coupled to $S$ through Eq.~\eqref{eq:A-transport}. The three components $H^{(0)},H^{(1)},H^{(2)}$ thus assemble into the closed coupled system
\begin{equation}\label{eq:coupled-system}
\begin{aligned}
&\text{HJ equation:}\quad &&\partial_t S + H^{(0)} + v\cdot\nabla_y S + \tfrac12\,(\nabla_y S)^T A\,(\nabla_y S) = 0,\\
&\text{Second-order transport:}\quad &&\frac{D A}{Dt} = \Mf\,A + A\,\Mf^{T} + R_A,\\
&\text{First-order degeneracy:}\quad &&A\,\nabla H^{(0)} = 0,\\
&\text{Zeroth-order transport:}\quad &&\partial_t H^{(0)} + v\cdot\nabla H^{(0)} = 0,\quad \nabla\!\cdot(\rho v) = 0,
\end{aligned}
\end{equation}
with unknowns $(S,A,\rho)$. The coupling mechanism is: the $\tfrac12(\nabla_y S)^T A(\nabla_y S)$ term makes $S$ nonlinear in $A$; the $\Mf=PMP$ and $R_A$ in the $A$ equation both depend on $H^{(0)}$ (hence on $\rho$); $\rho$ is determined by $\nabla\!\cdot(\rho v)=0$. The characteristic velocity $w=v+A\,\nabla_y S$ of the contact dynamical equations Eq.~\eqref{eq:contact-odes} simultaneously carries the drift $v$ and the fluctuation $A$ information, and is the core feature distinguishing the complete theory from the linear truncation.

\subsubsection{From the Contact Potential to Action}\label{subsec:action}

The coefficients $H^{(0)}$, $v$ and $A$ constructed above determine the constraint function through the operator $\DD$ of Eq.~\eqref{eq:constraint_function}: for the second-order potential Eq.~\eqref{eq:H-full} it is the function $\varepsilon=-H^{(0)}+\tfrac12\phi^{T}A\phi$ of Eq.~\eqref{eq:constraint-quadratic}, which does not contain $v$. The drift enters the action only through the characteristic field $X_H$.

Along a curve $\gamma$ in $\EE$ the action functional of Theorem~\ref{thm:least-constraint}(3) then takes the explicit form
\begin{equation}\label{eq:action-functional-explicit}
\mathcal{S}[\gamma]=\int_{\gamma}\varepsilon\,\dd s
=\int_{\gamma}\Bigl(\tfrac12\,\phi^{T}A\,\phi-H^{(0)}\Bigr)\,\dd s
=-\int_{\gamma}\Theta .
\end{equation}
On an integral curve of $X_H$ the first integral $\varepsilon$ is constant by Theorem~\ref{thm:least-constraint}(2), and Eq.~\eqref{eq:action-def} integrates to
\begin{equation}\label{eq:action-along-characteristic}
S\bigl(\gamma(t)\bigr)=S\bigl(\gamma(0)\bigr)+\int_{0}^{t}\varepsilon\bigl(\gamma(s)\bigr)\,\dd s
=S\bigl(\gamma(0)\bigr)+t\,\varepsilon ,
\end{equation}
with $\varepsilon$ in the last term the constant value on $\gamma$. On the Legendre section $\phi=\nabla_{y}S$ this action is the solution of the Hamilton--Jacobi equation Eq.~\eqref{eq:hj-full}, with $\dd S/\dd s=\varepsilon$ by Eq.~\eqref{eq:action-def}. The construction is thereby closed: the coefficients $(H^{(0)},v,A)$ determine the constraint function, the characteristic flow determines its integral, and the initial value $S_{0}$ determines the action $S$.

\section{Probability Measure and Rate Function}\label{sec:measure}

\subsection{Gibbs Form}\label{subsec:gibbs}

The constraint function is $\varepsilon=\DD[H]$, where $\DD$ is the operator of
Eq.~\eqref{eq:constraint_function}; on the Legendre section $\phi=\nabla_{y}S$ it coincides
with $\varepsilon_{t}$ in Eq.~\eqref{eq:action-linear}, and its initial value is
$\varepsilon_{0}$ in Eq.~\eqref{eq:eps-zero}. The case
$S_{0}\equiv0$, in which $\Sigma_{0}$ is the zero section, is called the zero-section case
and is treated in Section~\ref{subsec:dividing}.

The entropy production rate, the observable of the theory, is
\begin{equation}\label{eq:entropy-prod}
e(t,y):=\tfrac12\,g(y)^{T}A(t,y)\,g(y)-\sigma(y),
\qquad
g:=\nabla\ln\rho,
\qquad
\sigma:=-\nabla\!\cdot v,
\end{equation}
with $\sigma$ as in Eq.~\eqref{eq:sigma-def} and $A$ the transported fluctuation tensor of
Eq.~\eqref{eq:H2-solution}. The quadratic term is evaluated against the tensor at the same
time as the gradient of the entropy potential that enters the degeneracy
condition Eq.~\eqref{eq:degeneracy}; at $t=0$ that condition reads $A(0,y)\,g(y)=0$, since
$\nabla H^{(0)}(0,\cdot)=\gamma_{0}g$ by Eq.~\eqref{eq:H0-construction}.

The reduction of $S$ to $\varepsilon$ is exact, not asymptotic.

\subsubsection{Endpoint Parametrisation of Characteristics}

\begin{lemma}[Endpoint parametrisation of the characteristics]\label{lem:endpoint-map}
Let the objects and hypotheses of Section~\ref{sec:main} be given. Then, for
every $t\in(0,T)$:
\begin{enumerate}[label=(\arabic*),leftmargin=2.6em]
\item for every $y\in E$ there is exactly one integral curve of $X_{H}$ that starts on
$\Sigma_{0}$ and ends at $(t,y,\nabla_{y}S(t,y))$, namely
\begin{equation}\label{eq:gamma-ty}
\gamma_{t,y}(s):=\pi_{E}\Bigl(\Phi^{X_{H}}_{\,s-t}\bigl(y,\nabla_{y}S(t,y)\bigr)\Bigr),
\qquad
\gamma_{t,y}(t)=y,
\qquad
\gamma_{t,y}(0)=y_{0}(t,y),
\end{equation}
where $y_{0}(t,\cdot)$ is the inverse of the endpoint map Eq.~\eqref{eq:endpoint-map};
\item for every first integral $F$ of $X_{H}$, and in particular for $F=\varepsilon$, the
value of $F$ along $\gamma_{t,y}$ equals its value at the initial point; hence
$\varepsilon$ evaluated at $(t,y,\nabla_{y}S(t,y))$ is the function $\varepsilon_{t}(y)$
of Eq.~\eqref{eq:main-eps-t}, and it depends on the endpoint alone;
\item $y_{0}(t,\cdot)\in C^{\infty}(E)$ and $\varepsilon_{t}\in C^{\infty}(E)$.
\end{enumerate}
\end{lemma}
\begin{proof}
\emph{Step 1: the graph of $\nabla_{y}S$ is invariant.} Write
$G_{t}:=\{(t,y,\nabla_{y}S(t,y)):y\in E\}$ and let $s\mapsto(t+s,y(s),\phi(s))$ be an
integral curve of $X_{H}$ with $(y(0),\phi(0))\in G_{t}$. Set
$\Delta(s):=\phi(s)-\nabla_{y}S(t+s,y(s))$. By Eq.~\eqref{eq:contact-odes},
$\dot\phi=-\partial_{y}H$ and $\dot y=\partial_{\phi}H$; differentiating
$\partial_{t}S=-H(t,y,\nabla_{y}S)$ in $y$ gives
\begin{equation}\label{eq:graph-drift}
\frac{\dd}{\dd s}\nabla_{y}S(t+s,y(s))
=\partial_{t}\nabla_{y}S+D^{2}S\,\partial_{\phi}H
=-\partial_{y}H\bigl(t+s,y(s),\nabla_{y}S(t+s,y(s))\bigr).
\end{equation}
Therefore $\Delta$ solves
\begin{equation}\label{eq:graph-ode}
\frac{\dd\Delta}{\dd s}
=-\Bigl[\partial_{y}H\bigl(t+s,y(s),\phi(s)\bigr)
-\partial_{y}H\bigl(t+s,y(s),\nabla_{y}S(t+s,y(s))\bigr)\Bigr],
\qquad \Delta(0)=0 ,
\end{equation}
and $|\Delta'|\le L|\Delta|$, with $L$ the Lipschitz constant of $\partial_{y}H$ on the
compact segment of the orbit. Gronwall's inequality gives $\Delta\equiv0$, hence
$\Phi^{X_{H}}_{s}(G_{t})\subset G_{t+s}$ wherever the flow is defined; the reverse
inclusion follows by applying the same statement to the time $-s$ and exchanging the
roles, so $\Phi^{X_{H}}_{s}(G_{t})=G_{t+s}$.

\emph{Step 2: the integral curve.} Let $p:=(t,y,\nabla_{y}S(t,y))$. By Step 1 with
$s=-t$,
\begin{equation}\label{eq:preimage-in-section}
q:=\Phi^{X_{H}}_{-t}(p)\in G_{0}=\Sigma_{0},
\end{equation}
so $q=(0,y_{0},\nabla S_{0}(y_{0}))$
for a unique $y_{0}\in E$; the endpoint map Eq.~\eqref{eq:endpoint-map} sends this $y_{0}$ to
$\pi_{E}p=y$, hence $y_{0}=y_{0}(t,y)$. The curve Eq.~\eqref{eq:gamma-ty} is the projection of
$s\mapsto\Phi^{X_{H}}_{s-t}(p)$, so it is an integral curve with $\gamma_{t,y}(0)=y_{0}(t,y)$
and $\gamma_{t,y}(t)=y$. If $\tilde\gamma$ is any integral curve from $\Sigma_{0}$ to $p$,
the flow property gives
\begin{equation}\label{eq:uniqueness-flow}
\tilde\gamma(s)=\Phi^{X_{H}}_{s-t}(p)=\gamma_{t,y}(s),
\end{equation}
so the curve is unique.

\emph{Step 3: first integrals.} For a first integral $F$ one has
$\frac{\dd}{\dd s}F\bigl(\Phi^{X_{H}}_{s}(z)\bigr)=0$ for every $z$, hence
\begin{equation}\label{eq:first-integral-endpoint}
F(p)=F(q)=F\bigl(0,y_{0}(t,y),\nabla S_{0}(y_{0}(t,y))\bigr).
\end{equation}
Applied to
$F=\varepsilon$, which is a first integral by
Theorem~\ref{thm:least-constraint}(2), this gives
\begin{equation}\label{eq:eps-endpoint-well-defined}
\varepsilon_{t}(y):=\varepsilon\bigl(t,y,\nabla_{y}S(t,y)\bigr)
=\varepsilon_{0}\bigl(y_{0}(t,y)\bigr),
\end{equation}
a function of $y$ alone.

\emph{Step 4: smoothness.} The flow $\Phi^{X_{H}}$ is $C^{\infty}$ because $X_{H}$ is a
smooth vector field, and $y\mapsto\nabla S_{0}(y)$ is $C^{\infty}$; hence the endpoint map
Eq.~\eqref{eq:endpoint-map} is $C^{\infty}$, and its inverse $y_{0}(t,\cdot)$ is $C^{\infty}$
by the inverse function theorem. Since $\varepsilon_{0}$ in Eq.~\eqref{eq:eps-zero} is
$C^{\infty}$, Eq.~\eqref{eq:eps-endpoint-well-defined} gives
\begin{equation}\label{eq:eps-smooth}
\varepsilon_{t}=\varepsilon_{0}\circ y_{0}(t,\cdot)\in C^{\infty}(E).
\end{equation}
\end{proof}
\subsubsection{Action Linearity and Gibbs Form}

\begin{lemma}[Action and constraint function]\label{lem:action-linear}
Let the objects and hypotheses of Section~\ref{sec:main} be given, with
$H^{(0)}$ and $A$ the coefficients constructed in
Lemma~\ref{lem:H0-construction} and Lemma~\ref{lem:transport-deg}.
Then along every integral curve of $X_{H}$
\begin{equation}\label{eq:action-growth}
\frac{\dd S}{\dd s}=\varepsilon ,
\end{equation}
and the action is linear in the elapsed time up to its initial value,
\begin{equation}\label{eq:action-linear}
S(t,y)=S_{0}\bigl(y_{0}(t,y)\bigr)+t\,\varepsilon_{t}(y),
\qquad
\varepsilon_{t}(y)=\tfrac12\bigl(\nabla_{y}S(t,y)\bigr)^{T}\!A(t,y)\,\nabla_{y}S(t,y)
-H^{(0)}(t,y).
\end{equation}
The constraint function is transported from the initial section,
\begin{equation}\label{eq:explicit-eps}
\varepsilon_{t}=\varepsilon_{0}\circ y_{0}(t,\cdot),
\end{equation}
and
\begin{equation}\label{eq:consistency-identity}
\tfrac12\bigl(\nabla_{y}S(t,y)\bigr)^{T}\!A(t,y)\,\nabla_{y}S(t,y)
=\varepsilon_{0}\bigl(y_{0}(t,y)\bigr)+H^{(0)}(t,y)
\end{equation}
identifies the quadratic form with the conserved quantity
$C_{\varepsilon}=\varepsilon+H^{(0)}$ of Eq.~\eqref{eq:constraint-quadratic}. In the
zero-section case $S_{0}\equiv0$, Eq.~\eqref{eq:action-linear} reduces to
$S_{t}=t\,\varepsilon_{t}$ with $\varepsilon_{t}=-H^{(0)}(0,y_{0}(t,\cdot))$.
\end{lemma}
\begin{proof}For Eq.~\eqref{eq:action-growth},
$\dd S/\dd s=\partial_{t}S+\dot y^{i}\partial_{i}S$; substitute
$\dot y^{i}=\partial H/\partial\phi_{i}$ from Eq.~\eqref{eq:contact-odes},
$\partial_{i}S=\phi_{i}$ and $\partial_{t}S=-H$ from Eq.~\eqref{eq:hj}; the result is
\begin{equation}\label{eq:action-chain}
\phi_{i}\partial H/\partial\phi_{i}-H=\DD[H]=\varepsilon.
\end{equation}
Along $\gamma_{t,y}$ this is
integrated with $\dd\varepsilon/\dd s=0$ from Theorem~\ref{thm:least-constraint},
\begin{equation}\label{eq:action-integrated}
S(t,y)=S_{0}\bigl(y_{0}(t,y)\bigr)+\int_{0}^{t}\varepsilon\bigl(\gamma_{t,y}(s)\bigr)\,\dd s
=S_{0}\bigl(y_{0}(t,y)\bigr)+t\,\varepsilon_{t}(y),
\end{equation}
the value of $\varepsilon$ on the characteristic ending at $(t,y)$ being
$\varepsilon_{t}(y)$ by Lemma~\ref{lem:endpoint-map}(2). For
Eq.~\eqref{eq:explicit-eps}, $\varepsilon$ is a first integral, so its value at
$(t,y,\nabla_{y}S(t,y))$ equals its value at $(0,y_{0}(t,y),\nabla S_{0}(y_{0}(t,y)))$;
evaluating Eq.~\eqref{eq:constraint-quadratic} there gives Eq.~\eqref{eq:explicit-eps}.
Equation~\eqref{eq:consistency-identity} is Eq.~\eqref{eq:action-linear} with
$\varepsilon_{t}$ replaced by Eq.~\eqref{eq:explicit-eps}.\end{proof}
\begin{lemma}[Gibbs form of the contact probability measure]\label{lem:gibbs}
Let the objects and hypotheses of Section~\ref{sec:main} be given, and let
Theorem~\ref{thm:contact-measure} apply to the Legendre section $\phi=\nabla_{y}S_{t}$.
Then for every $t\in(0,T)$ the normalisation
$\mathcal{Z}_{t}:=\int_{E}e^{nS(t,y)}\,\dd^{n}y$ is finite and strictly positive, and the
probability density carried by the contact form is
\begin{equation}\label{eq:gibbs}
\begin{split}
\mathcal{P}_{t}(\dd y)&=\frac{1}{\mathcal{Z}_{t}}\,e^{\,nS(t,y)}\dd^{n}y\\
&=\frac{1}{\mathcal{Z}_{t}}\,\exp\bigl(n\bigl[S_{0}(y_{0}(t,y))+t\,\varepsilon_{t}(y)\bigr]\bigr)\dd^{n}y\\
&=\frac{1}{\int_{E}e^{-I_{t}}}\,e^{-I_{t}(y)}\dd^{n}y,
\end{split}
\end{equation}
with
\begin{equation}\label{eq:rate-I}
I_{t}(y)=n\Bigl[\sup_{E}\bigl(S_{0}\circ y_{0}+t\varepsilon_{t}\bigr)
-S_{0}\bigl(y_{0}(t,y)\bigr)-t\varepsilon_{t}(y)\Bigr].
\end{equation}
Thus $\mathcal{P}_{t}$ has the Gibbs form with formal inverse-temperature parameter
$nt$ and time-dependent effective energy $-\varepsilon_{t}$, up to bounded reweighting by
$\exp(nS_{0}\circ y_{0})$, and the rate function satisfies
\begin{equation}\label{eq:rate-nonneg}
\begin{aligned}
&I_{t}\ge0,\\
&I_{t}(y)=0\ \Longleftrightarrow\
S_{0}\bigl(y_{0}(t,y)\bigr)+t\varepsilon_{t}(y)=\sup_{E}\bigl(S_{0}\circ y_{0}+t\varepsilon_{t}\bigr),\\
&I_{t}\not\equiv0\ \Longleftrightarrow\ S_{0}\circ y_{0}+t\varepsilon_{t}\ \text{is not constant}.
\end{aligned}
\end{equation}
\end{lemma}
\begin{proof}
\emph{Step 1: the normalisation.} By Lemma~\ref{lem:endpoint-map}(3) the function
$S(t,\cdot)$ lies in $C^{\infty}(E)$, and $E$ is compact, so
$\inf_{E}S(t,\cdot)$ and $\sup_{E}S(t,\cdot)$ are finite and attained, and
\begin{equation}\label{eq:gibbs-bounds}
|E|\,e^{n\inf_{E}S(t,\cdot)}\ \le\
\mathcal{Z}_{t}=\int_{E}e^{nS(t,y)}\,\dd^{n}y\ \le\
|E|\,e^{n\sup_{E}S(t,\cdot)} ;
\end{equation}
in particular $0<\mathcal{Z}_{t}<\infty$.

\emph{Step 2: the density.} By Lemma~\ref{lem:endpoint-map}(1) the section
$\phi=\nabla_{y}S(t,\cdot)$ is the image of $\Sigma_{0}$ under $\Phi^{X_{H}}_{t}$, so
Theorem~\ref{thm:contact-measure} applies with this $S$ and gives
\begin{equation}\label{eq:gibbs-density}
\mathcal{P}_{t}(\dd y)=\mathcal{Z}_{t}^{-1}e^{nS(t,y)}\,\dd^{n}y .
\end{equation}

\emph{Step 3: the explicit form.} Substituting Eq.~\eqref{eq:action-linear} into
Eq.~\eqref{eq:gibbs-density} gives
\begin{equation}\label{eq:gibbs-explicit}
\mathcal{P}_{t}(\dd y)=\mathcal{Z}_{t}^{-1}\exp\bigl(n[S_{0}(y_{0}(t,y))
+t\varepsilon_{t}(y)]\bigr)\,\dd^{n}y ,
\end{equation}
which is the middle expression of Eq.~\eqref{eq:gibbs}.

\emph{Step 4: the rate function.} By Eq.~\eqref{eq:rate-I},
\begin{equation}\label{eq:gibbs-exponential}
e^{-I_{t}(y)}=\exp\bigl(-n\sup_{E}S(t,\cdot)\bigr)\,e^{nS(t,y)},
\qquad
\int_{E}e^{-I_{t}}\,\dd^{n}y=\exp\bigl(-n\sup_{E}S(t,\cdot)\bigr)\,\mathcal{Z}_{t},
\end{equation}
hence the quotient of the two identities in Eq.~\eqref{eq:gibbs-exponential} is the last
expression of Eq.~\eqref{eq:gibbs}, which therefore equals the first. Since $S(t,y)\le\sup_{E}S(t,\cdot)$,
$I_{t}\ge0$, with equality exactly at the maximisers of
$S_{0}\circ y_{0}+t\varepsilon_{t}$; and $I_{t}\equiv0$ if and only if
$S_{0}\circ y_{0}+t\varepsilon_{t}$ is constant on $E$. This is Eq.~\eqref{eq:rate-nonneg}.
\end{proof}

\subsubsection{Finite-time Rate Function}

\begin{lemma}[The rate function as the excess of a conserved quantity]
\label{lem:rate-squeeze}
For every $t\in(0,T)$ and every $y\in E$,
\begin{equation}\label{eq:rate-squeeze}
\begin{split}
nt\Bigl[\sup_{E}\varepsilon_{0}-\varepsilon_{0}\bigl(y_{0}(t,y)\bigr)\Bigr]-2n\|S_{0}\|_{\infty}
&\ \le\ I_{t}(y)\\
&\ \le\
nt\Bigl[\sup_{E}\varepsilon_{0}-\varepsilon_{0}\bigl(y_{0}(t,y)\bigr)\Bigr]+2n\|S_{0}\|_{\infty},
\end{split}
\end{equation}
and in the zero-section case $S_{0}\equiv0$ both bounds are equalities,
\begin{equation}\label{eq:explicit-rate}
I_{t}(y)=nt\Bigl[\sup_{E}\varepsilon_{0}-\varepsilon_{0}\bigl(y_{0}(t,y)\bigr)\Bigr],
\end{equation}
the supremum being taken through the bijection $y_{0}(t,\cdot)$ of
Lemma~\ref{lem:endpoint-map}.
\end{lemma}
\begin{proof}
By Eq.~\eqref{eq:explicit-eps} and Lemma~\ref{lem:endpoint-map}(1),
\begin{equation}\label{eq:sup-transport}
\sup_{E}\varepsilon_{t}=\sup_{E}\bigl(\varepsilon_{0}\circ y_{0}(t,\cdot)\bigr)=\sup_{E}\varepsilon_{0}.
\end{equation}
Since $|S_{0}\circ y_{0}|\le\|S_{0}\|_{\infty}$, Eq.~\eqref{eq:action-linear} gives
\begin{equation}\label{eq:sup-St-bounds}
t\sup_{E}\varepsilon_{0}-\|S_{0}\|_{\infty}
\ \le\ \sup_{E}S_{t}\ \le\
t\sup_{E}\varepsilon_{0}+\|S_{0}\|_{\infty},
\end{equation}
and the same bound at the point $y$ gives
\begin{equation}\label{eq:St-point-bounds}
t\,\varepsilon_{0}\bigl(y_{0}(t,y)\bigr)-\|S_{0}\|_{\infty}
\ \le\ S_{t}(y)\ \le\
t\,\varepsilon_{0}\bigl(y_{0}(t,y)\bigr)+\|S_{0}\|_{\infty}.
\end{equation}
Subtracting Eq.~\eqref{eq:St-point-bounds} from Eq.~\eqref{eq:sup-St-bounds} and multiplying by
$n$ gives Eq.~\eqref{eq:rate-squeeze}, by the definition of $I_{t}$ in
Eq.~\eqref{eq:main-measure}. With $S_{0}\equiv0$ both Eq.~\eqref{eq:sup-St-bounds} and
Eq.~\eqref{eq:St-point-bounds} are equalities, and Eq.~\eqref{eq:rate-squeeze} collapses to
Eq.~\eqref{eq:explicit-rate}.
\end{proof}

Equation~\eqref{eq:gibbs} contains no infimum or supremum over curves: the action in the
density is a first integral on the characteristic that ends at $y$, not the value of an
optimisation problem. The construction therefore needs no coercivity, superlinearity or
comparison principle for Eq.~\eqref{eq:hj} and no path measure; its only analytic input is the
absence of caustics in Lemma~\ref{lem:endpoint-map}.

\subsection{Non-degeneracy of Rate Function}\label{subsec:dividing}

The rate function of Eq.~\eqref{eq:gibbs} is non-degenerate in the large-time limit
provided $\varepsilon_{t}$ converges in $C^{0}(E)$. The first result exhibits an invariant
of the family
$\{\varepsilon_{t}\}$ which excludes the degenerate limit; the second solves the case
$S_{0}\equiv0$ with an explicit action, for which the absence of caustics is a
consequence, while the convergence of $\varepsilon_{t}$ remains a condition on the
drift.

\subsubsection{Oscillation Invariance}

\begin{lemma}[Oscillation invariance and absence of separated asymptotics]
\label{lem:oscillation}
Let the objects and hypotheses of Section~\ref{sec:main} be given, with
$T=\infty$.
\begin{enumerate}[label=(\arabic*),leftmargin=2.6em]
\item $y_{0}(t,\cdot)$ is a bijection of $E$, hence for every $t>0$
\begin{equation}\label{eq:oscillation}
\sup_{E}\varepsilon_{t}=\sup_{E}\varepsilon_{0},
\qquad
\inf_{E}\varepsilon_{t}=\inf_{E}\varepsilon_{0},
\qquad
\operatorname{osc}(\varepsilon_{t})=\operatorname{osc}(\varepsilon_{0}),
\end{equation}
and consequently, if $\varepsilon_{t}\to\varepsilon_{\infty}$ in $C^{0}(E)$, then
$\varepsilon_{\infty}$ is non-constant if and only if $\varepsilon_{0}$ is;
\item if in addition $S(t,y)=W(y)-\mu t+o(1)$ as $t\to\infty$, uniformly in $y$, for
some $W\in C^{0}(E)$ and $\mu\in\R$, then the exact linearity of
Lemma~\ref{lem:action-linear} forces
\begin{equation}\label{eq:no-separation}
\varepsilon_{t}\equiv-\mu\ \text{ on } E\ \text{ for every } t>0,
\qquad
S_{0}\circ y_{0}(t,\cdot)\ \longrightarrow\ W\ \text{ in } C^{0}(E),
\end{equation}
so a separated asymptotics is possible only in the degenerate case
$\varepsilon_{0}\equiv$ const, and $W$ is then the uniform limit of the pull-backs of the
initial section; along each characteristic, $W(\gamma(s))\to S_{0}(\gamma(0))$;
\item the amplitude of the rate function grows linearly,
\begin{equation}\label{eq:osc-rate}
nt\,\operatorname{osc}(\varepsilon_{0})-2n\|S_{0}\|_{\infty}
\ \le\ \operatorname{osc}(I_{t})\ \le\
nt\,\operatorname{osc}(\varepsilon_{0})+2n\|S_{0}\|_{\infty},
\end{equation}
so that $I_{t}\not\equiv0$ for all large $t$ whenever $\varepsilon_{0}$ is not constant,
and
\begin{equation}\label{eq:osc-I-zero-section}
\operatorname{osc}(I_{t})=nt\,\operatorname{osc}(\varepsilon_{0})
\end{equation}
in the zero-section case.
\end{enumerate}
\end{lemma}
\begin{proof}
\emph{(1)} By Lemma~\ref{lem:action-linear},
$\varepsilon_{t}=\varepsilon_{0}\circ y_{0}(t,\cdot)$, and $y_{0}(t,\cdot)$ is a
bijection of $E$ by Lemma~\ref{lem:endpoint-map}(1). For a bijection $\psi$ and a
continuous $f$,
\begin{equation}\label{eq:sup-bijection}
\sup_{E}f\circ\psi=\sup_{\psi(E)}f=\sup_{E}f,
\qquad
\inf_{E}f\circ\psi=\inf_{E}f,
\end{equation}
which is Eq.~\eqref{eq:oscillation}. If $\varepsilon_{t}\to\varepsilon_{\infty}$ uniformly,
then
\begin{equation}\label{eq:sup-convergence}
\Bigl|\sup_{E}\varepsilon_{t}-\sup_{E}\varepsilon_{\infty}\Bigr|
\le\|\varepsilon_{t}-\varepsilon_{\infty}\|_{\infty},
\qquad
\Bigl|\inf_{E}\varepsilon_{t}-\inf_{E}\varepsilon_{\infty}\Bigr|
\le\|\varepsilon_{t}-\varepsilon_{\infty}\|_{\infty},
\end{equation}
so $\sup_{E}\varepsilon_{\infty}=\sup_{E}\varepsilon_{0}$ and
$\inf_{E}\varepsilon_{\infty}=\inf_{E}\varepsilon_{0}$; hence
$\varepsilon_{\infty}$ is constant if and only if $\varepsilon_{0}$ is.

\emph{(2)} Fix $t>0$ and $y\in E$, let $\gamma=\gamma_{t,y}$, and let $\varepsilon$ be
the value of the first integral on it. For $0\le s\le t$ the point $(s,\gamma(s))$ lies
on the same characteristic, so Lemma~\ref{lem:action-linear} gives
\begin{equation}\label{eq:separation-exact}
S\bigl(s,\gamma(s)\bigr)=S_{0}(y_{0})+s\,\varepsilon ,
\end{equation}
while the separated form gives $S(s,\gamma(s))=W(\gamma(s))-\mu s+o(1)$ as $s\to\infty$.
Subtracting,
\begin{equation}\label{eq:separation-linear}
W\bigl(\gamma(s)\bigr)=S_{0}(y_{0})+s(\varepsilon+\mu)+o(1).
\end{equation}
Since $W$ is continuous on the compact $E$ it is bounded, so
Eq.~\eqref{eq:separation-linear} divided by $s$ gives $\varepsilon+\mu=0$, and
$\varepsilon=-\mu$ on every characteristic; since for fixed $t$ the characteristics cover
$E$, this is $\varepsilon_{t}\equiv-\mu$ for every $t>0$. In particular
$\varepsilon_{0}\equiv-\mu$ is constant. Returning to Eq.~\eqref{eq:separation-linear},
\begin{equation}\label{eq:W-along-characteristic}
W\bigl(\gamma(s)\bigr)\ \longrightarrow\ S_{0}\bigl(\gamma(0)\bigr)=S_{0}(y_{0})
\qquad\text{as } s\to\infty .
\end{equation}
Substituting $\varepsilon_{t}=-\mu$ into Eq.~\eqref{eq:action-linear} and comparing with the
separated form at the endpoint $y$ gives
\begin{equation}\label{eq:W-uniform-limit}
S_{0}\bigl(y_{0}(t,y)\bigr)-W(y)
=\Bigl[S(t,y)-W(y)+\mu t\Bigr]\ \longrightarrow\ 0
\qquad\text{uniformly in } y ,
\end{equation}
the bracket being the remainder of the separated form. Equation
Eq.~\eqref{eq:W-uniform-limit} is the convergence in Eq.~\eqref{eq:no-separation}, and taking
$y=\gamma(s)$ in it, with $y_{0}(s,\gamma(s))=\gamma(0)$, gives
Eq.~\eqref{eq:W-along-characteristic}.

\emph{(3)} $\sup_{E}S_{t}$ is a number, so
$\operatorname{osc}(I_{t})=n\operatorname{osc}(S_{t})$ by
Eq.~\eqref{eq:main-measure}. For any two bounded functions,
$\operatorname{osc}(f+g)$ lies between $\operatorname{osc}(f)-2\|g\|_{\infty}$ and
$\operatorname{osc}(f)+2\|g\|_{\infty}$; applying this to
$f=t\varepsilon_{t}$, $g=S_{0}\circ y_{0}$ and using Eq.~\eqref{eq:oscillation},
\begin{equation}\label{eq:osc-St}
t\operatorname{osc}(\varepsilon_{0})-2\|S_{0}\|_{\infty}
\ \le\ \operatorname{osc}(S_{t})\ \le\
t\operatorname{osc}(\varepsilon_{0})+2\|S_{0}\|_{\infty},
\end{equation}
which multiplied by $n$ is Eq.~\eqref{eq:osc-rate}. With $S_{0}\equiv0$ both bounds are
equalities.
\end{proof}

\subsubsection{Zero-section Case}

The finite-time conclusions of Theorem~\ref{thm:main}, combined with the contact-potential
construction, specialise as follows.

\begin{corollary}[Explicit action in the zero-section case]\label{cor:zero-section}
Let $S_{0}\equiv0$ and $\gamma_{0}\neq0$. Then:
\begin{enumerate}[label=(\arabic*),leftmargin=2.6em]
\item along every characteristic issued from $\Sigma_{0}$,
\begin{equation}\label{eq:zero-xi}
C_{\varepsilon}\equiv0,
\qquad
\xi=A\phi\equiv0,
\qquad
\varepsilon=-H^{(0)}(t,y),
\end{equation}
the projection is the drift orbit, $y(t;y_{0})=\Phi^{(5)}_{t}(y_{0})$, and the endpoint
map Eq.~\eqref{eq:endpoint-map} is therefore a diffeomorphism of $E$ for every $t>0$;
\item the transported constraint function is
\begin{equation}\label{eq:zero-eps}
\varepsilon_{t}(y)=-\gamma_{0}\ln\frac{\rho\bigl(\Phi^{(5)}_{-t}y\bigr)}{\rho_{0}};
\end{equation}
\item the function
\begin{equation}\label{eq:zero-S}
S(t,y)=-t\,H^{(0)}(t,y)
=-t\,\gamma_{0}\ln\frac{\rho\bigl(\Phi^{(5)}_{-t}y\bigr)}{\rho_{0}}
\ \in\ C^{\infty}\bigl([0,\infty)\times E\bigr)
\end{equation}
is the unique $C^{1}$ solution of Eq.~\eqref{eq:hj} with $S(0,\cdot)=0$ on
$[0,\infty)\times E$;
\item the entropy production at the initial time is
\begin{equation}\label{eq:zero-e}
e(0,\cdot)=-\sigma=\nabla\!\cdot v,
\qquad
\langle e(0,\cdot)\rangle_{\rho}=0 ,
\end{equation}
and at later times it is given by Eq.~\eqref{eq:entropy-prod} with the transported
tensor $A(t,\cdot)$;
\item for $\gamma_{0}>0$ the endpoint probability measure and the finite-time rate function are
\begin{equation}\label{eq:zero-measure}
\begin{split}
\mathcal{P}_{t}(\dd y)&=\frac{\rho\bigl(\Phi^{(5)}_{-t}y\bigr)^{-n\gamma_{0}t}}
{\int_{E}\rho\bigl(\Phi^{(5)}_{-t}y'\bigr)^{-n\gamma_{0}t}\,\dd^{n}y'}\,\dd^{n}y,\\
I_{t}(y)&=n t\,\gamma_{0}\ln\frac{\rho\bigl(\Phi^{(5)}_{-t}y\bigr)}{\min_{E}\rho},
\end{split}
\end{equation}
and $I_{t}\not\equiv0$ if and only if $\rho$ is not constant.
\end{enumerate}
\end{corollary}
\begin{proof}
\emph{Step 1: the conserved quantity on $\Sigma_{0}$.} The quadratic form
\begin{equation}\label{eq:conserved-quadratic}
C_{\varepsilon}=\tfrac12\phi^{T}A\phi=\varepsilon+H^{(0)}
\end{equation}
is a first integral by
Theorem~\ref{thm:least-constraint}(2) and Eq.~\eqref{eq:constraint-quadratic}; on
$\Sigma_{0}$ the fibre coordinate is $\nabla S_{0}=0$, so
\begin{equation}\label{eq:zero-C}
C_{\varepsilon}\big|_{\Sigma_{0}}=0,
\qquad
C_{\varepsilon}\equiv0\ \text{along every characteristic issued from }\Sigma_{0}.
\end{equation}

\emph{Step 2: the characteristic.} Since $A$ is positive semidefinite,
$\tfrac12\phi^{T}A\phi=0$ gives $A\phi=0$, that is $\xi=0$; the base equation
Eq.~\eqref{eq:complete-base} then reads
\begin{equation}\label{eq:zero-base}
\dot y=v,
\end{equation}
so $y(s)=\Phi^{(5)}_{s}(y_{0})$ and the endpoint map of Lemma~\ref{lem:endpoint-map} is
$\Phi^{(5)}_{t}$, a diffeomorphism of $E$ for every $t$ because $v$ is complete on the
compact $E$. With $C_{\varepsilon}=0$,
\begin{equation}\label{eq:zero-eps-char}
\varepsilon=C_{\varepsilon}-H^{(0)}=-H^{(0)}(t,y).
\end{equation}

\emph{Step 3: the action.} Set $S(t,y):=-tH^{(0)}(t,y)$. By the zeroth-order transport
Eq.~\eqref{eq:transport-zeroth-restated} and the degeneracy Eq.~\eqref{eq:degeneracy},
\begin{equation}\label{eq:zero-verify}
\begin{aligned}
\partial_{t}S&=-H^{(0)}+t\,v\!\cdot\!\nabla H^{(0)},\\
\nabla_{y}S&=-t\,\nabla H^{(0)},\\
\tfrac12(\nabla_{y}S)^{T}\!A\,\nabla_{y}S&=\tfrac12t^{2}\nabla H^{(0)T}\!A\nabla H^{(0)}=0,
\end{aligned}
\end{equation}
and the three terms sum to zero, while $S(0,\cdot)=0$ by Eq.~\eqref{eq:zero-S}; hence $S$
solves Eq.~\eqref{eq:hj} on $[0,\infty)\times E$.

\emph{Step 4: uniqueness.} Any $C^{1}$ solution with $S(0,\cdot)=0$ satisfies
$S=t\varepsilon_{t}$ by Lemma~\ref{lem:action-linear}, and
$\varepsilon_{t}=-H^{(0)}(t,\cdot)$ by Eq.~\eqref{eq:zero-eps-char}; hence it equals
Eq.~\eqref{eq:zero-S}.

\emph{Step 5: the transported constraint function.} Evaluating
Eq.~\eqref{eq:zero-eps-char} on the section $\phi=\nabla_{y}S$ and using
Eq.~\eqref{eq:H0-construction} gives Eq.~\eqref{eq:zero-eps}.

\emph{Step 6: the entropy production.} At $t=0$,
\begin{equation}\label{eq:A-g-vanishes}
A(0,\cdot)\,g=A(0,\cdot)\,\nabla H^{(0)}(0,\cdot)/\gamma_{0}=0
\end{equation}
by Eq.~\eqref{eq:degeneracy} and \eqref{eq:H0-construction}, so Eq.~\eqref{eq:entropy-prod} reads
$e(0,\cdot)=-\sigma$; and $\langle\sigma\rangle_{\rho}=0$ is
Eq.~\eqref{eq:mean-sigma-zero}. At later times the quadratic term of
Eq.~\eqref{eq:entropy-prod} is $\tfrac12g^{T}A(t,\cdot)g$ with $A(t,\cdot)$ of
Eq.~\eqref{eq:H2-solution}.

\emph{Step 7: the probability measure.} With $S_{0}\equiv0$, Lemma~\ref{lem:gibbs} gives
\begin{equation}\label{eq:zero-section-measure}
\mathcal{P}_{t}(\dd y)=\mathcal{Z}_{t}^{-1}e^{-ntH^{(0)}(t,y)}\,\dd^{n}y,
\end{equation}
which is Eq.~\eqref{eq:zero-measure} by Eq.~\eqref{eq:H0-construction}. The rate function is
Eq.~\eqref{eq:rate-I} with $S_{0}=0$, and for $\gamma_{0}>0$ the supremum of
$\varepsilon_{t}$ is attained where $\rho$ is minimal, since $\Phi^{(5)}_{-t}$ is a
bijection of $E$. This is the second identity of Eq.~\eqref{eq:zero-measure}.
\end{proof}

\subsection{Large Deviations as Special Case}\label{subsec:ldp}

The identity $\mathcal P_{t}=e^{-I_{t}}/\int_{E}e^{-I_{t}}$ of Lemma~\ref{lem:gibbs} is
exact at every finite $t$, so the large deviation principle is not a separate construction
but the concentration limit of that exact Gibbs form. Writing $I_{t}=tI+r_{t}$, the limit is
governed by the Laplace principle stated below, and $\|r_{t}\|_{\infty}=o(t)$ follows from
the single condition that the endpoint map converge in $C^{1}(E)$ to a diffeomorphism.

\subsubsection{Laplace Principle}

\begin{lemma}[Laplace principle; \cite{dembo1998large}]
\label{lem:laplace}
Let $X$ be compact, let $m$ be a finite Borel measure on $X$ assigning positive mass
to every non-empty open set, and let $f\in C^{0}(X)$. Then
\begin{equation}\label{eq:laplace}
\lim_{\beta\to\infty}\frac1{\beta}\ln\int_{X}e^{\beta f}\,\dd m=\sup_{X}f ,
\end{equation}
the same limit holds with $\sup_{G}f$ for every open $G\subset X$ with $m(G)>0$, and
\begin{equation}\label{eq:laplace-closed}
\limsup_{\beta\to\infty}\frac1{\beta}\ln\int_{F}e^{\beta f}\,\dd m\le\sup_{F}f
\qquad\text{for every closed }F\subset X,
\end{equation}
with the convention $\ln0=-\infty$.
\end{lemma}
\begin{proof}
\emph{Upper bounds.} For every $\beta>0$ and every Borel set $A\subset X$,
\begin{equation}\label{eq:laplace-upper}
\int_{A}e^{\beta f}\,\dd m\le m(A)\,e^{\beta\sup_{A}f},
\end{equation}
which gives the limsup inequality in Eq.~\eqref{eq:laplace-closed} and the limsup half of
Eq.~\eqref{eq:laplace}.

\emph{Lower bounds.} Let $A$ be $X$, or an open $G$ with $m(G)>0$, and let
$\delta>0$. The set $U_{\delta}:=\{x\in A: f(x)>\sup_{A}f-\delta\}$ is non-empty and open
in $X$, hence $m(U_{\delta})>0$, and
\begin{equation}\label{eq:laplace-lower}
\int_{A}e^{\beta f}\,\dd m\ge\int_{U_{\delta}}e^{\beta f}\,\dd m
\ge m(U_{\delta})\,e^{\beta(\sup_{A}f-\delta)},
\end{equation}
so $\liminf_{\beta}\beta^{-1}\ln\int_{A}e^{\beta f}\,\dd m\ge\sup_{A}f-\delta$ for every
$\delta>0$. This gives the liminf half of Eq.~\eqref{eq:laplace} and the corresponding
statement for $G$.

A closed set of $m$-measure zero gives $\ln0=-\infty$, so the identity cannot hold for
arbitrary closed $F$; Eq.~\eqref{eq:laplace-closed} is the sharp statement.
\end{proof}

\subsubsection{Endpoint Large Deviation Principle}

Under the additional endpoint-convergence condition, Theorem~\ref{thm:main} and the Laplace
principle yield the following corollary.

\begin{corollary}[Large deviation principle for the endpoint probability measure]\label{cor:ldp}
Let the objects and hypotheses of Section~\ref{sec:main} be given, let
Lemma~\ref{lem:gibbs} apply, and assume that
\begin{equation}\label{eq:main-condition}
T=\infty
\qquad\text{and}\qquad
y(t;\cdot)\longrightarrow y_\infty
\qquad\text{in }C^1(E),
\end{equation}
where $y_\infty$ is a diffeomorphism of $E$. Then
$\varepsilon_{t}\to\varepsilon_{\infty}$ in $C^{0}(E)$, where
$\varepsilon_{\infty}:=\varepsilon_{0}\circ y_{\infty}^{-1}$, and the family
$\{\mathcal{P}_{t}\}_{t>0}$ satisfies the large deviation principle at speed $t$ with
the good rate function
\begin{equation}\label{eq:ldp-rate}
I(y)=n\Bigl[\sup_{E}\varepsilon_{\infty}-\varepsilon_{\infty}(y)\Bigr]\ \ge\ 0 ,
\end{equation}
that is, for every closed $F\subset E$ and every open $G\subset E$,
\begin{align}
\liminf_{t\to\infty}\Bigl[-\frac1t\ln\mathcal{P}_{t}(F)\Bigr]&\ \ge\ \inf_{F}I,
\label{eq:ldp-closed}\\
\limsup_{t\to\infty}\Bigl[-\frac1t\ln\mathcal{P}_{t}(G)\Bigr]&\ \le\ \inf_{G}I .
\label{eq:ldp-open}
\end{align}
Moreover $I$ is continuous, vanishes exactly on
$y_{\infty}\bigl(\arg\max_{E}\varepsilon_{0}\bigr)$, and is the limit of the finite-time
rate function of Eq.~\eqref{eq:gibbs}: writing $I_{t}=tI+r_{t}$,
\begin{equation}\label{eq:rate-convergence}
\lim_{t\to\infty}\frac{\|r_{t}\|_{\infty}}{t}
=\lim_{t\to\infty}\Bigl\|\frac{I_{t}}{t}-I\Bigr\|_{\infty}=0 ,
\end{equation}
and the probability measure takes a Gibbs-form large deviation expression with formal
inverse-temperature parameter $t$ and effective energy $I$,
\begin{equation}\label{eq:ldp-measure}
\mathcal{P}_{t}(\dd y)
=\frac{e^{-tI(y)-r_{t}(y)}}{\int_{E}e^{-tI-r_{t}}\,\dd^{n}y}\,\dd^{n}y
=\frac{e^{-tI(y)}}{\int_{E}e^{-tI}\,\dd^{n}y}\;e^{\,o(t)}\,\dd^{n}y,
\end{equation}
the $o(t)$ being uniform in $y\in E$.
\end{corollary}
\begin{proof}
\emph{Step 1: the action.} By Lemma~\ref{lem:action-linear},
$S_{t}(y)=S_{0}(y_{0}(t,y))+t\varepsilon_{t}(y)$ with
$\varepsilon_{t}=\varepsilon_{0}\circ y_{0}(t,\cdot)$ and
$y_{0}(t,\cdot)=y(t;\cdot)^{-1}$. Since $\varepsilon_{0}$ is Lipschitz on the compact
$E$, the $C^{1}$ convergence of $y(t;\cdot)$ gives
\begin{equation}\label{eq:ldp-uniform}
\|\varepsilon_{t}-\varepsilon_{\infty}\|_{\infty}
\le\operatorname{Lip}(\varepsilon_{0})\,\|y_{0}(t,\cdot)-y_{\infty}^{-1}\|_{\infty}\to0 .
\end{equation}

\emph{Step 2: the finite-time rate function converges.}
Lemma~\ref{lem:oscillation}(1) and Eq.~\eqref{eq:ldp-uniform} give
$\sup_{E}\varepsilon_{0}=\sup_{E}\varepsilon_{\infty}$ and
\begin{equation}\label{eq:eps-converges}
\varepsilon_{0}\circ y_{0}(t,\cdot)=\varepsilon_{t}\to\varepsilon_{\infty}
\end{equation}
uniformly, so
Lemma~\ref{lem:rate-squeeze} yields
\begin{equation}\label{eq:ldp-conv}
\begin{split}
\Bigl\|\frac{I_{t}}{t}-I\Bigr\|_{\infty}
&\le n\bigl\|\varepsilon_{t}-\varepsilon_{\infty}\bigr\|_{\infty}
+\frac{2n\|S_{0}\|_{\infty}}{t}\ \longrightarrow\ 0 ,\\
I&:=n\Bigl[\sup_{E}\varepsilon_{\infty}-\varepsilon_{\infty}\Bigr],
\end{split}
\end{equation}
which is Eq.~\eqref{eq:ldp-rate} and \eqref{eq:rate-convergence}. Write
$\eta_{t}:=\|I_{t}/t-I\|_{\infty}$; then $\eta_{t}\to0$ and
\begin{equation}\label{eq:ldp-compare}
e^{-t\eta_{t}}\,e^{-tI}\ \le\ e^{-I_{t}}\ \le\ e^{\,t\eta_{t}}\,e^{-tI}
\qquad\text{on }E .
\end{equation}
Set $r_{t}:=I_{t}-tI$, so that $\|r_{t}\|_{\infty}=t\eta_{t}=o(t)$. The exact identity
$\mathcal{P}_{t}=e^{-I_{t}}/\int_{E}e^{-I_{t}}$ of Eq.~\eqref{eq:gibbs} then reads
\begin{equation}\label{eq:ldp-measure-proof}
\begin{split}
\mathcal{P}_{t}(\dd y)
&=\frac{e^{-tI(y)-r_{t}(y)}}{\int_{E}e^{-tI-r_{t}}\,\dd^{n}y}\,\dd^{n}y,\\
e^{-t\eta_{t}}\int_{E}e^{-tI}\,\dd^{n}y
&\le\int_{E}e^{-tI-r_{t}}\,\dd^{n}y
\le e^{\,t\eta_{t}}\int_{E}e^{-tI}\,\dd^{n}y,
\end{split}
\end{equation}
the second line being Eq.~\eqref{eq:ldp-compare} integrated over $E$; the two lines together
are Eq.~\eqref{eq:ldp-measure}, with a factor $e^{\,o(t)}$ uniform in $y$.

\emph{Step 3: the closed-set bound.} Let $F\subset E$ be closed. By the third form of
Eq.~\eqref{eq:gibbs} and by Eq.~\eqref{eq:ldp-compare},
\begin{equation}\label{eq:ldp-upper}
\begin{split}
\frac1t\ln\mathcal{P}_{t}(F)
&=\frac1t\ln\int_{F}e^{-I_{t}}\,\dd^{n}y-\frac1t\ln\int_{E}e^{-I_{t}}\,\dd^{n}y\\
&\le2\eta_{t}+\frac1t\ln\int_{F}e^{-tI}\,\dd^{n}y
-\frac1t\ln\int_{E}e^{-tI}\,\dd^{n}y .
\end{split}
\end{equation}
Lemma~\ref{lem:laplace} with $X=E$, $m=\dd^{n}y$ and $f=-I$ gives
\begin{equation}\label{eq:ldp-upper-final}
\limsup_{t\to\infty}\frac1t\ln\int_{F}e^{-tI}\,\dd^{n}y\ \le\ -\inf_{F}I,
\qquad
\lim_{t\to\infty}\frac1t\ln\int_{E}e^{-tI}\,\dd^{n}y\ =\ -\inf_{E}I=0 ,
\end{equation}
the first inequality being Eq.~\eqref{eq:laplace-closed}, the limit being Eq.~\eqref{eq:laplace},
and $\inf_{E}I=0$ because $\varepsilon_{\infty}$ attains its supremum on the compact $E$.
Inserting Eq.~\eqref{eq:ldp-upper-final} into Eq.~\eqref{eq:ldp-upper} and letting $t\to\infty$
gives Eq.~\eqref{eq:ldp-closed}.

\emph{Step 4: the open-set bound.} Let $G\subset E$ be open; then $\dd^{n}y(G)>0$. The
lower inequality of Eq.~\eqref{eq:ldp-compare} gives
\begin{equation}\label{eq:ldp-lower}
\frac1t\ln\mathcal{P}_{t}(G)
\ \ge\ -2\eta_{t}+\frac1t\ln\int_{G}e^{-tI}\,\dd^{n}y
-\frac1t\ln\int_{E}e^{-tI}\,\dd^{n}y ,
\end{equation}
and the open-set clause of Lemma~\ref{lem:laplace} gives
\begin{equation}\label{eq:laplace-open}
\lim_{t}\frac1t\ln\int_{G}e^{-tI}\,\dd^{n}y=-\inf_{G}I;
\end{equation}
together with
Eq.~\eqref{eq:ldp-upper-final} this is Eq.~\eqref{eq:ldp-open}.

\emph{Step 5: the rate function.} $I$ is continuous and non-negative on $E$ and vanishes
exactly where $\varepsilon_{\infty}=\sup_{E}\varepsilon_{\infty}$, that is on
$y_{\infty}\bigl(\arg\max_{E}\varepsilon_{0}\bigr)$; its level sets are closed subsets of
the compact $E$, hence compact, so $I$ is a good rate function.
\end{proof}

Corollary~\ref{cor:ldp} and Lemma~\ref{lem:oscillation} give the following
criterion for the non-degeneracy of the limit rate function.

\subsubsection{Non-degeneracy Criterion}

\begin{corollary}[Non-degeneracy criterion for the rate function]\label{cor:dichotomy}
\begin{enumerate}[label=(\arabic*),leftmargin=2.6em]
\item If $S_{0}\equiv0$ and $\gamma_{0}\neq0$, then the hypotheses of
Lemma~\ref{lem:endpoint-map} hold with $T=\infty$, and $\xi\equiv0$ and
$C_{\varepsilon}\equiv0$. The rate function
is Eq.~\eqref{eq:zero-measure} for every $t$. If in addition the endpoint map converges in
$C^{1}(E)$ to a diffeomorphism, Corollary~\ref{cor:ldp} applies with
$\varepsilon_{\infty}=-\gamma_{0}\ln(\rho_{\infty}/\rho_{0})$, where
$\rho_{\infty}:=\lim_{t}\rho\circ\Phi^{(5)}_{-t}$, and
\begin{equation}\label{eq:zero-rate}
I(y)=n\gamma_{0}\ln\frac{\rho_{\infty}(y)}{\min_{E}\rho},
\qquad
I\not\equiv0\ \Longleftrightarrow\ \rho\ \text{is not constant}.
\end{equation}
When $\sigma\equiv0$, Eq.~\eqref{eq:lnrho-derivative} gives $\varepsilon_{t}\equiv\varepsilon_{0}$,
hence $\rho_{\infty}=\rho$.
\item Let the hypotheses of Corollary~\ref{cor:ldp} hold, let $S_{0}\not\equiv0$, and let
$\varepsilon_{0}$ be non-constant. Then $\varepsilon_{\infty}$ is non-constant and the
rate function of Corollary~\ref{cor:ldp} is non-degenerate.
\end{enumerate}
Whenever the limit exists, $I\equiv0$ if and only if $\varepsilon_{0}$ is constant,
that is, if and only if
\begin{equation}\label{eq:degeneracy-criterion}
\tfrac12\bigl(\nabla S_{0}\bigr)^{T}\!A(0,\cdot)\,\nabla S_{0}
-\gamma_{0}\ln\frac{\rho}{\rho_{0}}\ \equiv\ \text{const},
\end{equation}
a condition on the initial section alone.
\end{corollary}
\begin{proof}
\emph{(1)} By Corollary~\ref{cor:zero-section}(1) the hypotheses of
Lemma~\ref{lem:endpoint-map} hold with $T=\infty$, and $\xi\equiv0$ and
$C_{\varepsilon}\equiv0$. The rate function is
Eq.~\eqref{eq:zero-measure} for every $t$ by Corollary~\ref{cor:zero-section}(5). Under the
$C^{1}$ convergence of the endpoint map, Corollary~\ref{cor:ldp} applies with
\begin{equation}\label{eq:zero-eps-infty}
\varepsilon_{\infty}(y)=-\gamma_{0}\ln\frac{\rho_{\infty}(y)}{\rho_{0}},
\end{equation}
so
\begin{equation}\label{eq:zero-rate-inf}
I=n\Bigl[\sup_{E}\varepsilon_{\infty}-\varepsilon_{\infty}\Bigr]
=n\gamma_{0}\ln\frac{\rho_{\infty}}{\min_{E}\rho_{\infty}},
\end{equation}
which is Eq.~\eqref{eq:zero-rate} because $\min_{E}\rho_{\infty}=\min_{E}\rho$ by
Lemma~\ref{lem:oscillation}. Moreover $I\equiv0$ if and only if
$\rho_{\infty}$ is constant, and by Lemma~\ref{lem:oscillation}(1) this holds if
and only if $\rho$ is constant.

\emph{(2)} By Lemma~\ref{lem:action-linear},
$\varepsilon_{\infty}=\varepsilon_{0}\circ y_{\infty}^{-1}$; the composition is
non-constant because $\operatorname{osc}(\varepsilon_{0})>0$ and $y_{\infty}^{-1}$ is a
bijection (Lemma~\ref{lem:oscillation}(1)). Hence $I\not\equiv0$ by
Eq.~\eqref{eq:rate-nonneg}.

\emph{The closing criterion.} By Eq.~\eqref{eq:eps-zero}, $\varepsilon_{0}$ is constant if
and only if Eq.~\eqref{eq:degeneracy-criterion} holds, a condition on the initial section
alone; by Eq.~\eqref{eq:rate-nonneg} this is equivalent to $I\equiv0$.
\end{proof}

\section{Entropy Production and Fluctuation Duality}\label{sec:applications}

The contact probability measure of Theorem~\ref{thm:contact-measure} and the endpoint large deviation
principle of Corollary~\ref{cor:ldp} are now applied to the entropy production rate $e$ of
Eq.~\eqref{eq:entropy-prod}. The tilted contact potential and the characteristic path
measure give the large deviation principle for $e$, and the time-reversal involution gives
the fluctuation duality between the forward and the reversed ensembles.

\subsection{Entropy Production and Tilted Potential}\label{subsec:scgf}

The observable of the theory is the entropy production rate $e$ of
Eq.~\eqref{eq:entropy-prod}; at every time its mean over the invariant measure is
$\tfrac12\langle g^{T}A(t,\cdot)g\rangle_{\rho}$, hence non-negative.

\begin{lemma}[Mean entropy production]\label{lem:e-mean}
The Liouville identity $\nabla\!\cdot(\rho v)=0$ reads $v\cdot g=\sigma$, and for every
$t\ge0$,
\begin{equation}\label{eq:sigma-mean-zero}
\langle\sigma\rangle_{\rho}=0,
\qquad
\langle e(t,\cdot)\rangle_{\rho}
=\tfrac12\bigl\langle g^{T}A(t,\cdot)g\bigr\rangle_{\rho}\ \ge\ 0,
\end{equation}
with equality in the second relation if and only if $A(t,\cdot)g=0$ at $\rho$-almost every
point. At $t=0$ equality holds at every point, by the degeneracy condition
Eq.~\eqref{eq:degeneracy} with $\nabla H^{(0)}(0,\cdot)=\gamma_{0}g$; in the divergence-free
subclass $\sigma\equiv0$ the entropy potential is stationary, $\nabla H^{(0)}(t,\cdot)\equiv\gamma_{0}g$,
the degeneracy condition gives $A(t,\cdot)g\equiv0$ for every $t$, and $e\equiv0$.
\end{lemma}
\begin{proof}
\emph{Step 1: the pointwise identity.} Since $g=\nabla\ln\rho=\rho^{-1}\nabla\rho$ and
$\nabla\!\cdot(\rho v)=0$,
\begin{equation}\label{eq:mean-g}
v\cdot g=\frac1\rho\,v\cdot\nabla\rho=-\nabla\!\cdot v=\sigma .
\end{equation}

\emph{Step 2: vanishing mean.} By Eq.~\eqref{eq:mean-sigma-zero}, obtained by
integrating $v\cdot g=\sigma$ against $\rho\,\dd^{n}y$ on the compact boundaryless
manifold $E$, $\langle\sigma\rangle_{\rho}=0$.

\emph{Step 3: the mean entropy production.} By Eq.~\eqref{eq:entropy-prod} and Step 2,
\begin{equation}\label{eq:mean-e}
\langle e(t,\cdot)\rangle_{\rho}
=\tfrac12\bigl\langle g^{T}A(t,\cdot)g\bigr\rangle_{\rho}
=\tfrac12\int_{E}\bigl|A(t,\cdot)^{1/2}g\bigr|^{2}\rho\,\dd^{n}y\ \ge\ 0,
\end{equation}
and the integrand is non-negative because $A(t,\cdot)$ is positive semidefinite. It
vanishes $\rho$-almost everywhere if and only if $A(t,\cdot)^{1/2}g=0$, that is
$A(t,\cdot)g=0$, $\rho$-almost everywhere.

\emph{Step 4: the initial time and the divergence-free subclass.} At $t=0$ the
degeneracy condition Eq.~\eqref{eq:degeneracy} reads $A(0,\cdot)\,\nabla H^{(0)}(0,\cdot)=0$
with $\nabla H^{(0)}(0,\cdot)=\gamma_{0}g$ by Eq.~\eqref{eq:H0-construction}, hence
$A(0,\cdot)g=0$ everywhere and $e(0,\cdot)=-\sigma$. When $\sigma\equiv0$ the integral in
Eq.~\eqref{eq:H0-construction} vanishes, $H^{(0)}(t,\cdot)=\gamma_{0}\ln(\rho/\rho_{0})$
for every $t$, the degeneracy condition gives $A(t,\cdot)g=0$ for every $t$, and
$e\equiv0$.
\end{proof}

\begin{definition}[Tilted contact potential]\label{def:tilted}
For $q\in\R$ the tilted contact potential, its constraint function and its tilted
action are
\begin{equation}\label{eq:Hq}
\begin{split}
&H_{q}(t,y,\phi):=H(t,y,\phi)-q\,e(t,y),
\qquad
\varepsilon_{q}:=\DD[H_{q}]=\varepsilon+q\,e,\\
&S_{q}[\gamma]:=\int_{\gamma}\varepsilon_{q}\,\dd s
=S[\gamma]+q\!\int_{\gamma}e(s,\gamma(s))\,\dd s ,
\end{split}
\end{equation}
the last two identities holding for every curve $\gamma$, whether or not $\gamma$ is an
integral curve of $X_{H}$ or of $X_{H_{q}}$. The identity $\varepsilon_{q}=\varepsilon+qe$
uses $\DD[qe]=-qe$, which holds because $e$ is independent of $\phi$.
\end{definition}

\begin{lemma}[The tilt breaks the conservation of the constraint function]
\label{lem:tilt-obstruction}
The tilted potential $H_{q}$ satisfies the master
equation $\DD[\partial_{t}H_{q}]+\{\DD[H_{q}],H_{q}\}=0$ identically in $(t,y,\phi)$
if and only if
\begin{equation}\label{eq:tilt-criterion}
\partial_{t}e+v\cdot\nabla e=0
\qquad\text{and}\qquad
A\,\nabla e=0 .
\end{equation}
Along an integral curve of $X_{H_{q}}$ the constraint function satisfies the exact identity
\begin{equation}\label{eq:tilt-evolution}
\frac{\dd\varepsilon_{q}}{\dd s}
=\partial_{t}\varepsilon_{q}+\{\varepsilon_{q},H_{q}\}
=q\bigl(\partial_{t}e+v\cdot\nabla e+2\,\phi^{T}A\,\nabla e\bigr).
\end{equation}
Consequently $\varepsilon_{q}$ is not a first integral of $X_{H_{q}}$ in general, the
linearity Eq.~\eqref{eq:action-linear} fails for the tilted action, and $S_{q}(t,y)$ must
be computed as the integral in Eq.~\eqref{eq:Hq} along the integral curve of
$X_{H_{q}}$.
\end{lemma}
\begin{proof}
\emph{Step 1: the $q$-linear part of the tilted master equation.} Since $e$ does not
depend on $\phi$, $\DD[e]=-e$, hence $\DD[H_{q}]=\varepsilon+qe$ and
$\partial_{t}H_{q}=\partial_{t}H-q\,\partial_{t}e$, so that
\begin{equation}\label{eq:tilt-time-derivative}
\DD[\partial_{t}H_{q}]=\partial_{t}\varepsilon+q\,\partial_{t}e.
\end{equation}
Therefore
\begin{equation}\label{eq:tilt-master}
\DD[\partial_{t}H_{q}]+\{\DD[H_{q}],H_{q}\}
=\bigl[\partial_{t}\varepsilon+\{\varepsilon,H\}\bigr]
+q\bigl(\partial_{t}e+\{e,H\}-\{\varepsilon,e\}\bigr),
\end{equation}
the bracket multiplying $1$ vanishing because $H$ satisfies the master equation
Eq.~\eqref{eq:master-equation}.

\emph{Step 2: the two brackets.} By the definition Eq.~\eqref{eq:poisson-bracket}, and
because $e$ is independent of $\phi$,
\begin{equation}\label{eq:tilt-brackets}
\{e,H\}=(v+A\phi)\cdot\nabla e,
\qquad
\{\varepsilon,e\}=-(A\phi)\cdot\nabla e,
\end{equation}
hence
\begin{equation}\label{eq:tilt-criterion-q}
\{e,H\}-\{\varepsilon,e\}=(v+2A\phi)\cdot\nabla e
=v\cdot\nabla e+2\,\phi^{T}A\,\nabla e .
\end{equation}

\emph{Step 3: the criterion.} The master equation Eq.~\eqref{eq:tilt-master} holds
identically in $(t,y,\phi)$ if and only if $\partial_{t}e$ plus
Eq.~\eqref{eq:tilt-criterion-q} vanishes identically in $\phi$. The two terms of
Eq.~\eqref{eq:tilt-criterion-q} are homogeneous of degree $0$ and $1$ in $\phi$, and
$\partial_{t}e$ is of degree $0$, so this is equivalent to
\begin{equation}\label{eq:tilt-criterion-repeated}
\partial_{t}e+v\cdot\nabla e=0
\qquad\text{and}\qquad
\phi^{T}A\,\nabla e=0\ \text{for every }\phi,
\end{equation}
that is $\partial_{t}e+v\cdot\nabla e=0$ and $A\nabla e=0$ (the matrix $A$ being
symmetric). This is Eq.~\eqref{eq:tilt-criterion}. Applying Eq.~\eqref{eq:epsilon_evolution} to
$H_{q}$ and substituting Step 1 gives Eq.~\eqref{eq:tilt-evolution}. When the criterion fails,
$\varepsilon_{q}$ is not a first integral of $X_{H_{q}}$, and the linearity of
Lemma~\ref{lem:action-linear} does not apply to the tilted action.
\end{proof}

\begin{lemma}[Symmetry of the tilted potential]\label{lem:symmetry}
Let $J:(t,y,\phi)\mapsto(t,y,-\phi-g)$ and
$v^{\mathrm{rev}}(t,y):=-v(t,y)+A(t,y)\,g(y)$, and let
$H^{\mathrm{rev}}$ be obtained from $H$ by replacing $v$ with $v^{\mathrm{rev}}$,
$H^{(0)}$ and $A$ unchanged. Then for every $q\in\R$,
\begin{equation}\label{eq:symmetry}
H_{q}(t,y,-\phi-g)=H^{\mathrm{rev}}_{q-1}(t,y,\phi),
\qquad
\varepsilon^{\mathrm{rev}}(t,y,\phi)=\varepsilon(t,y,\phi).
\end{equation}
\end{lemma}
\begin{proof}
\emph{Step 1: the left-hand side.} Substituting $\phi\mapsto-\phi-g$ in $H_{q}$,
\begin{equation}\label{eq:symmetry-expand}
\begin{aligned}
H_{q}(t,y,-\phi-g)
&=H^{(0)}-v\cdot\phi-v\cdot g+\tfrac12(\phi+g)^{T}A(\phi+g)-qe\\
&=H^{(0)}-v\cdot\phi+\tfrac12\phi^{T}A\phi+\phi^{T}Ag
+\bigl(\tfrac12g^{T}Ag-v\cdot g-qe\bigr) .
\end{aligned}
\end{equation}

\emph{Step 2: the right-hand side.} By definition,
\begin{equation}\label{eq:symmetry-rev}
H^{\mathrm{rev}}_{q-1}(t,y,\phi)
=H^{(0)}-v\cdot\phi+\phi^{T}Ag+\tfrac12\phi^{T}A\phi-(q-1)e ,
\end{equation}
in which the term $\phi^{T}Ag$ is the one of Eq.~\eqref{eq:symmetry-expand}. Subtracting
Eq.~\eqref{eq:symmetry-rev} from Eq.~\eqref{eq:symmetry-expand},
\begin{equation}\label{eq:symmetry-difference}
H_{q}(t,y,-\phi-g)-H^{\mathrm{rev}}_{q-1}(t,y,\phi)
=\tfrac12g^{T}Ag-v\cdot g-e=0 ,
\end{equation}
the last equality being Eq.~\eqref{eq:entropy-prod} together with Eq.~\eqref{eq:mean-g},
$e=\tfrac12g^{T}Ag-\sigma$ and $\sigma=v\cdot g$; this holds for every symmetric $A$,
degenerate or not, and is the first identity of Eq.~\eqref{eq:symmetry}.

\emph{Step 3: the constraint function.} By Eq.~\eqref{eq:constraint-quadratic},
$\varepsilon=-H^{(0)}+\tfrac12\phi^{T}A\phi$ contains no $v$, and $H^{\mathrm{rev}}$
differs from $H$ only in its first-order term; hence
$\varepsilon^{\mathrm{rev}}(t,y,\phi)=\varepsilon(t,y,\phi)$, the second identity of
Eq.~\eqref{eq:symmetry}.
\end{proof}

The second identity in Eq.~\eqref{eq:symmetry} states that the constraint function, and hence
the slope of the action, is common to the forward and reversed ensembles. It does not make
the endpoint Gibbs-form density time-reversal invariant, because the endpoint maps generated by
$v$ and $v^{\mathrm{rev}}$ are in general different. The first identity is an identity of
\emph{potentials}, and $J$ acts on the fibre by $\phi\mapsto-\phi-g$, which does not leave
$\varepsilon$ invariant. The duality is therefore a statement about the contact potential
and its characteristics, and the object it relates is the growth rate of the tilted
normalisation, not the endpoint probability measure itself.

\subsection{Large Deviations of Entropy Production}\label{subsec:entropy-ldp}

\begin{definition}[Characteristic path measure]\label{def:path-measure}
For $t>0$ let
\begin{equation}\label{eq:path-measure}
\mathbb{P}_{t}:=\bigl(\gamma_{t,\bullet}\bigr)_{\#}\,\mathcal{P}_{t}
\end{equation}
be the push-forward of the endpoint probability measure Eq.~\eqref{eq:gibbs} under the map
$y\mapsto\gamma_{t,y}\in C([0,t];E)$ of Lemma~\ref{lem:endpoint-map}, and let
$\mathbb{E}_{t}$ denote expectation with respect to $\mathbb{P}_{t}$. The
time-averaged entropy production is the function on $C([0,t];E)$
\begin{equation}\label{eq:et-def}
n\bar e_{t}(\gamma):=\frac{n}{t}\int_{0}^{t}e\bigl(s,\gamma(s)\bigr)\,\dd s ,
\end{equation}
and its pull-back to $E$ is $n\bar e_{t}\bigl(\gamma_{t,y}\bigr)$, a Borel
function of $y$. Throughout this section the entropy production is assumed bounded on
$[0,\infty)\times E$, which holds in particular whenever $A$ is bounded there.
\end{definition}
The push-forward Eq.~\eqref{eq:path-measure} is a Borel probability measure: smoothness of
$\Phi^{X_{H}}$ makes $(s,y)\mapsto\gamma_{t,y}(s)$ smooth, hence $y\mapsto\gamma_{t,y}$
is continuous from $E$ into $C([0,t];E)$ with the uniform topology, hence Borel. Under
the boundedness assumption of Definition~\ref{def:path-measure},
$|n\bar e_{t}|\le n\sup_{[0,\infty)\times E}|e|$ for every $t$ and
every $\gamma$, so $n\bar e_{t}$ is bounded.

\begin{corollary}[Large deviation principle for the entropy production]\label{cor:ge}
Let the hypotheses of Corollary~\ref{cor:ldp} hold, let
$\bar e_{t}(y):=\bar e_{t}(\gamma_{t,y})$ be the pull-back of Eq.~\eqref{eq:et-def}, and
assume that $\bar e_{t}\to\bar e_{\infty}$ in $C^{0}(E)$. Then the laws on $\mathbb R$ of
the random variables $n\bar e_{t}$ satisfy the large deviation principle at speed $t$
with the good rate function
\begin{equation}\label{eq:ge-rate}
J(s):=\inf\bigl\{I(y):\ \bar e_{\infty}(y)=s\bigr\},
\qquad
I(y)=n\Bigl[\sup_{E}\varepsilon_{\infty}-\varepsilon_{\infty}(y)\Bigr],
\end{equation}
and $J(s)=0$ if and only if $s=\bar e_{\infty}(y)$ for some
$y\in\arg\max_{E}\varepsilon_{\infty}$.
\end{corollary}
\begin{proof}
\emph{Step 1: the push-forward.} By Definition~\ref{def:path-measure},
$\mathbb{P}_{t}=(\gamma_{t,\cdot})_{\#}\mathcal{P}_{t}$; hence for every Borel
$F\subset\R$,
\begin{equation}\label{eq:ge-push}
\mathbb{P}_{t}\bigl(n\bar e_{t}\in F\bigr)
=\mathcal{P}_{t}\bigl(n\bar e_{t}(\cdot)\in F\bigr).
\end{equation}

\emph{Step 2: the upper bound.} Let $F\subset\R$ be closed and $\delta>0$. Uniform
convergence gives $\{n\bar e_{t}\in F\}\subset\{n\bar e_{\infty}\in F_{\delta}\}$ for all
large $t$, and the latter set is closed; by Step 1 and Corollary~\ref{cor:ldp},
\begin{equation}\label{eq:ge-upper}
\limsup_{t\to\infty}\frac1t\ln\mathbb{P}_{t}\bigl(n\bar e_{t}\in F\bigr)
\le-\inf_{F_{\delta}}I\ \xrightarrow[\delta\downarrow0]{}\ -\inf_{F}J ,
\end{equation}
the limit using the continuity of $I$ on the compact $E$ and
$\inf_{F}J=\inf\{I(y):\bar e_{\infty}(y)\in F\}$.

\emph{Step 3: the lower bound.} Let $G\subset\R$ be open and $\delta>0$. Uniform
convergence gives $\{n\bar e_{t}\in G\}\supset\{n\bar e_{\infty}\in G_{-\delta}\}$ for all
large $t$, and the latter set is open; hence
\begin{equation}\label{eq:ge-lower}
\liminf_{t\to\infty}\frac1t\ln\mathbb{P}_{t}\bigl(n\bar e_{t}\in G\bigr)
\ge-\inf_{G_{-\delta}}I\ \xrightarrow[\delta\downarrow0]{}\ -\inf_{G}J .
\end{equation}
Steps 2 and 3 are the large deviation principle with rate $J$.

\emph{Step 4: the rate function.} For $c\in\R$ the level set
$\{s:J(s)\le c\}=\bar e_{\infty}(\{y:I(y)\le c\})$ is compact, being the continuous image
of a compact set; hence $J$ is a good rate function. Finally $J(s)=0$ if and only if
there is $y$ with $\bar e_{\infty}(y)=s$ and $I(y)=0$, that is, with
$y\in\arg\max_{E}\varepsilon_{\infty}$.
\end{proof}

Set $\delta_{t}:=\|n\bar e_{t}-n\bar e_{\infty}\|_{\infty}$, which converges to zero by
hypothesis. The pointwise bounds $n\bar e_{\infty}-\delta_{t}\le n\bar e_{t}\le
n\bar e_{\infty}+\delta_{t}$ give
$e^{-|q|t\delta_{t}}\le\mathbb E_{t}[e^{qnt\bar e_{t}}]/\mathbb E_{t}[e^{qnt\bar e_{\infty}}]
\le e^{|q|t\delta_{t}}$, so the two scaled cumulant generating functions coincide. Since
$\bar e_{\infty}$ is bounded and continuous on $E$, Varadhan's lemma \cite{dembo1998large}
applies to the fixed observable $qn\bar e_{\infty}$, and the scaled cumulant generating
function of $n\bar e_{t}$ is the Legendre transform of $J$:
\begin{equation}\label{eq:lambda-def}
\lambda(q):=\lim_{t\to\infty}\frac1t\ln\mathbb{E}_{t}\bigl[e^{q\,n t\bar e_{t}}\bigr]
=\sup_{s\in\R}\bigl[q\,s-J(s)\bigr],
\end{equation}
the limit existing and being finite for every $q\in\R$.

The rate function Eq.~\eqref{eq:ge-rate} is non-trivial only when $\bar e_{\infty}$ is
non-constant. If $\bar e_{\infty}\equiv c$, then $J(c)=0$ and $J(s)=+\infty$ for $s\neq c$;
in particular $e\equiv0$ gives the point mass at zero. A non-trivial application therefore
requires both the hypotheses of Corollary~\ref{cor:ldp} and a non-constant limit of the
time-averaged observable.

\subsection{Fluctuation Duality}\label{subsec:duality}

The degeneracy Eq.~\eqref{eq:degeneracy} fixes the geometry of the fibre at each time:
$A(t,\cdot)$ is positive semidefinite and $\nabla H^{(0)}(t,\cdot)$ lies in its kernel, so
$\tfrac12\phi^{T}A\phi$ is constant along $\ker A(t,\cdot)$ and the tilted cell equation
below constrains only the $\operatorname{im}A_{\infty}$ component of $\nabla W$. The
potential $H^{(0)}$ and the tensor $A$ are time-dependent in general, so the cell problem
is posed for their limits: in addition to the hypothesis of Corollary~\ref{cor:ldp}, assume
independently that $A(t,\cdot)\to A_{\infty}$, $e(t,\cdot)\to e_{\infty}$ and
$H^{(0)}(t,\cdot)\to H^{(0)}_{\infty}$ in $C^{0}(E)$, where
\begin{equation}\label{eq:Hinf}
H^{(0)}_{\infty}:=\gamma_{0}\ln\frac{\rho_{\infty}}{\rho_{0}},
\qquad
H^{\infty}_{q}(y,\phi):=H^{(0)}_{\infty}(y)+v(y)\cdot\phi
+\tfrac12\phi^{T}A_{\infty}(y)\,\phi-q\,e_{\infty}(y),
\end{equation}
and the critical value of the corresponding cell equation is
\begin{equation}\label{eq:nu-def}
\nu(q):=\inf\bigl\{\mu\in\R:\ \exists\,W\in C^{1}(E)\ \text{with}\
H^{\infty}_{q}\bigl(y,\nabla W(y)\bigr)=\mu\ \text{for every }y\in E\bigr\};
\end{equation}
$\nu_{\mathrm{rev}}(q)$ is defined in the same way with $H^{\infty}_{q}$ replaced by
$H^{\mathrm{rev},\infty}_{q}$. If the admissible set is empty, the infimum is understood
as $+\infty$. Whenever a cell value enters the statements below it is assumed finite, and
the constant appearing in Eq.~\eqref{eq:separation} is the value defined by Eq.~\eqref{eq:nu-def}.

\begin{definition}[The reversed system]\label{def:reversed}
The reversed contact potential is
\begin{equation}\label{eq:Hrev}
H^{\mathrm{rev}}(t,y,\phi):=H^{(0)}(t,y)+v^{\mathrm{rev}}(t,y)\cdot\phi
+\tfrac12\,\phi^{T}A(t,y)\,\phi,
\qquad
v^{\mathrm{rev}}:=-v+Ag ,
\end{equation}
with the same $H^{(0)}$ and the same $A$ as $H$, the products $Ag$ and $v$ evaluated at
$(t,y)$. Its tilted potentials
$H^{\mathrm{rev}}_{q}:=H^{\mathrm{rev}}-q\,e$, its constraint functions
$\varepsilon^{\mathrm{rev}}_{q}$, its actions $S^{\mathrm{rev}}_{q}(t,y)$, its
normalisations $\mathcal{Z}^{\mathrm{rev}}_{t}(q)$ and its growth rates
$\lambda^{Z}_{\mathrm{rev}}(q)$ are defined by
Definitions~\ref{def:tilted} and Eq.~\eqref{eq:lambdaZ-def} with $H$ replaced by
$H^{\mathrm{rev}}$; and its critical value is $\nu_{\mathrm{rev}}(q)$ of Eq.~\eqref{eq:nu-def}. The forward
quantities are written $\nu_{\mathrm{fwd}}(q)=\nu(q)$ and
$\lambda^{Z}_{\mathrm{fwd}}(q)=\lambda^{Z}(q)$.
\end{definition}

\begin{lemma}[Gauge character of the tilt]\label{lem:gauge}
For every $q\in\R$, every $t\ge0$ and every $W\in C^{1}(E)$,
\begin{equation}\label{eq:gauge-identity}
H_{q}(t,y,\nabla W-qg)-H(t,y,\nabla W)
=-q\,\bigl(A(t,y)\,g(y)\bigr)\cdot\Bigl(\nabla W-\tfrac{q-1}{2}\,g(y)\Bigr) .
\end{equation}
Hence the fibre translation $\phi\mapsto\phi-qg$ relates the cell equation at $q$ to the
cell equation at $0$ if and only if $A(t,\cdot)\,g\equiv0$; in that case $\nu(q)=\nu(0)$
for every $q$.
\end{lemma}
\begin{proof}
By Eq.~\eqref{eq:Hq},
\begin{equation}\label{eq:gauge-expand}
\begin{aligned}
H_{q}(t,y,\nabla W-qg)
&=H^{(0)}+v\cdot(\nabla W-qg)+\tfrac12(\nabla W-qg)^{T}A\,(\nabla W-qg)-q\,e\\
&=H(t,y,\nabla W)-q\,(v\cdot g+e)-q\,(Ag)\cdot\nabla W+\tfrac12 q^{2}\,g^{T}Ag .
\end{aligned}
\end{equation}
By Eq.~\eqref{eq:entropy-prod} and \eqref{eq:mean-g},
\begin{equation}\label{eq:gauge-mean}
v\cdot g+e=\tfrac12\,g^{T}Ag ,
\end{equation}
and substituting Eq.~\eqref{eq:gauge-mean} into Eq.~\eqref{eq:gauge-expand} yields
\begin{equation}\label{eq:gauge-final}
\begin{split}
H_{q}(t,y,\nabla W-qg)-H(t,y,\nabla W)
&=-q\,(Ag)\cdot\nabla W+\tfrac12 q(q-1)\,g^{T}Ag\\
&=-q\,(Ag)\cdot\Bigl(\nabla W-\tfrac{q-1}{2}\,g\Bigr),
\end{split}
\end{equation}
which is Eq.~\eqref{eq:gauge-identity}. The right-hand side vanishes for every
$W\in C^{1}(E)$ if and only if $A(t,\cdot)\,g\equiv0$; in that case
$W\mapsto W-q\ln\rho$ is a bijection between the admissible functions of the two cell
equations at the limit potential, preserving the constant, hence $\nu(q)=\nu(0)$.
\end{proof}

The tilted growth rate is the growth rate of the tilted normalisation, and is distinct
from the cumulant generating function $\lambda$ of Eq.~\eqref{eq:lambda-def}:
\begin{equation}\label{eq:lambdaZ-def}
\lambda^{Z}(q):=\lim_{t\to\infty}\frac1t\ln\frac{\mathcal{Z}_{t}(q)}{\mathcal{Z}_{t}(0)},
\qquad
\mathcal{Z}_{t}(q):=\int_{E}e^{\,nS_{q}(t,y)}\,\dd^{n}y ,
\end{equation}
with $S_{q}(t,y)$ the tilted action of Eq.~\eqref{eq:Hq} along the integral curve of
$X_{H_{q}}$ ending at $(t,y)$.

\begin{lemma}[Normalisation of the tilted partition function]
\label{lem:tilted-normalisation}
Let $q$ be such that Eq.~\eqref{eq:nu-def} is finite and, for some $W_{q}\in C^{0}(E)$,
\begin{equation}\label{eq:separation}
S_{q}(t,y)=W_{q}(y)-\nu(q)\,t+r_{q}(t,y),
\qquad
\sup_{t>0}\,\bigl\|r_{q}(t,\cdot)\bigr\|_{C^{0}(E)}<\infty .
\end{equation}
Then $W_{q}$ is bounded and there is a constant $C_{q}$, independent of $t$, with
\begin{equation}\label{eq:tilted-lnZ}
\bigl|\ln\mathcal{Z}_{t}(q)+n\,\nu(q)\,t\bigr|\ \le\ C_{q} .
\end{equation}
\end{lemma}
\begin{proof}
\emph{Step 1: boundedness.} $W_{q}\in C^{0}(E)$ with $E$ compact gives
$\|W_{q}\|_{\infty}<\infty$; by the hypothesis there is $R_{q}<\infty$ with
$\|r_{q}(t,\cdot)\|_{\infty}\le R_{q}$ for every $t>0$.

\emph{Step 2: factorisation.} By Eq.~\eqref{eq:separation},
\begin{equation}\label{eq:Z-factorised}
\mathcal{Z}_{t}(q)=\int_{E}e^{nS_{q}(t,y)}\,\dd^{n}y
=e^{-n\nu(q)t}\int_{E}e^{nW_{q}(y)+nr_{q}(t,y)}\,\dd^{n}y .
\end{equation}

\emph{Step 3: the estimate.} The integrand is bounded by
$e^{n(\inf_{E}W_{q}-R_{q})}$ from below and $e^{n(\sup_{E}W_{q}+R_{q})}$ from above, so
\begin{equation}\label{eq:Z-bounds}
|E|\,e^{n(\inf_{E}W_{q}-R_{q})}
\le\int_{E}e^{nW_{q}+nr_{q}}\,\dd^{n}y
\le|E|\,e^{n(\sup_{E}W_{q}+R_{q})} .
\end{equation}
Taking logarithms in Eq.~\eqref{eq:Z-factorised} and using Eq.~\eqref{eq:Z-bounds},
\begin{equation}\label{eq:Z-log-bound}
\bigl|\ln\mathcal{Z}_{t}(q)+n\nu(q)t\bigr|
\le\ln|E|+n\bigl(\|W_{q}\|_{\infty}+R_{q}\bigr)=:C_{q} .
\end{equation}
This is Eq.~\eqref{eq:tilted-lnZ}.
\end{proof}

\begin{remark}[Contact pressure]\label{rem:pressure}
Equation~\eqref{eq:tilted-lnZ} gives
\begin{equation}\label{eq:pressure-value}
\lim_{t\to\infty}\Bigl[-\frac{1}{nt}\ln\mathcal Z_{t}(q)\Bigr]=\nu(q),
\end{equation}
so $\nu(q)$ is the growth rate per unit time of the tilted normalisation, that is, the
contact pressure of the tilted ensemble.
\end{remark}

\begin{remark}[Pressure without a cell solution]\label{rem:pressure-general}
Suppose that
\begin{equation}\label{eq:growth-profile}
\bar s_{q}(y):=\lim_{t\to\infty}\frac{S_{q}(t,y)}{t}
\end{equation}
exists in $C^{0}(E)$. Laplace's principle on the compact manifold $E$ then gives
\begin{equation}\label{eq:pressure-general}
\lim_{t\to\infty}\frac1t\ln\mathcal Z_{t}(q)=n\max_{E}\bar s_{q},
\qquad
\nu_{P}(q):=-\max_{E}\bar s_{q},
\end{equation}
and consequently
\begin{equation}\label{eq:pressure-lambda}
\lambda^{Z}(q)
=n\bigl[\max_{E}\bar s_{q}-\max_{E}\bar s_{0}\bigr]
=-n\bigl[\nu_{P}(q)-\nu_{P}(0)\bigr].
\end{equation}
No solvability of Eq.~\eqref{eq:nu-def} is required. When Eq.~\eqref{eq:separation} holds,
$\bar s_{q}=-\nu(q)$ and $\nu_{P}(q)=\nu(q)$. A shift between the forward and reversed
pressures $\nu_{P}$ is not asserted here; it requires a separate proof.
\end{remark}

\begin{lemma}[Involution of the cell equations]\label{lem:involution}
For every $q\in\R$,
\begin{equation}\label{eq:mu-involution}
\nu_{\mathrm{fwd}}(q)=\nu_{\mathrm{rev}}(q-1),
\end{equation}
with $\nu_{\mathrm{fwd}}=\nu$ and $\nu_{\mathrm{rev}}$ as in Eq.~\eqref{eq:nu-def}.
\end{lemma}
\begin{proof}
Let $W\in C^{1}(E)$ satisfy $H^{\infty}_{q}(y,\nabla W(y))=\mu$ on $E$, and set
$W^{\mathrm{rev}}:=-W-\ln\rho$. Then $\nabla W^{\mathrm{rev}}=-\nabla W-g$, and
Lemma~\ref{lem:symmetry} gives, with $H^{(0)}_{\infty}$, $A_{\infty}$ and $e_{\infty}$ in
place of $H^{(0)}$, $A$ and $e$,
\begin{equation}\label{eq:involution-cell}
H^{\mathrm{rev},\infty}_{q-1}\bigl(y,\nabla W^{\mathrm{rev}}(y)\bigr)
=H^{\infty}_{q}\bigl(y,\nabla W(y)\bigr)=\mu .
\end{equation}
The map $W\mapsto W^{\mathrm{rev}}$ is an involution, since
$-( -W-\ln\rho)-\ln\rho=W$, so it is a bijection between the solutions of the forward
cell equation with constant $\mu$ and those of the reversed cell equation with the same
constant. The two admissible sets coincide, hence so do their infima, which is
Eq.~\eqref{eq:mu-involution}.
\end{proof}

The tilted-normalisation and involution lemmas yield the following duality under the
separation hypotheses.

\begin{corollary}[Fluctuation duality]\label{cor:duality}
Assume that all cell values used below are finite. For the forward system, assume
Eq.~\eqref{eq:separation} at $q$ and $0$; for the reversed system, assume it at $q-1$ and
$-1$, in each case with the constant given by Eq.~\eqref{eq:nu-def}. Then
\begin{equation}\label{eq:duality}
\lambda^{Z}_{\mathrm{fwd}}(q)=-n\bigl[\nu_{\mathrm{fwd}}(q)-\nu_{\mathrm{fwd}}(0)\bigr],
\end{equation}
and, when the same hypotheses hold with $q$ replaced by $q+1$,
\begin{equation}\label{eq:duality-final}
\lambda^{Z}_{\mathrm{fwd}}(q)=\lambda^{Z}_{\mathrm{rev}}(q-1)-\lambda^{Z}_{\mathrm{rev}}(-1),
\qquad\text{equivalently}\qquad
\lambda^{Z}_{\mathrm{rev}}(q)=\lambda^{Z}_{\mathrm{fwd}}(q+1)+\lambda^{Z}_{\mathrm{rev}}(-1).
\end{equation}
\end{corollary}
\begin{proof}
\emph{Step 1: the growth rate.} By Eq.~\eqref{eq:separation} and
Lemma~\ref{lem:tilted-normalisation}, applied to $q$ and to $0$, there are constants
$C_{q},C_{0}$ with
\begin{equation}\label{eq:duality-rate}
-n\bigl[\nu(q)-\nu(0)\bigr]t-(C_{q}+C_{0})
\le\ln\frac{\mathcal{Z}_{t}(q)}{\mathcal{Z}_{t}(0)}
\le-n\bigl[\nu(q)-\nu(0)\bigr]t+(C_{q}+C_{0}),
\end{equation}
so the quotient in Eq.~\eqref{eq:lambdaZ-def} converges and
\begin{equation}\label{eq:duality-lambda}
\lambda^{Z}(q)=-n\bigl[\nu(q)-\nu(0)\bigr],
\end{equation}
which is Eq.~\eqref{eq:duality}.

\emph{Step 2: the shift.} By Lemma~\ref{lem:involution},
$\nu_{\mathrm{fwd}}(q)=\nu_{\mathrm{rev}}(q-1)$.

\emph{Step 3: the relation.} Combining Eq.~\eqref{eq:duality-lambda} and
Lemma~\ref{lem:involution},
\begin{equation}\label{eq:duality-chain}
\begin{aligned}
\lambda^{Z}_{\mathrm{fwd}}(q)
&=-n\bigl[\nu_{\mathrm{fwd}}(q)-\nu_{\mathrm{fwd}}(0)\bigr]\\
&=-n\bigl[\nu_{\mathrm{rev}}(q-1)-\nu_{\mathrm{rev}}(-1)\bigr]\\
&=\lambda^{Z}_{\mathrm{rev}}(q-1)-\lambda^{Z}_{\mathrm{rev}}(-1),
\end{aligned}
\end{equation}
which is the first identity of Eq.~\eqref{eq:duality-final}; the second follows by replacing
$q$ with $q+1$.
\end{proof}

\begin{remark}[Scope of the separation hypothesis]\label{rem:separation-scope}
At $q=0$, Eq.~\eqref{eq:separation} gives $S_{t}/t\to-\nu(0)$, whereas
Theorem~\ref{thm:main}(1) and Corollary~\ref{cor:ldp} give
$S_{t}/t\to\varepsilon_{\infty}$. Hence
\begin{equation}\label{eq:separation-trivial}
\varepsilon_{\infty}\equiv-\nu(0),
\qquad
I=n\bigl[\sup_{E}\varepsilon_{\infty}-\varepsilon_{\infty}\bigr]\equiv0 .
\end{equation}
The duality therefore concerns tilted growth rates in a class whose endpoint rate
function is identically zero. It does not assert a non-trivial endpoint large deviation
principle in that class.
\end{remark}

The remark below identifies the obstruction to applying the standard
Hamilton--Jacobi large-deviation theorems in this setting. It is a consequence of the
axioms of Section~\ref{sec:H} and explains why no theorem requiring coercivity of $H$
in the fibre variable is cited in this paper.

\begin{remark}[Stationary entropy potential when $v\perp\ker A$]
\label{rem:coercivity-obstruction}
Assume the $N=2$ closure of Section~\ref{sec:H} and $v\perp\ker A$. Then
$\nabla H^{(0)}\in\ker A$ by Eq.~\eqref{eq:degeneracy}, so $v\cdot\nabla H^{(0)}=0$ and the
zeroth-order transport Eq.~\eqref{eq:transport-zeroth-restated} gives
$\partial_{t}H^{(0)}=0$. Hence $H^{(0)}=H^{(0)}(y)$, and Eq.~\eqref{eq:H0-construction}
gives $\nabla H^{(0)}=\gamma_{0}g$; then $A(t,\cdot)\,g=0$ for every $t$ by
Eq.~\eqref{eq:degeneracy} and $\sigma=v\cdot g=0$, so $e\equiv0$ and, by
Lemma~\ref{lem:gauge}, $\lambda^{Z}\equiv0$, consistent with the non-negative
mean of $e$ in Lemma~\ref{lem:e-mean}. In this class the potential
descends to $T^{*}E/\ker A$ and is coercive there, but the tilt is trivial.
\end{remark}

\section{Numerical Verification of Measure and Entropy Production}
\label{sec:numerical}

\subsection{Model and Its Exact Endpoint Measure}

The coupling $\xi=A\phi$ vanishes identically on the zero section, where
$S_{t}=-tH^{(0)}(t,\cdot)$ gives $\nabla S_{t}=-t\nabla H^{(0)}(t,\cdot)$ and hence
$A\nabla S_{t}=0$ by Eq.~\eqref{eq:degeneracy}: the fluctuation then acts on the base only through
the entropy production. A nonzero initial section is necessary for $\xi$ to act on the
base, and the model below has one.

Let $E=\mathbb T^{2}$ with coordinates $(x,y)$ and set
\begin{equation}\label{eq:mc3-model}
\begin{aligned}
&\rho(y)=1+\tfrac12\sin2\pi y,\qquad
v(x,y)=\Bigl(0,\frac{\kappa\sin2\pi x}{\rho(y)}\Bigr),\\
&A_0=a\,\partial_x\otimes\partial_x,\qquad
S_0(x)=-\beta\cos2\pi x,
\end{aligned}
\end{equation}
with $n=2$, $\gamma_0=\rho_0=1$, $\kappa=0.7$, $a=0.05$ and $\beta=1.5$. Then
$\nabla\!\cdot(\rho v)=0$, so $\rho$ is the invariant density of $v$, and
\begin{equation}\label{eq:mc3-sigma}
\sigma=-\nabla\!\cdot v=\frac{\kappa\sin2\pi x\,\rho'(y)}{\rho(y)^{2}},
\end{equation}
which is nonzero pointwise and changes sign. The degeneracy condition holds at the initial
time, $A_0\nabla H^{(0)}(0,\cdot)=0$, and the transported tensor is $A(t,x,y)=a\,m m^{T}$
with $m=(1,J)^{T}$,
\begin{equation}\label{eq:mc3-jac}
J(t,x,y)=\frac{2\pi\kappa\cos(2\pi x)\,t}{\rho(y)},
\end{equation}
so that $A\nabla H^{(0)}=0$ at every time. The section, however, is not aligned with
$\nabla H^{(0)}$: $\nabla S_0=(2\pi\beta\sin2\pi x,0)^{T}$ and
\begin{equation}\label{eq:mc3-xi}
\xi_0=A_0\nabla S_0=2\pi a\beta\sin(2\pi x)\,\partial_x\neq0 .
\end{equation}

The Hamiltonian of the model is
\begin{equation}\label{eq:mc3-H}
H=H^{(0)}(t,x,y)+v_y(y)\phi_y+\tfrac12 a\,(m\cdot\phi)^{2},
\qquad m=(1,J)^{T},
\end{equation}
where $H^{(0)}(t,x,y)=\ln\rho(y_0)$ and
$y_0=\Lambda^{-1}\bigl(\Lambda(y)-\kappa\sin(2\pi x)\,t\bigr)$,
$\Lambda(y)=\int_0^{y}\rho$. The contact characteristic field is
\begin{equation}\label{eq:mc3-flow}
\begin{aligned}
\dot x&=a\,(m\cdot\phi), &
\dot y&=v_y+aJ\,(m\cdot\phi),\\
\dot\phi_x&=-\partial_xH, &
\dot\phi_y&=-\partial_yH .
\end{aligned}
\end{equation}
The constraint function $\varepsilon=\tfrac12a(m\cdot\phi)^{2}-H^{(0)}$ is a first integral
of this field, so the action Eq.~\eqref{eq:main-action} is affine in time along each
characteristic,
\begin{equation}\label{eq:mc3-eps}
S_t=S_0\circ y_0+t\,\varepsilon_0\circ y_0,
\qquad
\varepsilon_0=\tfrac12 a\,(2\pi\beta\sin2\pi x_0)^{2}-\ln\rho(y_0),
\end{equation}
and the measure theorem gives the finite-time endpoint measure
\begin{equation}\label{eq:mc3-measure}
\mathcal P_t(\dd x\,\dd y)=\frac{e^{nS_t(x,y)}}{\mathcal Z_t}\,\dd x\,\dd y,
\qquad
\mathcal Z_t=\int_E e^{nS_t}\,\dd x\,\dd y,
\end{equation}
with $S_t$ evaluated at the endpoint of the characteristic issued from $(x_0,y_0)$ on the
section $\phi_0=\nabla S_0$.

The structure of $\mathcal P_t$ is made explicit by the evolution equation of its density.
Writing $p_t=e^{nS_t}/\mathcal Z_t$, the Hamilton--Jacobi equation Eq.~\eqref{eq:hj-full} gives
\begin{equation}\label{eq:explicit-hj-density}
\partial_tS_t=-H^{(0)}-v\cdot\nabla S_t-\tfrac12(\nabla S_t)^TA\nabla S_t.
\end{equation}
Differentiating
$p_t$ and using $\nabla p_t=np_t\nabla S_t$ then yields
\begin{equation}\label{eq:mc3-transport}
\partial_tp_t+u_t\cdot\nabla p_t+n\bigl(H^{(0)}-\bar H_t\bigr)p_t=0,
\qquad
u_t=\frac1n\Bigl(v+\tfrac12A\nabla S_t\Bigr),
\qquad
\bar H_t=\frac1n\,\partial_t\ln\mathcal Z_t .
\end{equation}
The density of $\mathcal P_t$ therefore solves a first-order transport equation with a
source: there is no second-order term. By contrast, the density of any base diffusion driven
by $v$ and $A$ satisfies a second-order Fokker--Planck equation. The two cannot coincide:
$\mathcal P_t$ is not the push-forward of $\mathcal P_0$ by any base dynamics, and the factor
by which it differs is precisely the action weight $e^{nt\varepsilon_0}$ carried by
Eq.~\eqref{eq:mc3-measure}.

\subsection{Stochastic-Path Evaluation}

Since $\mathcal P_t$ is not a base-dynamics push-forward, it is evaluated by importance
sampling along the base equation Eq.~\eqref{eq:stochastic-coupling} of the theory, recalled here:
\begin{equation}\label{eq:mc3-base-recall}
\dot y^{i}=v^{i}(t,y)+\xi^{i},
\qquad
\xi^{i}=A^{ij}(t,y)\,\phi_{j} .
\end{equation}
Substituting $\xi=A\phi$ into the degeneracy condition Eq.~\eqref{eq:degeneracy} and using the
symmetry $A=A^{T}$ gives the two structural identities that fix the proposal law,
\begin{equation}\label{eq:mc3-xi-structure}
\xi=A^{1/2}\bigl(A^{1/2}\phi\bigr)\in\operatorname{im}A^{1/2},
\qquad
\xi\cdot\nabla H^{(0)}=\phi^{T}\!\bigl(A\nabla H^{(0)}\bigr)=0 .
\end{equation}
The coupling is therefore a mean-zero fluctuation carried by $A^{1/2}$; realising it as a
white-in-time fluctuation with the same tensor as covariance,
\begin{equation}\label{eq:mc3-xi-white}
\xi^{i}\,\dd t=\sqrt{2\vartheta}\,\bigl(A^{1/2}\bigr)^{ij}\,\dd W^{j}_t,
\qquad
\mathbb E\bigl[\xi^{i}\xi^{j}\bigr]
=2\vartheta\,\bigl(A^{1/2}\bigr)^{ik}\bigl(A^{1/2}\bigr)^{jk}
=2\vartheta\,A^{ij},
\end{equation}
with $\vartheta>0$ the intensity, turns Eq.~\eqref{eq:mc3-base-recall} into the Langevin proposal
\begin{equation}\label{eq:mc3-langevin}
\dd y=v(t,y)\,\dd t+\xi\,\dd t
=v(t,y)\,\dd t+\sqrt{2\vartheta}\,A(t,y)^{1/2}\,\dd W_t .
\end{equation}
The covariance is taken along $A$ so that Eq.~\eqref{eq:mc3-xi-structure} holds pathwise. In the
rank-one model Eq.~\eqref{eq:mc3-model} this reads, componentwise,
\begin{equation}\label{eq:mc3-sde}
\dd x=b\,\dd W,
\qquad
\dd y=\frac{\kappa\sin2\pi x}{\rho(y)}\,\dd t+b\,J(t,x,y)\,\dd W,
\qquad
b=\sqrt{2\vartheta a},
\end{equation}
in the Stratonovich sense, under which the chain rule and $\xi\cdot\nabla H^{(0)}=0$ imply
the conservation of $H^{(0)}$ along the path. The initial states are drawn from
$\mathcal P_0\propto e^{nS_0}$, and Eq.~\eqref{eq:mc3-sde} is integrated by the Stratonovich
midpoint scheme with $\Delta t=10^{-3}$. Writing $q_t$ for the density of Eq.~\eqref{eq:mc3-sde},
which solves the second-order Fokker--Planck equation
\begin{equation}\label{eq:mc3-fp}
\partial_tq_t=-\nabla\!\cdot(vq_t)+\vartheta\,\partial_i\partial_j(A^{ij}q_t),
\end{equation}
the sample is reweighted by the ratio of the two endpoint densities,
\begin{equation}\label{eq:mc3-weight}
w=\frac{p_t(Y_t)}{q_t(Y_t)},
\qquad
\hat{\mathcal P}_t(G)
=\frac{\sum_{i=1}^{N}w_i\,\mathbf 1_{G}\bigl(Y_t^{i}\bigr)}{\sum_{i=1}^{N}w_i},
\end{equation}
the self-normalised importance-sampling estimator of $\mathcal P_t(G)$, valid for every
$\vartheta$. The exact density $p_t$ is evaluated on a $128^{2}$ periodic grid by pushing
the labels $(x_0,y_0)$ forward with Eq.~\eqref{eq:mc3-flow} and weighting by
$e^{nS_t}\lvert\det D\Phi_t\rvert$, with $S_t$ from Eq.~\eqref{eq:main-action}; no closed form of
$\mathcal P_t$ enters. The proposal density $q_t$ is estimated from an independent pilot
sample of Eq.~\eqref{eq:mc3-sde} by a periodic Gaussian kernel of width $0.8$ grid cells,
floored at $10^{-4}$ of its maximum so that the pilot cells left empty in the tails at
small intensity cannot turn the ratio into a runaway weight; the ratio is interpolated at
the sample endpoints. The sample is the merged pool of six independent runs of
$4\times10^{4}$ paths each, $N=2.4\times10^{5}$ in total, which removes the single-run
fluctuation of the reweighted marginals.

\subsection{Results and Diagnostics}

The section is transported without folding only up to
\begin{equation}\label{eq:mc3-fold}
t_c=\frac{1}{4\pi^{2}a\beta}=0.338,
\end{equation}
where $\partial x_t/\partial x_0=1-4\pi^{2}a\beta t$ vanishes. The three times $t=0.1$,
$0.2$ and $0.3$ lie inside this window, and over them the determinant of $D\Phi_t$ stays
between $0.11$ and $3.7$. The coupling is nonzero and tangent to the iso-$H^{(0)}$ surfaces:
over the ensemble at $t=0.1$, $|\xi|$ has mean $0.3145$ and $90\%$ quantile $0.4724$, while
$\max|\xi\cdot\nabla H^{(0)}|=3.3\times10^{-16}$; the constraint function
$\varepsilon=\tfrac12a(m\cdot\phi)^2-H^{(0)}$ is preserved by the numerical flow of
Eq.~\eqref{eq:mc3-flow} to $10^{-12}$.

The observable is the endpoint function
\begin{equation}\label{eq:mc3-e}
n\bar e_t(y)=\frac{n}{t}\int_0^te\bigl(s,y_s(y)\bigr)\,\dd s,
\qquad
e=\tfrac12 a\,J^{2}\Bigl(\frac{\rho'}{\rho}\Bigr)^{2}
-\frac{\kappa\sin2\pi x\,\rho'}{\rho^{2}},
\end{equation}
the path average of Eq.~\eqref{eq:entropy-prod} along the characteristic $y_s(y)$ that ends at
$y$; it is read off the exact grid as the conditional mean over each cell, so that the same
function of the endpoint is averaged under both laws.
Table~\ref{tab:mc3} reports, at $b=0.10$, the Kolmogorov--Smirnov distances of the two
reweighted marginals against the exact ones and the mean of Eq.~\eqref{eq:mc3-e}.

\begin{table}[htbp]
\centering
\small
\caption{Stochastic-path evaluation of model Eq.~\eqref{eq:mc3-model} at $N=2.4\times10^{5}$
(six merged runs of $4\times10^{4}$ paths), $\Delta t=10^{-3}$ and $b=0.10$.
$\mathrm{KS}_x$ and $\mathrm{KS}_y$ are the
Kolmogorov--Smirnov distances between the reweighted empirical marginals and the exact ones
from Eq.~\eqref{eq:mc3-measure}; the last two columns give the sampled and exact mean of
Eq.~\eqref{eq:mc3-e} ($\pm$ one standard error).}
\label{tab:mc3}
\begin{tabular}{@{}lrrrr@{}}
\toprule
$t$ & $\mathrm{KS}_x$ & $\mathrm{KS}_y$ & $\langle n\bar e_t\rangle$ (sampled) & (exact)\\
\midrule
$0.1$ & $9.1\times10^{-3}$ & $2.5\times10^{-3}$ & $0.1284\pm0.0049$ & $0.1224$\\
$0.2$ & $7.0\times10^{-3}$ & $2.8\times10^{-3}$ & $0.1986\pm0.0050$ & $0.2021$\\
$0.3$ & $8.2\times10^{-3}$ & $4.4\times10^{-3}$ & $0.2077\pm0.0059$ & $0.1960$\\
\bottomrule
\end{tabular}
\end{table}

The reweighted marginals reproduce the exact ones at every recorded time
(Figure~\ref{fig:mc3}). Over the intensity range $b\in\{0.02,0.05,0.10,0.20,0.40\}$ and the
three times, the Kolmogorov--Smirnov distances of the $x$-marginal lie between
$3.2\times10^{-3}$ and $1.04\times10^{-2}$ and those of the $y$-marginal between
$1.7\times10^{-3}$ and $4.7\times10^{-3}$; the effective sample size ratio
$\bigl(N\sum_iw_i^{2}\bigr)^{-1}$ lies between $0.36$ and $0.83$, and the mean of
Eq.~\eqref{eq:mc3-e} stays within $2.0$ standard errors of the exact value at the reported
intensity $b=0.10$, and within $3.1$ standard errors over all fifteen time--intensity pairs. The whole
cumulative distribution of Eq.~\eqref{eq:mc3-e} at $t=0.3$ coincides with its exact law, not
only its first two moments. The $x$-marginal has the single peak at $x=0.5$ imposed by the
section $S_0=-\beta\cos2\pi x$, and its height grows with $t$ from $4.5$ to $6.4$; the
$y$-marginal deepens its dip at $y\simeq0.25$ and its peak at $y\simeq0.75$ over the same
window. Both features are reproduced by the reweighted sample, and both follow from the
coupling Eq.~\eqref{eq:mc3-xi}, which is absent in the zero-section class.

The model tests the finite-time statements only. The section folds at $t_c$ of
Eq.~\eqref{eq:mc3-fold}, so no limit $t\to\infty$ is taken, and neither the endpoint large
deviation principle nor the large deviation principle for the entropy production of
Corollaries~\ref{cor:ldp} and~\ref{cor:ge} is claimed here. The complementary long-time
test, on the same model family with the section removed and the drift damped so that
the fold of Eq.~\eqref{eq:mc3-fold} is avoided, is given in
Subsection~\ref{subsec:long-time}.

\begin{figure}[t]
\centering
\includegraphics[width=0.94\textwidth]{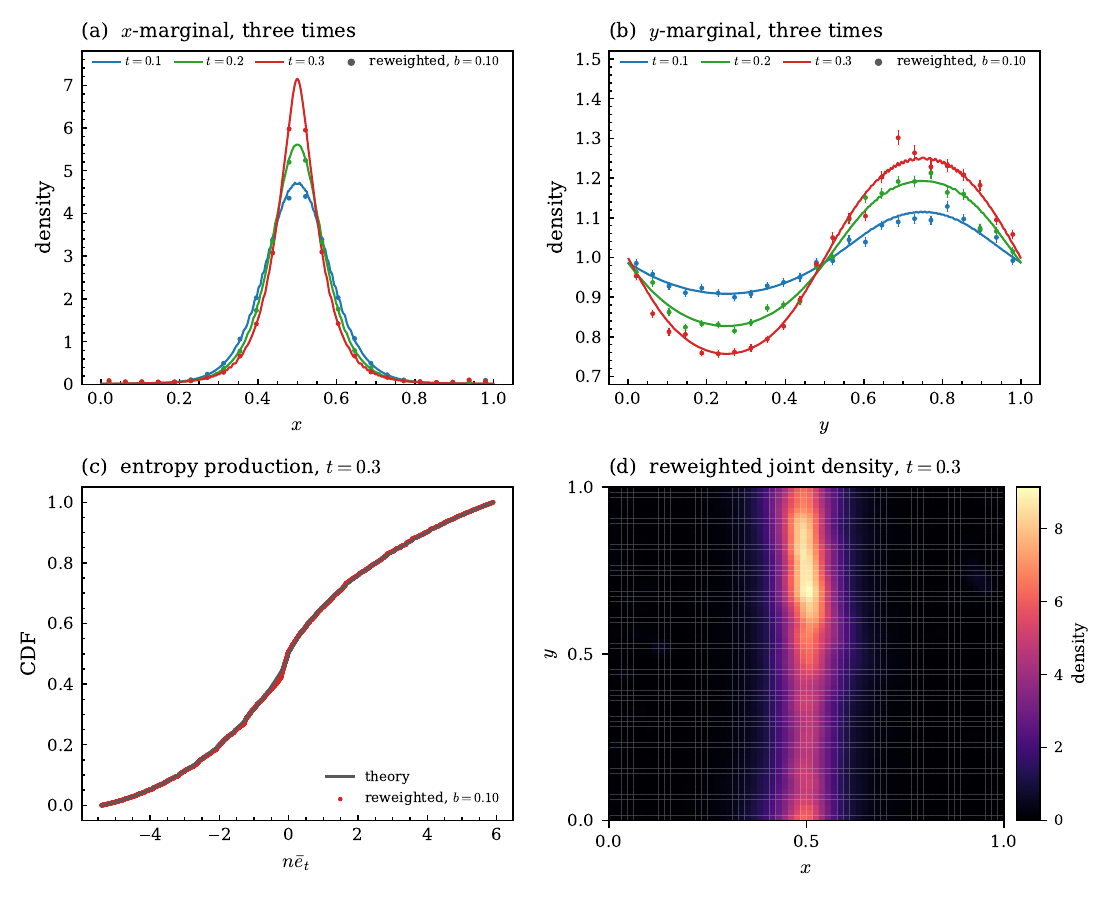}
\caption{Stochastic-path evaluation of model Eq.~\eqref{eq:mc3-model} at $b=0.10$.
(a) $x$-marginal of $\mathcal P_t$ at $t=0.1$, $0.2$ and $0.3$: exact forward quadrature
(lines) against the reweighted empirical marginal (points with one standard error).
(b) $y$-marginal over the same times and with the same convention; colour encodes $t$.
(c) Cumulative distribution of Eq.~\eqref{eq:mc3-e} at $t=0.3$, exact against the reweighted
empirical law. (d) Joint density of the reweighted sample on the $(x,y)$ plane at
$t=0.3$; the marginals and the cumulative distribution above are its projections and its
path-average law.}
\label{fig:mc3}
\end{figure}

The law of the endpoint entropy production Eq.~\eqref{eq:mc3-e} is compared with its exact
quadrature in Figure~\ref{fig:ep-law}. Over $t=0.05,0.10,\dots,0.30$ the reweighted mean of
Eq.~\eqref{eq:mc3-e} stays within $3.0$ standard errors of the exact mean, the largest
deviation occurring at $t=0.15$, and the reweighted standard deviation differs from the
exact one by less than $0.5\%$ at every recorded time.

\begin{figure}[t]
\centering
\includegraphics[width=0.94\textwidth]{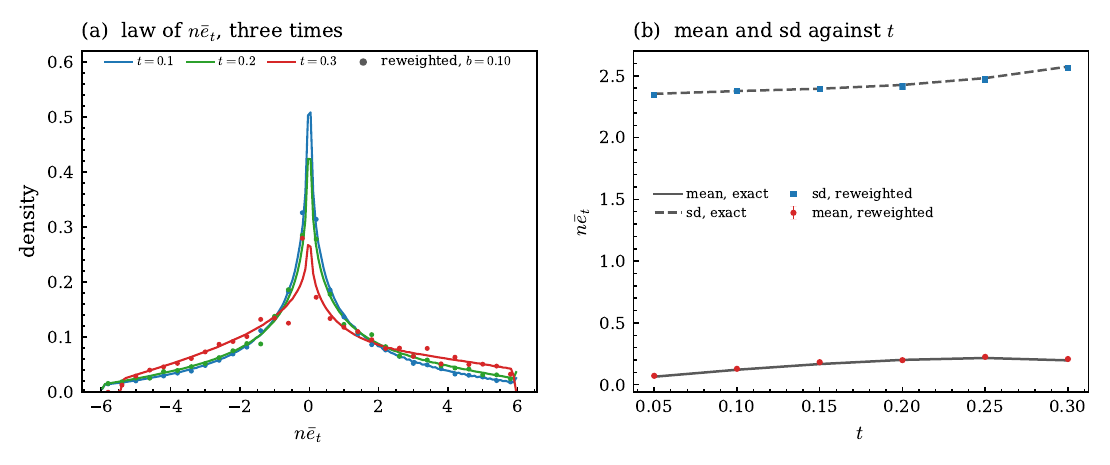}
\caption{Law of the endpoint entropy production Eq.~\eqref{eq:mc3-e} at $b=0.10$.
(a) Density of $n\bar e_t$ at $t=0.1$, $0.2$ and $0.3$: exact forward quadrature (lines)
against the reweighted sample (points with one standard error).
(b) Mean and standard deviation of $n\bar e_t$ against $t$: exact quadrature over
$t=0.05,0.10,\dots,0.30$ (lines) against the reweighted sample (points, one standard
error on the mean).}
\label{fig:ep-law}
\end{figure}

\subsection{Long-Time Verification of the Rate Functions}
\label{subsec:long-time}

The model of Eq.~\eqref{eq:mc3-model} tests the finite-time statements of
Theorem~\ref{thm:main}: its section $S_{0}=-\beta\cos2\pi x$ folds at
$t_{c}$ of Eq.~\eqref{eq:mc3-fold}, so no limit $t\to\infty$ is taken there and
Corollaries~\ref{cor:ldp} and~\ref{cor:ge} are not exercised. The same model
family, with the section removed and the drift damped, supplies the
complementary long-time test. The two examples share the invariant density
$\rho$, the rank-one tensor $A_{0}$, the observable Eq.~\eqref{eq:mc3-e}, the
self-normalised estimator Eq.~\eqref{eq:mc3-weight} and the diagnostics of
Table~\ref{tab:mc3}; they differ only in the initial section and in the time
dependence of the drift. Where the finite-time example probes the Gibbs form
of Theorem~\ref{thm:main} before the fold, the damped example probes the rate
functions of Corollaries~\ref{cor:ldp} and~\ref{cor:ge} as $t\to\infty$.

\subsubsection{Damped zero-section model}

On $E=\mathbb T^{2}$ keep $\rho$ and $A_{0}$ of Eq.~\eqref{eq:mc3-model}, set
$S_{0}=0$, and multiply the drift by $e^{-t}$:
\begin{equation}\label{eq:lt-model}
\begin{gathered}
\rho(y)=1+\tfrac12\sin(2\pi y),\qquad S_{0}=0,\qquad
A_{0}=a\,\partial_x\otimes\partial_x,\\
v(t,x,y)=e^{-t}\Bigl(0,\frac{\kappa\sin(2\pi x)}{\rho(y)}\Bigr),
\end{gathered}
\end{equation}
with $a=0.05$, $\kappa=0.05$ (against $\kappa=0.7$ in
Eq.~\eqref{eq:mc3-model}), $n=2$ and $\gamma_{0}=\rho_{0}=1$. With
$\Lambda(y)=\int_{0}^{y}\rho(z)\,\dd z$ and $F(t)=1-e^{-t}$, the zero-section
characteristics of Eq.~\eqref{eq:mc3-flow} integrate to
\begin{equation}\label{eq:lt-endpoint}
\Lambda(y_t)=\Lambda(y_0)+\kappa F(t)\sin(2\pi x_0),
\qquad x_t=x_0,
\end{equation}
so the endpoint map converges in $C^{1}$ to the shear diffeomorphism at
$F=1$; since $\kappa F\le0.05$ the section never folds and $t_c=\infty$, in
contrast with Eq.~\eqref{eq:mc3-fold}. The transported tensor remains
$A=a\,m\,m^{T}$, $m=(1,J)$, with Eq.~\eqref{eq:mc3-jac} evaluated at $F(t)$
instead of $t$:
\begin{equation}\label{eq:lt-jac}
J(t,x,y)=\frac{\kappa(2\pi)\cos(2\pi x)}{\rho(y)}F(t).
\end{equation}
The observable is the same path average as Eq.~\eqref{eq:mc3-e}, with the
damped Jacobian Eq.~\eqref{eq:lt-jac} and $\sigma_t=e^{-t}\kappa\sin(2\pi x)\,
\rho'(y)/\rho(y)^{2}$:
\begin{equation}\label{eq:lt-obs}
\begin{gathered}
e(t,x,y)=\frac a2\,J(t,x,y)^{2}\Bigl(\frac{\rho'(y)}{\rho(y)}\Bigr)^{2}
-e^{-t}\,\frac{\kappa\sin(2\pi x)\,\rho'(y)}{\rho(y)^{2}},\\
Z_t=n\bar e_t=\frac nt\int_0^te(s,\cdot)\,\dd s .
\end{gathered}
\end{equation}
The endpoint rate function of Corollary~\ref{cor:ldp} and the limiting
observable of Corollary~\ref{cor:ge} are
\begin{equation}\label{eq:lt-rate}
I(y)=2\ln\!\bigl(2\rho(y_{0\infty}(y))\bigr),
\qquad
e_\infty(x,y)=\frac a2\,
\frac{\kappa^{2}(2\pi)^{2}\cos^{2}(2\pi x)\,\rho'(y)^{2}}{\rho(y)^{4}},
\end{equation}
and, since the recorded observable is $Z_t=n\bar e_t$, its contraction rate
carries the explicit factor $n$:
\begin{equation}\label{eq:lt-JZ}
J_Z(z)=\inf\{I(y):n\,e_\infty(y)=z\}=J(z/n).
\end{equation}

\subsubsection{Independent solvers and shared diagnostics}

Three independent routes are used, none of which inserts a closed-form rate
function into the output. The first solves the transport equations for
$H^{(0)}$ and $A$ together with the Hamilton--Jacobi equation
Eq.~\eqref{eq:hj-full} on periodic grids of $128^{2}$, $256^{2}$ and $512^{2}$
points, with upwind advection, a small local Lax--Friedrichs viscosity and
SSP-RK2 stepping; Table~\ref{tab:long-pde} reports the convergence of this
solver against the exact finite-time solution at $t=10$. The second integrates
the zero-section characteristic ODE, the Lie-transport equation, the action $Q$
and the entropy-production integral of Eq.~\eqref{eq:lt-obs} by RK4 from
sampled labels; the contact measure is then built from the numerically
integrated action together with the Jacobian
$|\det D\Phi_t|=\rho(y_0)/\rho(y_t)$ of the endpoint map, which is exactly the
estimator Eq.~\eqref{eq:mc3-weight} with the proposal density of the sampled
labels. The labels are drawn from the tilted law
$g_t(y_0)\propto\rho(y_0)^{-nt+1}$, uniform in $x_0$, for which the weight
reduces to the bounded factor $1/\rho(y_t)$; the uniform-label proposal of the
finite-time example is retained only as the efficiency diagnostic of
Figure~\ref{fig:long-ldp}(c). All samples are diagnosed by the same
Kolmogorov--Smirnov and effective-sample-size statistics as
Table~\ref{tab:mc3}. The third is an independent Gauss--Legendre quadrature
over time and labels, which plays the r\^ole of the exact forward quadrature
of Figures~\ref{fig:mc3} and~\ref{fig:ep-law}. Halving the time step changes
the reported Monte Carlo means by less than $5\times10^{-13}$.

\begin{table}[htbp]
\centering
\small
\caption{PDE errors at $t=10$ against the exact finite-time solution, for the
solver described in the text.}
\label{tab:long-pde}
\begin{tabular}{@{}cccc@{}}
\toprule
Grid & $\|S_{10}-S^{\mathrm{exact}}\|_\infty$ &
$\|I_{10}/10-I\|_\infty$ & $\|A_{22}-A_{22}^{\mathrm{exact}}\|_\infty$ \\
\midrule
$128^2$ & $0.24447$ & $0.04183$ & $1.10\times10^{-4}$\\
$256^2$ & $0.08222$ & $0.01519$ & $5.56\times10^{-5}$\\
$512^2$ & $0.04146$ & $0.00794$ & $2.80\times10^{-5}$\\
\bottomrule
\end{tabular}
\end{table}

\subsubsection{Endpoint measure and law of the observable}

Figure~\ref{fig:long-measure} repeats, for the damped model, the four views of
Figure~\ref{fig:mc3}: the two marginals of $\mathcal P_t$ at $t=80$, $160$ and
$320$, the cumulative distribution of $Z_t$ at $t=320$, and the joint density
of the sample. Theory lines are the exact Gibbs density of
Corollary~\ref{cor:ldp}; points are the reweighted sample with one standard
error per bin, and the cumulative panel carries a $95\%$ band from 96
weighted-bootstrap resamples. At $t=10$ the marginal Kolmogorov--Smirnov
distances are $0.0020$ and $0.0039$, and the cumulative distribution of $Z_t$
differs from the independent quadrature by $0.0034$; at $t=320$ they are
$0.0027$, $0.0055$ and $0.0084$. The uniform-label proposal of the finite-time
example concentrates at large $t$ (Figure~\ref{fig:long-ldp}(c)), which is why
the sample here uses the tilted proposal; under it the shallow $\rho$
modulation of the $x$-marginal is resolved bin by bin at every time.
The $y$-marginal sharpens into the two branches
of the limiting shear image as $t$ grows, while the $x$-marginal keeps the
shallow modulation imposed by $\rho$.

Figure~\ref{fig:long-law} is the counterpart of Figure~\ref{fig:ep-law}: the
density of $Z_t$ at the same three times, and its mean and standard deviation
against $t$. At $t=80$ the Monte Carlo and quadrature means are $0.0229084$
and $0.0228697$; at $t=320$ they are $0.0244442$ and $0.0243961$, against the
exact limiting mean $0.0249183$ of Eq.~\eqref{eq:lt-rate}. The approach is a
visible Ces\`aro transient of the path average, not an early coincidence with
the limit, and the standard deviation collapses onto the same transient.

\begin{figure}[t]
\centering
\includegraphics[width=0.94\textwidth]{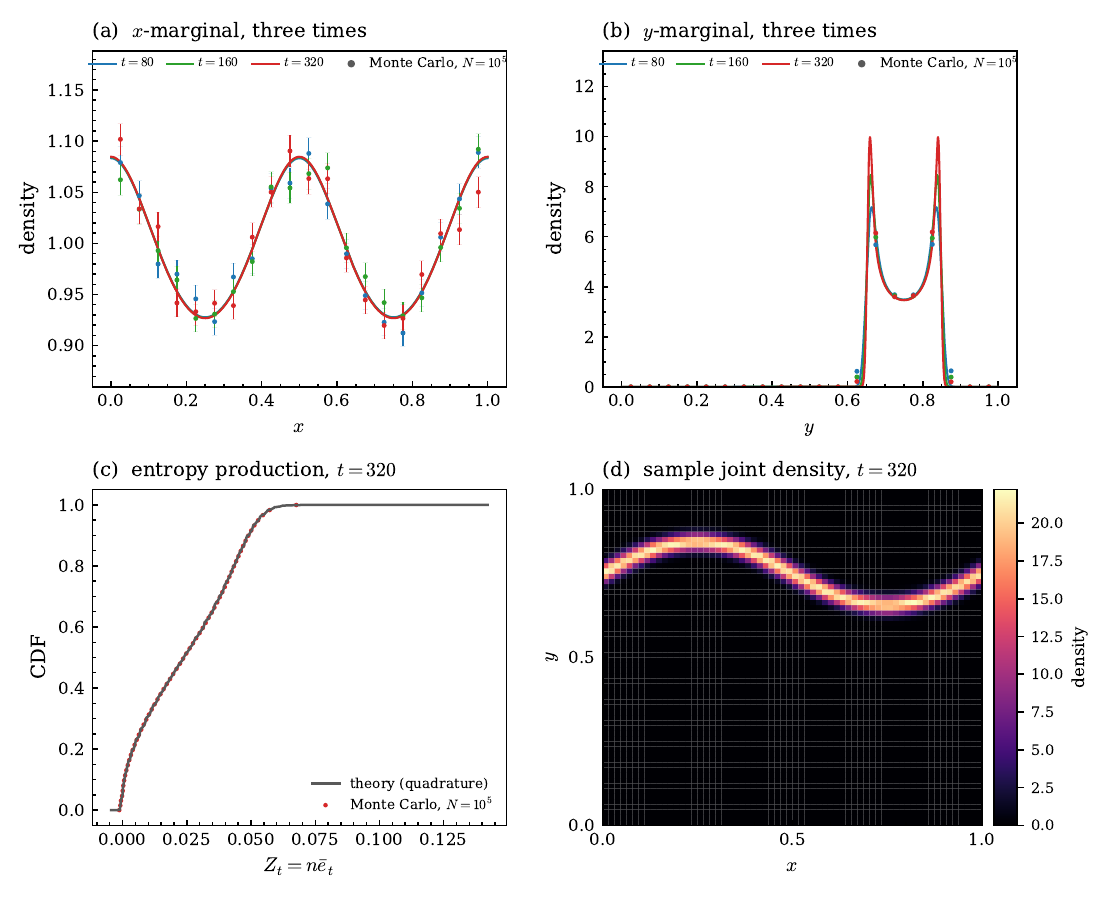}
\caption{Endpoint measure of the damped zero-section model
Eq.~\eqref{eq:lt-model}, in the convention of Figure~\ref{fig:mc3}; the sample
uses the tilted proposal described in the text.
(a) $x$-marginal of $\mathcal P_t$ at $t=80$, $160$ and $320$: exact Gibbs
density (lines) against the reweighted sample (points, one standard error).
(b) $y$-marginal over the same times, same convention. (c) Cumulative
distribution of $Z_t$ of Eq.~\eqref{eq:lt-obs} at $t=320$: independent
quadrature (line) against the reweighted sample (points), with a $95\%$
weighted-bootstrap band. (d) Joint density of the reweighted sample on the
$(x,y)$ plane at $t=320$; the marginals and the cumulative distribution above
are its projections and its path-average law.}
\label{fig:long-measure}
\end{figure}

\begin{figure}[t]
\centering
\includegraphics[width=0.94\textwidth]{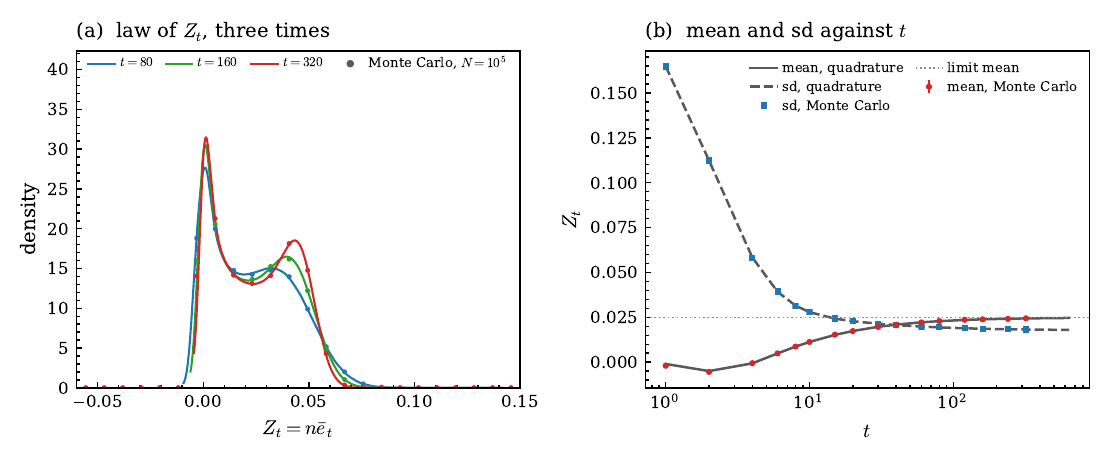}
\caption{Law of the endpoint entropy production $Z_t=n\bar e_t$ of
Eq.~\eqref{eq:lt-obs}, in the convention of Figure~\ref{fig:ep-law}.
(a) Density of $Z_t$ at $t=80$, $160$ and $320$: independent quadrature
(lines) against the reweighted sample (points, one standard error).
(b) Mean and standard deviation of $Z_t$ against $t$: quadrature over
$t=1,2,\dots,640$ (lines) against the reweighted sample (points, one standard
error on the mean), with the limiting mean of Eq.~\eqref{eq:lt-rate} dotted.}
\label{fig:long-law}
\end{figure}

\subsubsection{Tail exponents and convergence diagnostics}

For each interval $G$ the exponent $-t^{-1}\log P_t(Z_t\in G)$ is compared
with $\inf_GJ_Z$ of Eq.~\eqref{eq:lt-JZ}; Table~\ref{tab:long-tail} reports
$t=320$ and Figure~\ref{fig:long-ldp}(a) the whole time sequence, with the
correct lower bounds, three of which vanish because $G$ meets the zero-cost
set $I=0$ shown in panel (d). The mean error scales as $0.1665/t$ while the
uniform path-average error $\|n(\bar e_t-e_\infty)\|_\infty$ scales as
$0.5182/t$ (panel (b)); the slower of the two explains why the tails enter
the asymptotic regime later than the mean. The independent quadrature at
$t=640$ gives exponents $0.00432$, $0.00370$, $0.00438$ and $0.01068$, closer
still to the bounds. Panel (c) records the effective sample-size ratio of the
estimator Eq.~\eqref{eq:mc3-weight} under both proposals: the tilted one stays
at $0.996$--$0.997$ over the whole range, while the uniform proposal of the
finite-time example decays to $0.0436$ at $t=80$ and $0.0219$ at $t=320$,
which is why a uniform sample cannot resolve the marginals of
Figure~\ref{fig:long-measure} at large $t$.

\begin{table}[htbp]
\centering
\small
\caption{Tail exponents of $Z_t$ at $t=320$: reweighted sample and independent
quadrature against the lower bound $\inf_GJ_Z$ of Eq.~\eqref{eq:lt-JZ}.}
\label{tab:long-tail}
\begin{tabular}{@{}cccc@{}}
\toprule
$G$ & Monte Carlo & Quadrature & $\inf_GJ_Z$\\
\midrule
$[0.03,0.035)$ & $0.00823$ & $0.00832$ & $0$\\
$[0.04,0.045)$ & $0.00736$ & $0.00755$ & $0$\\
$[0.05,0.055)$ & $0.00910$ & $0.00901$ & $0$\\
$[0.06,0.065)$ & $0.01582$ & $0.01608$ & $0.00360$\\
\bottomrule
\end{tabular}
\end{table}

\begin{figure}[t]
\centering
\includegraphics[width=0.94\textwidth]{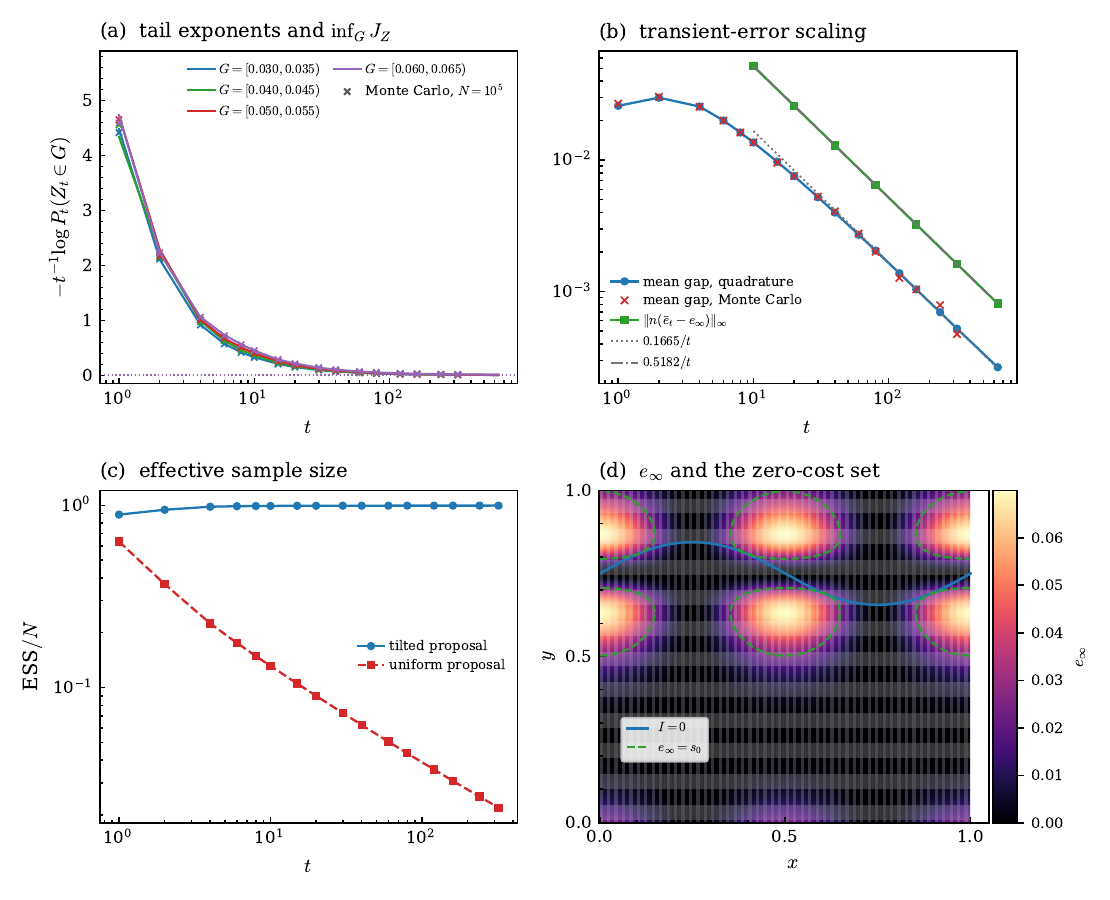}
\caption{Long-time large deviation diagnostics of the damped model
Eq.~\eqref{eq:lt-model}. (a) Tail exponents $-t^{-1}\log P_t(Z_t\in G)$ for
four intervals: independent quadrature (lines) against the reweighted sample
(crosses), with the lower bounds $\inf_GJ_Z$ dotted. (b) Mean gap
$|\mathrm{E}Z_t-\mathrm{E}Z_\infty|$ for both routes and uniform path-average
error, with the reference slopes $0.1665/t$ and $0.5182/t$. (c) Effective
sample-size ratio of the estimator Eq.~\eqref{eq:mc3-weight} under the tilted
and the uniform proposals, $N=10^{5}$. (d) Limiting observable $e_\infty$ of
Eq.~\eqref{eq:lt-rate}, the zero-cost set $I=0$ and the level set
$e_\infty=s_0$ that fixes the vanishing bounds in panel (a).}
\label{fig:long-ldp}
\end{figure}

Taken together, the example of Eq.~\eqref{eq:mc3-model} and the damped example
of Eq.~\eqref{eq:lt-model} cover the two regimes of the theory with one
estimator and one set of diagnostics: the former verifies the finite-time
Gibbs form and the transported tensor before the fold, the latter verifies
the endpoint and entropy-production rate functions as $t\to\infty$. The
latter supports Corollaries~\ref{cor:ldp} and~\ref{cor:ge} as a long-time
numerical trend; it is not a numerical proof, and in particular it does not
verify the fluctuation duality of Subsection~\ref{subsec:duality}. Finally,
the damped model uses the prescribed invariant density $\rho$ of
Eq.~\eqref{eq:mc3-model}; its drift also admits other invariant densities,
and no unique-ergodicity assumption is claimed.

\section{Conclusions}\label{sec:conclusion}

The construction rests on a single statement, the main theorem
(Theorem~\ref{thm:main}): the probability measure carried by the contact form is, at every
finite time, the exponential of a function affine in the elapsed time, whose slope is $n$
times a first integral of the characteristic field. Nothing asymptotic enters that
statement. The large deviation principle is instead a corollary, obtained by letting the
time grow under the endpoint-convergence condition.

The input is two external theorems, the least-constraint theorem of
Section~\ref{sec:found} and the contact probability measure of
Theorem~\ref{thm:contact-measure}, together with the second-order contact potential
$H=H^{(0)}+v\cdot\phi+\tfrac12\phi^{T}\!A\phi$ of Section~\ref{sec:H}. From these data
the probability measure and its large deviations are obtained in three steps.

The first step is finite-time and unconditional.
Lemma~\ref{lem:action-linear} reduces the contact action to the constraint function,
$S_{t}=S_{0}\circ y_{0}(t,\cdot)+t\varepsilon_{t}$ with
$\varepsilon_{t}=\varepsilon_{0}\circ y_{0}(t,\cdot)$, Lemma~\ref{lem:gibbs} writes the
contact probability measure in Gibbs form with exponent $nS_{t}$, and
Lemma~\ref{lem:rate-squeeze} gives the finite-time rate function as the excess of
$\varepsilon_{0}$ over its supremum, up to a term bounded uniformly in $t$.

The second step passes to $t\to\infty$. The amplitude of $\varepsilon_{t}$ is conserved,
and no separated asymptotics is possible unless $\varepsilon_{0}$ is constant
(Lemma~\ref{lem:oscillation}), so the non-degeneracy of the rate function is
decided by the initial section alone. When the endpoint map converges in $C^{1}$ to a
diffeomorphism, Corollary~\ref{cor:ldp} gives the large deviation principle for the
endpoint probability measure with rate
$I=n[\sup_{E}\varepsilon_{\infty}-\varepsilon_{\infty}]$, and Corollary~\ref{cor:ge}
transfers it to the entropy production by contraction, with rate
$J(s)=\inf\{I(y):\bar e_{\infty}(y)=s\}$. The scaled cumulant generating function is the
Legendre transform of $J$, by Varadhan's lemma.

The third step is the tilted theory. The tilt does not preserve the constraint function
(Lemma~\ref{lem:tilt-obstruction}), so the tilted theory is formulated on the cell
equation rather than on the action; the involution $W\mapsto-W-\ln\rho$ preserves the
constant of that equation (Lemma~\ref{lem:involution}), and under the separation
condition Eq.~\eqref{eq:separation} the tilted growth rates satisfy the duality relation
Eq.~\eqref{eq:duality-final} (Corollary~\ref{cor:duality}).

Two statements delimit the result. Corollary~\ref{cor:dichotomy} gives the
non-degeneracy criterion for the limit rate function, and
Remark~\ref{rem:coercivity-obstruction} records that $v\perp\ker A$ makes the entropy
potential stationary and forces $e\equiv0$; no coercive Hamilton--Jacobi theorem is used in
the construction.

\section*{Acknowledgments}
The author thanks the Department of Hydraulic Engineering at Tsinghua University for providing the research environment that made this work possible.

\section*{Funding}
This research received no specific grant from any funding agency in the public, commercial, or not-for-profit sectors.

\section*{Conflict of Interest}
The author declares no conflict of interest.

\section*{Data Availability}
This is a theoretical study involving no experimental data. All mathematical derivations are contained within the manuscript.

\appendix

\end{document}